\documentclass[format=acmsmall, review=false, screen=true]{acmart}

\usepackage{booktabs} 

\usepackage{graphicx}
\usepackage{balance}  
\usepackage{times}
\usepackage{amsmath,amsthm,epsfig,amsbsy}
\usepackage[noabbrev,capitalize]{cleveref}

\usepackage[ruled,vlined,linesnumbered,commentsnumbered]{algorithm2e}
\usepackage{booktabs}
\usepackage[framemethod=TikZ]{mdframed}
\usepackage{multicol}
\usepackage{multirow}
\usepackage{epsfig}
\usepackage{epstopdf}
\usepackage{url}
\usepackage{float}
\usepackage{subcaption}
\usepackage{makecell}
\usepackage[labelfont=bf,skip=7pt]{caption}
\usepackage{changepage}
\usepackage{soul}
\usepackage{wrapfig}
\usepackage{ragged2e}
\usepackage{multirow}
\usepackage{xargs}
\usepackage{pifont}
\usepackage[colorinlistoftodos,prependcaption,textsize=tiny]{todonotes}

\usepackage{sistyle}
\SIthousandsep{,}

\newtheorem{lemma}{Lemma}
\newtheorem{corollary}{Corollary}

\newcommand{\mycomment}[1]{{\color{red}#1}}
\newcommand{\modified}[1]{{\color{black}#1}}

	\definecolor{darkpink}{rgb}{0.91, 0.33, 0.5}
\newcommand{\vahid}[1]{{\color{black}#1}}
\newcommand{\remove}[1]{}

\newcommand{\mce}{{\tt MCE}}
\newcommand{\cliques}{\mathcal{C}}
\newcommand{\tomita}{{\tt TTT}\xspace}
\newcommand{\partomita}{{\tt ParTTT}\xspace}
\newcommand{\parmce}{{\tt ParMCE}\xspace}
\newcommand{\parmcedegree}{{\tt ParMCEDegree}\xspace}
\newcommand{\parmcedegen}{{\tt ParMCEDegen}\xspace}
\newcommand{\parmcetriangle}{{\tt ParMCETri}\xspace}
\newcommand{\parpivot}{{\tt ParPivot}\xspace}
\newcommand{\gp}{{\tt GP}\xspace}
\newcommand{\cand}{{\tt cand}\xspace}
\newcommand{\fini}{{\tt fini}\xspace}
\newcommand{\pivot}{{\tt pivot}\xspace}
\newcommand{\ext}{{\tt ext}\xspace}
\newcommand{\rank}{{\tt rank}\xspace}
\newcommand{\intersect}{{\tt intersect}\xspace}

\newcommand{\greedybb}{{\tt GreedyBB}\xspace}
\newcommand{\hashing}{{\tt Hashing}\xspace}
\newcommand{\bkd}{{\tt BKDegeneracy}\xspace}
\newcommand{\cliqueenum}{{\tt Clique Enumerator}\xspace}
\newcommand{\peamc}{{\tt Peamc}\xspace}
\newcommand{\peco}{{\tt PECO}\xspace}
\newcommand{\pecodegree}{{\tt PECODegree}\xspace}
\newcommand{\pecodegen}{{\tt PECODegen}\xspace}
\newcommand{\pecotri}{{\tt PECOTri}\xspace}

\newcommand{\imce}{{\tt IMCE}\xspace}
\newcommand{\parimce}{{\tt ParIMCE}\xspace}
\newcommand{\csnewttt}{{\tt FastIMCENewClq}\xspace}
\newcommand{\parcsnewttt}{{\tt ParIMCENew}\xspace}
\newcommand{\cssub}{{\tt IMCESubClq}\xspace}
\newcommand{\parcssub}{{\tt ParIMCESub}\xspace}
\newcommand{\parcssubnew}{{\tt ParSubClq}\xspace}

\newcommand{\partomitaE}{{\tt ParTTTExcludeEdges}\xspace}
\newcommand{\tomitaE}{{\tt {TTTExcludeEdges}}\xspace}
\newcommand{\wtalk}{{\small\texttt{\textbf{Wiki-Talk}}}\xspace}
\newcommand{\skitter}{{\small\texttt{\textbf{As-Skitter}}}\xspace}
\newcommand{\dblp}{{\small\texttt{\textbf{DBLP-Coauthor}}}\xspace}
\newcommand{\wikipedia}{{\small\texttt{\textbf{Wikipedia}}}\xspace}
\newcommand{\orkut}{{\small\texttt{\textbf{Orkut}}}\xspace}
\newcommand{\flickr}{{\small\texttt{\textbf{Flickr}}}\xspace}
\newcommand{\cacit}{{\small\texttt{\textbf{Ca-Cit-HepTh}}}\xspace}
\newcommand{\journal}{\small{\texttt{\textbf{LiveJournal}}}\xspace}

\newcommandx{\unsure}[2][1=]{\todo[linecolor=red,backgroundcolor=red!25,bordercolor=red,#1]{#2}}
\newcommandx{\change}[2][1=]{\todo[linecolor=blue,backgroundcolor=blue!25,bordercolor=blue,#1]{#2}}
\newcommandx{\info}[2][1=]{\todo[linecolor=OliveGreen,backgroundcolor=OliveGreen!25,bordercolor=OliveGreen,#1]{#2}}
\newcommandx{\improvement}[2][1=]{\todo[linecolor=Plum,backgroundcolor=Plum!25,bordercolor=Plum,#1]{#2}}
\newcommandx{\thiswillnotshow}[2][1=]{\todo[disable,#1]{#2}}

\mdfdefinestyle{MyFrame}{%
    linecolor=black,
    outerlinewidth=2pt,
    innertopmargin=4pt,
    innerbottommargin=4pt,
    innerrightmargin=4pt,
    innerleftmargin=14pt,
        leftmargin = 4pt,
        rightmargin = 4pt,
	roundcorner=12pt, 
	linecolor=black, 
	outerlinewidth=1pt,
	backgroundcolor=gray!10
        }

\newcommand\Tstrut{\rule{0pt}{3.2ex}}        
 
\SetAlFnt{\small}
\SetAlCapFnt{\small}
\SetAlCapNameFnt{\small}
\SetAlCapHSkip{0pt}
\IncMargin{-\parindent}

\begin{document}
	\sloppy
	\title{Shared-Memory Parallel Maximal Clique Enumeration from Static and Dynamic Graphs}
	
	\titlenote{This paper is accepted in ACM Transactions on Parallel Computing (TOPC). A preliminary version of this work appeared in the proceedings of the 25th IEEE International Conference on. High Performance Computing, Data, and Analytics (HiPC), 2018 \cite{DST-HIPC-18}}
	
	\author{Apurba Das}
	\affiliation{%
		\institution{Iowa State University}}
	\email{adas@iastate.edu}
	\author{Seyed-Vahid Sanei-Mehri}
	\affiliation{%
		\institution{Iowa State University}}
	\email{vas@iastate.edu}
	\author{Srikanta Tirthapura}
	\affiliation{%
		\institution{Iowa State University}}
	\email{snt@iastate.edu}


\begin{abstract}
Maximal Clique Enumeration (MCE) is a fundamental graph mining problem, and is useful as a primitive in identifying dense structures in a graph. Due to the high computational cost of MCE, parallel methods are imperative for dealing with large graphs. We present shared-memory parallel algorithms for MCE, with the following properties: (1)~the parallel algorithms are provably work-efficient relative to a state-of-the-art sequential algorithm (2)~the algorithms have a provably small parallel depth, showing they can scale to a large number of processors, and (3)~our implementations on a multicore machine show good speedup and scaling behavior with increasing number of cores, and are substantially faster than prior shared-memory parallel algorithms for MCE; for instance, on certain input graphs, while prior works either ran out of memory or did not complete in 5 hours, our implementation finished within a minute using 32 cores. We also present work-efficient parallel algorithms for maintaining the set of all maximal cliques in a dynamic graph that is changing through the addition of edges.
\end{abstract}
	
	\maketitle
	
\section{Introduction}
\label{sec:intro}
We study Maximal Clique Enumeration (MCE) from a graph, which requires to enumerate all cliques (complete subgraphs) in the graph that are maximal. A clique $C$ in a graph $G = (V,E)$ is a (dense) subgraph such that every pair of vertices in $C$ are directly connected by an edge. A clique $C$ is said to be maximal when there is no clique $C'$ such that $C$ is a proper subgraph of $C'$ (see \cref{fig:maximality-example}). Maximal cliques are perhaps the most fundamental dense subgraphs, and MCE has been widely used in diverse research areas, e.g. the works of Palla et al.~\cite{PD+05} on discovering protein groups by mining cliques in a protein-protein network, of Koichi et al.~\cite{KA+14} on discovering chemical structures using MCE on a graph derived from large-scale chemical databases, various problems on mining biological data~\cite{GA+93, HO+03, MBR04, CC05, JB06, RWY07, ZP+08}, overlapping community detection~\cite{YQ+-TKDE-17}, location-based services on spatial databases~\cite{ZZ+-ICDE-19}, and inference in graphical models~\cite{KF09}.
MCE is important in static as well as dynamic networks, i.e. networks that change over time due to the addition and deletion of vertices and edges. For example, Chateau et al.~\cite{CRR11} use MCE in studying changes in the genetic rearrangements of bacterial genomes, when new genomes are added to the existing genome population. Maximal cliques are building blocks for computing approximate common intervals of multiple genomes.
Duan et al.~\cite{DL+12} use MCE on a dynamic graph to track highly interconnected communities in a dynamic social network. 

MCE is a computationally hard problem since it is harder than the {\em maximum} clique problem, a classical NP-complete combinatorial problem that asks to find a clique of the largest size in a graph. Note that maximal clique and maximum clique are two related, but distinct notions. A maximum clique is also a maximal clique, but a maximal clique need not be a maximum clique. The computational cost of enumerating maximal cliques can be higher than the cost of finding the maximum clique, since the output size (set of all maximal cliques) can itself be very large. Moon and Moser~\cite{MM65} showed that a graph on $n$ vertices can have as many as $3^{n/3}$ maximal cliques, which is proved to be a tight bound. Real-world networks typically do not have cliques of such high complexity and as a result, it is feasible to enumerate maximal cliques from large graphs. The literature is rich on sequential algorithms for MCE. Bron and Kerbosch~\cite{BK73} introduced a backtracking search method to enumerate maximal cliques. Tomita et. al \cite{TTT06} introduced the idea of ``pivoting'' in the backtracking search, which led to a significant improvement in the runtime. This has been followed up by further work such as due to Eppstein et al.~\cite{ELS10}, who used a degeneracy-based vertex ordering on top of the pivot selection strategy.

Sequential approaches to MCE can lead to high runtimes on large graphs. Based on our experiments, a real-world network \orkut with approximately $3$ million vertices, $117$ million edges requires approximately $8$ hours to enumerate all maximal cliques using an efficient sequential algorithm due to Tomita et al.~\cite{TTT06}. Graphs that are larger and/or more complex cannot be handled by sequential algorithms with a reasonable turnaround time, and the high computational complexity of MCE calls for parallel methods.

\begin{figure} [t]
	\centering
	\begin{tabular}{cccc}
		\includegraphics[width=.3\textwidth]{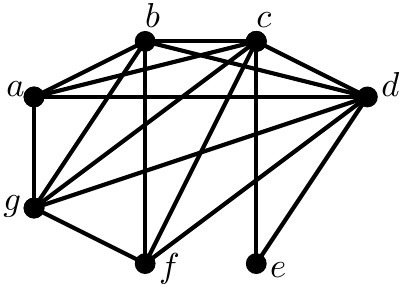} &
		\includegraphics[width=0.3\textwidth]{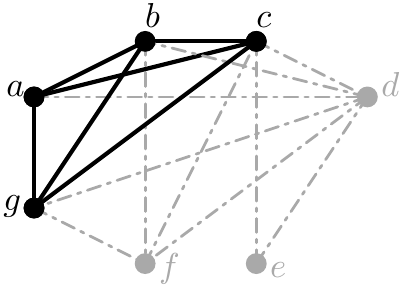} &
		\includegraphics[width=0.3\textwidth]{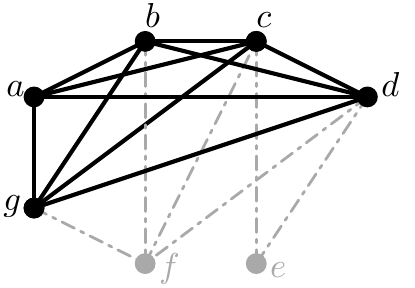} \\
		\textbf{(a)} Input Graph $G$  & \textbf{(b)} Non-Maximal Clique in $G$ & \textbf{(c)} Maximal Clique in $G$ 
	\end{tabular}
	\caption{Maximal clique in a graph} 
	\label{fig:maximality-example}
\end{figure}

In this work, we consider shared-memory parallel methods for MCE. In the shared-memory model, the input graph can reside within globally shared memory, and multiple threads can work in parallel on enumerating maximal cliques. Shared-memory parallelism is of high interest today since machines with tens to hundreds of cores and hundreds of Gigabytes of shared-memory are readily available. The advantage of using shared-memory approach over a distributed memory approach are: (1)~Unlike distributed memory, it is not necessary to divide the graph into subgraphs and communicate the subgraphs among processors. In shared-memory, different threads can work concurrently on a single shared copy of the graph (2)~Sub-problems generated during MCE are often highly imbalanced, and it is hard to predict which sub-problems are small and which are large, while initially dividing the problem into sub-problems. With a shared-memory method, it is possible to further subdivide sub-problems and process them in parallel. With a distributed memory method, handling such imbalances in sub-problem sizes requires greater coordination and is more complex. 

To show how imbalanced the sub-problems can be, in \cref{fig:irregular}, we show data for two real-world networks \skitter and \wtalk. These two networks have millions of edges and tens of millions of maximal cliques (for statistics on these networks, see \cref{sec:exp}). Consider a division of the MCE problem into per-vertex sub-problems, where each sub-problem corresponds to the set of all maximal cliques containing a single vertex in a network and suppose these sub-problems were solved independently, while taking care to prune out search for the same maximal clique multiple times. For \skitter, we observed that $90\%$ of total runtime required for MCE is taken by only $0.022\%$ of the sub-problems and less than $0.4\%$ of all sub-problems yield $90\%$ of all maximal cliques. Even larger skew in sub-problem sizes is observed in the \wtalk graph. This data demonstrates that load balancing is a central issue for parallel MCE.

\begin{figure}[t]
	\centering
	\begin{tabular}{cc}
		\includegraphics[width=.5\textwidth]{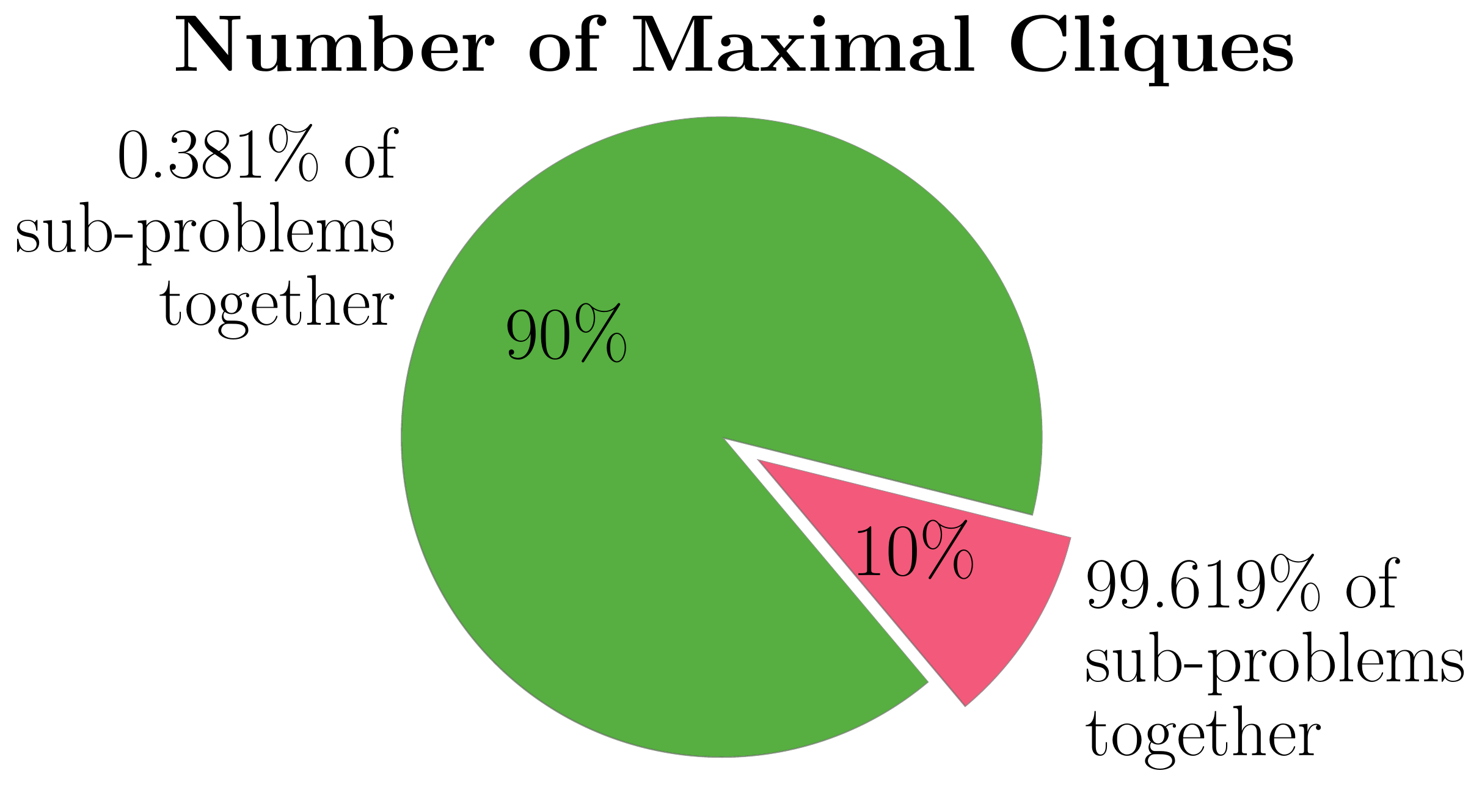} &
		\includegraphics[width=.5\textwidth]{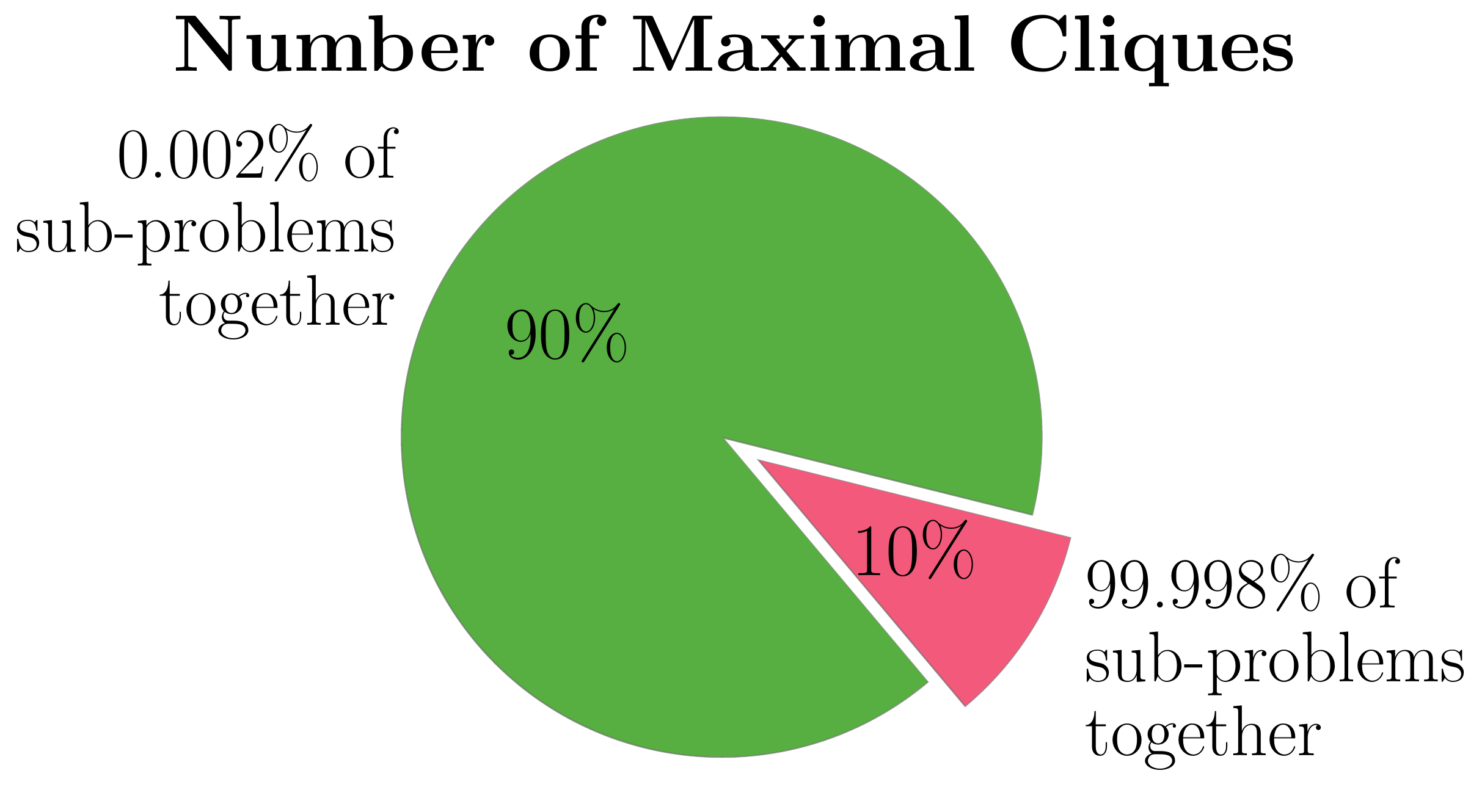} \\
		\textbf{(a)} $\skitter$ &
		\textbf{(b)} $\wtalk$ \bigskip\\
		\includegraphics[width=.5\textwidth]{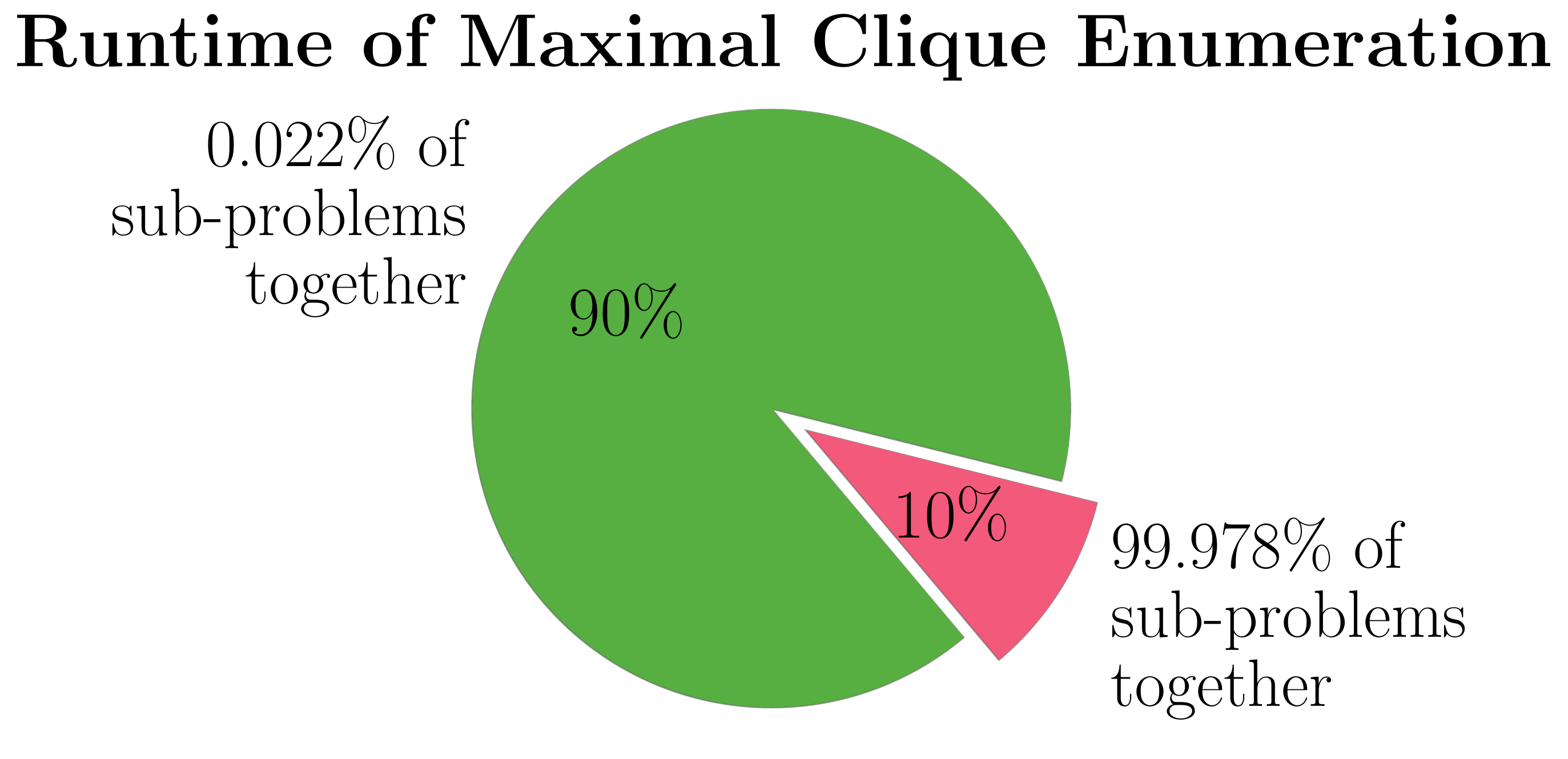} &
		\includegraphics[width=.5\textwidth]{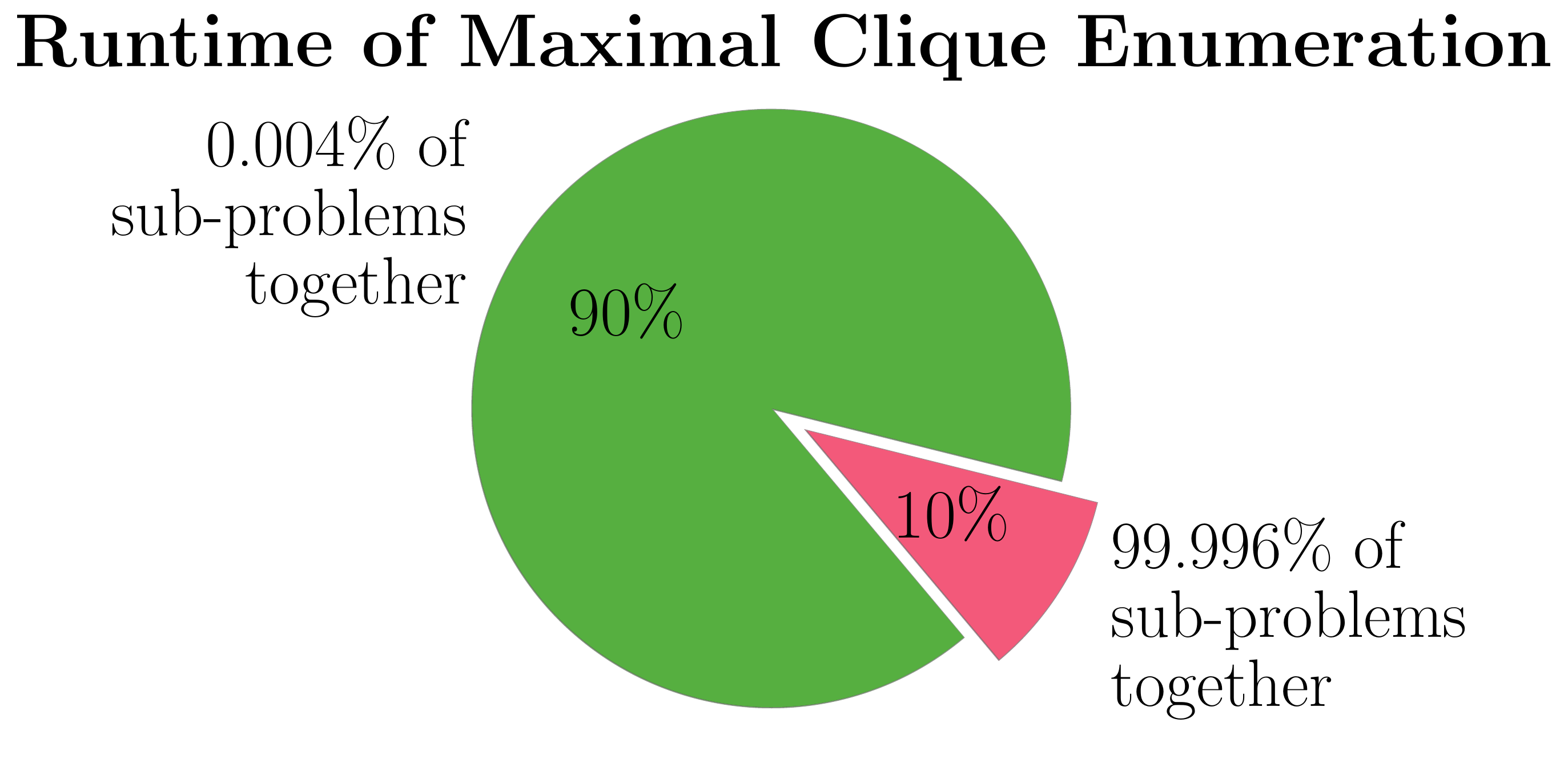}\\	
		\textbf{(c)} $\skitter$&
		\textbf{(d)} $\wtalk$\\
	\end{tabular}
	\caption{Imbalanced in sizes of sub-problems for MCE, where each sub-problem corresponds to the maximal cliques of a single vertex in the given graph. (a) $\skitter$: $0.3\%$ of sub-problems form $90\%$ of total number of maximal cliques. (b) $\wtalk$: only $0.002\%$ of sub-problems yield $90\%$ of all maximal cliques. (c) $\skitter$: $0.02\%$ of sub-problems take $90\%$ of total runtime of MCE. (d) $\wtalk$: only $0.004\%$ of sub-problems take $90\%$ of total runtime of MCE \label{fig:irregular}}
\end{figure}

Prior works on parallel MCE have largely focused on distributed memory algorithms~\cite{WY+09, SS+09, LGG10, SMT+15}. There are a few works on shared-memory parallel algorithms~\cite{ZA+05, DW+09, LP+17}. However, these algorithms do not scale to larger graphs due to memory or computational bottlenecks -- either the algorithms miss out significant pruning opportunities as in~\cite{DW+09}, or they need to generate a large number of non-maximal cliques as in~\cite{ZA+05, LP+17}.

\subsection{Our Contributions}
We make the following contributions towards enumerating all maximal cliques in a simple graph.

\medskip

\noindent \textbf{Theoretically Efficient Parallel Algorithm:}~We present a shared-memory parallel algorithm $\partomita$, which takes as input a graph $G$ and enumerates all maximal cliques in $G$. $\partomita$ is an efficient parallelization of the algorithm due to Tomita, Tanaka, and Takahashi~\cite{TTT06}. Our analysis of $\partomita$ using a work-depth model of computation~\cite{S17} shows that it is work-efficient when compared with~\cite{TTT06} and has a low parallel depth. To our knowledge, this is the first shared-memory parallel algorithm for MCE with such provable properties.

\medskip

\noindent \textbf{Optimized Parallel Algorithm:}~We present a shared-memory parallel algorithm $\parmce$ that builds on $\partomita$ and yields improved practical performance. Unlike $\partomita$, which starts with a single task at the top level that spawns recursive subtasks as it proceeds, leading to a lack of parallelism at the top level of recursion, $\parmce$ spawns multiple parallel tasks at the top level. To achieve this, $\parmce$ uses per-vertex parallelization, where a separate sub-problem is created for each vertex and different sub-problems are processed in parallel. Each sub-problem is required to enumerate cliques which contain the assigned vertex and care is taken to prevent overlap between sub-problems. Each per-vertex sub-problem is further processed in parallel using $\partomita$ -- this additional (recursive) level of parallelism using $\partomita$ is important since different per-vertex sub-problems may have significantly different computational costs, and having each run as a separate sequential task may lead to uneven load balance. To further address load balance, we use a vertex ordering in assigning cliques to different per-vertex sub-problems. For ordering the vertices, we use various metrics such as degree, triangle count, and the degeneracy number of the vertices.

\medskip

\noindent \textbf{Incremental Parallel Algorithm:}~Next, we present a parallel algorithm $\parimce$ that can maintain the set of maximal cliques in a dynamic graph, when the graph is updated due to the addition of new edges.  When a batch of edges are added to the graph, $\parimce$ can (in parallel) enumerate the set of all new maximal cliques that emerged and the set of all maximal cliques that are no longer maximal (subsumed cliques). $\parimce$ consists of two parts: $\parcsnewttt$ for enumerating new maximal cliques, and $\parcssub$ for enumerating subsumed maximal cliques. We analyze $\parimce$ using the work-depth model and show that it is work-efficient relative to an efficient sequential algorithm, and has a low parallel depth. A summary of our algorithms is shown in Table~\ref{tab:algo}.

\begin{table}[]
	\caption{Summary of shared-memory parallel algorithms for MCE.}
	\label{tab:algo}
	\resizebox{\textwidth}{!}{%
		\begin{tabular}{|l|c|l|}
			\hline
			\multicolumn{1}{|c|}{\textbf{Algorithm}} & \textbf{Type} & \multicolumn{1}{c|}{\textbf{Description}} \\ \hline\hline
			$\partomita$ & Static & A work-efficient parallel algorithm for MCE on a static graph \\ \hline
			$\parmce$ & Static & A practical parallel algorithm for MCE on a static graph dealing with load imbalance \\ \hline
			\multirow{2}{*}{$\parimce$} & \multirow{2}{*}{Dynamic} & \begin{tabular}[c]{@{}l@{}}$\parcsnewttt$: A work-efficient parallel algorithm for enumerating new \\ maximal cliques in a dynamic graph.\end{tabular} \\ \cline{3-3} 
			&  & \begin{tabular}[c]{@{}l@{}}$\parcssub$: A work-efficient parallel algorithm for enumerating subsumed \\ maximal cliques in a dynamic graph.\end{tabular} \\ \hline
		\end{tabular}%
	}
\end{table}

\medskip 

\noindent \textbf{Experimental Evaluation:}~We implemented all our algorithms, and our experiments show that $\parmce$ yields a speedup of \textbf{15x-21x} when compared with an efficient sequential algorithm (due to Tomita et al.~\cite{TTT06}) on a multicore machine with $32$ physical cores and $1$ TB RAM. For example, on the $\wikipedia$ network with around $1.8$ million vertices, $36.5$ million edges, and $131.6$ million maximal cliques, $\partomita$ achieves a \textbf{16.5x} parallel speedup over the sequential algorithm, and the optimized $\parmce$ achieves a \textbf{21.5x} speedup, and completed in approximately two minutes. In contrast, prior shared-memory parallel algorithms for MCE~\cite{ZA+05,DW+09,LP+17} failed to handle the input graphs that we considered, and either ran out of memory (\cite{ZA+05,LP+17}) or did not complete in 5 hours (\cite{DW+09}).

On dynamic graphs, we observe that $\parimce$ gives a \textbf{3x}-\textbf{19x} speedup over a state-of-the-art sequential algorithm $\imce$~\cite{DST16} on a multicore machine with 32 cores. Interestingly, the speedup of the parallel algorithm increases with the magnitude of change in the set of maximal cliques -- the ``harder'' the dynamic enumeration task is, the larger is the speedup obtained.  For example, on a dense graph such as $\cacit$ (with original graph density of $0.01$), we get approximately a \textbf{19x} speedup over the sequential $\imce$. More details are presented in Section~\ref{sec:exp}.

\medskip 

\noindent \textbf{Techniques for Load Balancing:} Our parallel methods can effectively balance the load in solving parallel MCE.  As shown in \cref{fig:irregular}, ``natural'' sub-problems of MCE are highly imbalanced, and, therefore, load balancing is not trivial. In our algorithms, sub-problems of MCE are broken down into smaller sub-problems, according to the search used by the sequential algorithm~\cite{TTT06}, and this process continues recursively. As a result, the final sub-problem that is solved in a single task is not so large as to create load imbalances. Our experiments in \cref{sec:exp} demonstrate that the recursive splitting of sub-problems in MCE is essential for achieving a high speedup over existing algorithms \cite{TTT06}. In order to efficiently assign these (dynamically created) tasks to threads at runtime, we utilize a \textit{work stealing scheduler}~\cite{R95,RL99}. 

\modified{
\textbf{Roadmap.} The rest of the paper is organized as follows. We discuss related works in Section \ref{sec:related}, we present preliminaries in Section~\ref{sec:prelims}, followed by a description of algorithms for a static graph in  Section~\ref{sec:algo}, algorithms for a dynamic graph in Section~\ref{sec:inc-algo}, an experimental evaluation in Section~\ref{sec:exp}, and conclusions in Section~\ref{sec:conclude}.
}

\section{Related Work}
\label{sec:related}
Maximal Clique Enumeration (MCE) from a graph is a fundamental problem that has been extensively studied for more than two decades, and there are multiple prior works on sequential and parallel algorithms. We first discuss sequential algorithms for MCE, followed by parallel algorithms.\\

\noindent\textbf{Sequential MCE:} Bron and Kerbosch~\cite{BK73} presented an algorithm for MCE based on depth-first-search. Following their work, a number of algorithms have been presented ~\cite{TI+77,CN85,TTT06,Koch01,MU04,CK08,ELS10}. The algorithm of Tomita et al.~\cite{TTT06} has a worst-case time complexity $O(3^{\frac{n}{3}})$ for an $n$ vertex graph, which is optimal in the worst-case, since the size of the output can be as large as $O(3^{\frac{n}{3}})$~\cite{MM65}. Eppstein et al.~\cite{ELS10,ES11} present an algorithm for sparse graphs whose complexity can be parameterized by the degeneracy of the graph, a measure of graph sparsity.

Another approach to MCE is a class of ``output-sensitive'' algorithms whose time complexity for enumerating maximal cliques is a function of the size of the output. There exist many such output-sensitive algorithms for MCE, including~\cite{CG+16,CN85,MU04,TI+77}, which can be viewed as instances of a general paradigm called ``reverse-search''~\cite{AF93}. 
\vahid{A recent algorithm~\cite{CG+16} provides favorable tradeoffs for delay (time between two enumerated maximal cliques), when compared with prior works. In terms of practical performance, the best output-sensitive algorithms~\cite{CG+16,CN85,MU04} are not as efficient as the best depth-first-search based algorithms~\cite{TTT06,ELS10}.}
Other sequential methods for MCE include algorithms due to Kose et al.~\cite{KW+01}, Johnson et al.~\cite{JYP88}, Li et al.~\cite{LS+-DASFAA-19}, on a special class of graphs due to Fox et al.~\cite{FRSW-SIAMJ-18}, on temporal graph due to Qin et al.~\cite{QL+-ICDE-19}. Multiple works have considered sequential algorithms for maintaining the set of maximal cliques~\cite{DST16,S04,OV10} on a dynamic graph, and, to our knowledge, the most efficient algorithm is the one due to Das et al.~\cite{DST16}.



\medskip

\noindent\textbf{Parallel MCE: }There are multiple prior works on parallel algorithms for MCE ~\cite{ZA+05,DW+06,WY+09,SS+09,LGG10,SMT+15, YL-PC-19}. We first discuss shared-memory algorithms and then distributed memory algorithms.
Zhang et al.~\cite{ZA+05} present a shared-memory parallel algorithm based on the sequential algorithm due to Kose et al.~\cite{KW+01}. This algorithm computes maximal cliques in an iterative manner, and, in each iteration, it maintains a set of cliques that are not necessarily maximal and, for each such clique, maintains the set of vertices that can be added to form larger cliques. This algorithm does not provide a theoretical guarantee on the runtime and suffers for large memory requirement. Du et al.~\cite{DW+06} present a output-sensitive shared-memory parallel algorithm for MCE, but their algorithm suffers from poor load balancing as also pointed out by Schmidt et al.~\cite{SS+09}. Lessley et al.~\cite{LP+17} present a shared memory parallel algorithm that generates maximal cliques using an iterative method, where in each iteration, cliques of size $(k-1)$ are extended to cliques of size $k$.  The algorithm of~\cite{LP+17} is memory-intensive, since it needs to store a number of intermediate non-maximal cliques in each iteration. Note that the number of non-maximal cliques may be far higher than the number of maximal cliques that are finally emitted, and a number of distinct non-maximal cliques may finally lead to a single maximal clique. In the extreme case, a complete graph on $n$ vertices has $(2^n-1)$ non-maximal cliques, and only a single maximal clique. We present a comparison of our algorithm with~\cite{LP+17,ZA+05,DW+06} in later sections. 

Distributed memory parallel algorithms for MCE include works due to Wu et al.~\cite{WY+09}, designed for the MapReduce framework, Lu et al.~\cite{LGG10}, which is based on the sequential algorithm due to Tsukiyama et al.~\cite{TI+77}, Svendsen et al.~\cite{SMT+15}, Wang et al. \cite{WC+17}, and algorithm for sparse graph due to Yu and Liu~\cite{YL-PC-19}. 
Other works on parallel and sequential algorithms for enumerating dense subgraphs from a massive graph include sequential algorithms for enumerating quasi-cliques \cite{SD+18,LL08,U10}, parallel algorithms for enumerating $k$-cores~\cite{MD+13,DDZ14,KM17,SSP17}, $k$-trusses~\cite{SSP17,KM17-2,KM17-3}, nuclei~\cite{SSP17}, and distributed memory algorithms for enumerating $k$-plexes \cite{CM+18}, bicliques~\cite{MT17}.


\section{Preliminaries}
\label{sec:prelims}

We consider a simple undirected graph without self loops or multiple edges. For graph $G$, let $V(G)$ denote the set of vertices in $G$ and $E(G)$ denote the set of edges in $G$. Let $n$ denote the size of $V(G)$, and $m$ denote the size of $E(G)$. For vertex $u \in V(G)$, let $\Gamma_G(u)$ denote the set of vertices adjacent to $u$ in $G$. When the graph $G$ is clear from the context, we use $\Gamma(u)$ to mean $\Gamma_G(u)$. Let $\cliques(G)$ denote the set of all maximal cliques in $G$.\\ 

\noindent\textbf{Sequential Algorithm $\tomita$: } The algorithm due to Tomita, Tanaka, and Takahashi.~\cite{TTT06}, which we call $\tomita$, is a recursive backtracking-based algorithm for enumerating all maximal cliques in an undirected graph, with a worst-case time complexity of $O(3^{n/3})$ where $n$ is the number of vertices in the graph. In practice, this is one of the most efficient sequential algorithms for MCE. Since we use $\tomita$ as a subroutine in our parallel algorithms, we present a short description here.

In any recursive call, $\tomita$ maintains three disjoint sets of vertices $K$, $\cand$, and $\fini$, where $K$ is a candidate clique to be extended, $\cand$ is the set of vertices that can be used to extend $K$, and $\fini$ is the set of vertices that are adjacent to $K$, but need not be used to extend $K$ (these are being explored along other search paths). Each recursive call iterates over vertices from $\cand$ and in each iteration, a vertex $q \in \cand$ is added to $K$ and a new recursive call is made with parameters $K\cup\{q\}$, $\cand_q$, and $\fini_q$ for generating all maximal cliques of $G$ that extend $K\cup\{q\}$ but do not contain any vertices from $\fini_q$. The sets $\cand_q$ and $\fini_q$ can only contain vertices that are adjacent to all vertices in $K\cup \{q\}$. The clique $K$ is a maximal clique when both $\cand$ and $\fini$ are empty.

The ingredient that makes $\tomita$ different from the algorithm due to Bron and Kerbosch~\cite{BK73} is the use of a ``pivot'' where a vertex $u\in \cand\cup\fini$ is selected that maximizes $\vert \cand \cap \Gamma(u) \vert$. Once the pivot $u$ is computed, it is sufficient to iterate over all the vertices of $\cand \setminus \Gamma(u)$, instead of iterating over all vertices of $\cand$. The pseudo code of $\tomita$ is presented in Algorithm~\ref{algo:tomita}. For the initial call, $K$ and $\fini$ are initialized to an empty set, $\cand$ is the set of all vertices of $G$.\\ 


\begin{algorithm}[htp!]
\DontPrintSemicolon
\caption{$\tomita({G},K,\cand,\fini)$}
\label{algo:tomita}
\KwIn{${G}$ - the input graph \\ \hspace{1cm} $K$ - a clique to extend, \\
$\cand$ - a set of vertices that can be used extend $K$, \\ 
$\fini$ - a set of vertices that have been used to extend $K$}
\KwOut{Set of all maximal cliques of $G$ containing $K$ and vertices from $\cand$ but not containing any vertex from $\fini$}
\If{$(\cand = \emptyset)$ \& $(\fini = \emptyset)$}{
	Output $K$ and return\;
}
$\pivot \gets (u \in \cand \cup \fini)$ such that $u$ maximizes the size of $\cand  \cap \Gamma_{{G}}(u)$\;
$\ext \gets \cand - \Gamma_{{G}}(\pivot)$\;
\For{$q \in$ \ext}{
	$K_q \gets K\cup\{q\}$\;
	$\cand_q \gets \cand\cap\Gamma_{{G}}(q)$\;
	$\fini_q \gets \fini\cap\Gamma_{{G}}(q)$\;
    $\cand \gets \cand-\{q\}$\;
	$\fini \gets \fini \cup\{q\}$\;
	$\tomita({G},K_q,\cand_q,\fini_q)$\;
}
\end{algorithm}

\noindent\textbf{Parallel Cost Model: } For analyzing our shared-memory parallel algorithms, we use the CRCW PRAM model~\cite{BM10}, which is a model of shared parallel computation that assumes concurrent reads and concurrent writes. Our parallel algorithm can also work in other related models of shared-memory such as EREW PRAM (exclusive reads and exclusive writes), with a logarithmic factor increase in work as well as parallel depth. We measure the effectiveness of the parallel algorithm using the {\em work-depth} model~\cite{S17}. Here, the ``work'' of a parallel algorithm is equal to the total number of operations of the parallel algorithm, and the ``depth'' (also called the ``parallel time'' or the ``span'') is the longest chain of dependent computations in the algorithm. A parallel algorithm is said to be {\em work-efficient} if its total work is of the same order as the work due to the best sequential algorithm.\footnote{Note that work-efficiency in the CRCW PRAM model does not imply work-efficiency in the EREW PRAM model} We aim for work-efficient algorithms with a low depth, ideally poly-logarithmic in the size of the input. Using Brent's theorem~\cite{BM10}, it can be seen that a parallel algorithm on input size $n$ with a depth of $d$ can theoretically achieve $\Theta(p)$ speedup on $p$ processors as long as $p = O(n/d)$.

We next restate a result on concurrent hash tables~\cite{SS06} that we use in proving the work and depth bounds of our parallel algorithms.

\begin{theorem}[Theorem 3.15 \cite{SS06}]\label{thm:parset}
There is an implementation of a hash table, which, given a hash function with expected uniform distribution, performs $n_1$ insert, $n_2$ delete and $n_3$ find operations in parallel using $O(n_1+n_2+n_3)$ work and $O(1)$ depth on average. 
\end{theorem}

\section{Parallel MCE Algorithms on a Static Graph}
\label{sec:algo}
In this section, we present new shared-memory parallel algorithms for MCE. We first describe a parallel algorithm $\partomita$, a parallelization of the sequential $\tomita$ algorithm and an analysis of its theoretical properties. Then, we discuss bottlenecks in $\partomita$ that arise in practice, leading us to another algorithm $\parmce$ with a better practical performance. $\parmce$ uses $\partomita$ as a subroutine -- it creates appropriate sub-problems that can be solved in parallel and hands off the enumeration task to $\partomita$. 
 
\subsection{Algorithm $\partomita$}
Our first algorithm $\partomita$ is a work-efficient parallelization of the sequential $\tomita$ algorithm. The two main components of $\tomita$ (Algorithm~\ref{algo:tomita}) are (1)~Selection of the pivot element (Line~$3$) and (2)~Sequential backtracking for extending candidate cliques until all maximal cliques are explored (Line~$5$ to Line~$11$). We discuss how to parallelize each of these steps.

\paragraph{Parallel Pivot Selection: } Within a single recursive call of $\partomita$, the pivot element is computed in parallel using two steps, as described in $\parpivot$ (Algorithm~\ref{algo:pivot}). In the first step, the size of the intersection $\cand \cap \Gamma(u)$ is computed in parallel for each vertex $u \in \cand \cup \fini$. In the second step, the vertex with the maximum intersection size is selected. The parallel algorithm for selecting a pivot is presented in Algorithm~\ref{algo:pivot}. The following lemma proves that the parallel pivot selection is work-efficient with logarithmic depth:

\begin{lemma}
\label{lem:parpivot}
The total work of $\parpivot$ is $O(\sum_{w \in \cand \cup \fini}  (\min\{|\cand|, |\Gamma(w)|\}))$, which is $O(n^2)$, and depth is $O(\log n)$.
\end{lemma}

\begin{proof}
If sets $\cand$ and $\Gamma(w)$ are stored as hashsets, then, for vertex $w$, the size $t_w = \vert\intersect(\cand, \Gamma(w))\vert$ can be computed sequentially in time $O(\min\{|\cand|, |\Gamma(w)|\})$ -- the intersection of two sets $S_1$ and $S_2$ can be found by considering the smaller set among the two, say $S_2$, and searching for its elements within the larger set, say $S_1$. It is possible to parallelize the computation of $\intersect(S_1,S_2)$ by executing the search elements of $S_2$ in parallel, followed by counting the number of elements that lie in the intersection, which can also be done in parallel in a work-efficient manner using $O(1)$ depth using Theorem~\ref{thm:parset}. Since computing the maximum of a set of $n$ numbers can be accomplished using work $O(n)$ and depth $O(\log n)$, for vertex $w$, $t_w$ can be computed using work $O(\min\{|\cand|, |\Gamma(w)|\})$ and depth $O(\log n)$. Once $t_w$ (i.e. $|\cand \cap \Gamma(w)|$) is computed for every vertex $w \in \cand\cup\fini$, $argmax(\{t_w : w \in \cand\cup\fini\})$ can be obtained using additional work $O(\vert \cand \cup \fini \vert)$ and depth $O(\log n)$. Hence, the total work of $\parpivot$ is $O(\sum_{w \in \cand \cup \fini}  (\min\{|\cand|, |\Gamma(w)|\})$. Since the size of $\cand$, $\fini$, and $\Gamma(w)$ are bounded by $n$, this is $O(n^2)$,  but typically much smaller in practice.
\end{proof}
 
\begin{figure}[htp!]
\centering
\begin{mdframed}[style=MyFrame]
\begin{minipage}[t]{\textwidth}
\centering
\scalebox{0.9}{
\begin{algorithm}[H]
\DontPrintSemicolon
\caption{$\parpivot(G,K, \cand,\fini)$}
\label{algo:pivot}
\KwIn{ 
$G$ - the input graph;
$K$ - a clique in $G$ that may be further extended;
$\cand$ - a set of vertices that may extend $K$;
$\fini$ - a set of vertices that have been used to extend $K$.}
\KwOut{
Pivot vertex $v \in \cand\cup\fini$.
}
\ForPar {$w \in \cand\cup\fini$}{
	In parallel, compute $t_w \gets |\intersect(\cand, \Gamma_{G}(w))|$\;
}
In parallel, find $v \gets argmax(\{t_w : w \in \cand\cup\fini\})$\;
\Return{$v$}\;
\end{algorithm}}
\end{minipage}
\end{mdframed}
\end{figure}

\paragraph{Parallelization of Backtracking: } We first note that there is a sequential dependency among the different iterations within a recursive call of $\tomita$. In particular, the contents of the sets $\cand$ and $\fini$ in a given iteration are derived from the contents of $\cand$ and $\fini$ in the previous iteration. Such sequential dependence of updates to $\cand$ and $\fini$ restricts us from calling the recursive $\tomita$ for different vertices of $\ext$ in parallel. To remove this dependency, we adopt a different view of $\tomita$ which enables us to make the recursive calls in parallel. The elements of $\ext$, the vertices to be considered for extending a maximal clique, are arranged in a predefined total order. Then, we unroll the loop and explicitly compute the parameters $\cand$ and $\fini$ for recursive calls. 

Suppose $\langle v_1, v_2, ..., v_{\kappa} \rangle$ is the order of vertices in $\ext$. Each vertex $v_i \in \ext$, once added to $K$, should be removed from further consideration in $\cand$. To ensure this, in $\partomita$, we explicitly remove vertices $v_1, v_2, ..., v_{i-1}$ from $\cand$ and add them to $\fini$, before making the recursive calls. As a consequence, parameters of the $i$th iteration are computed independently of prior iterations. 

\begin{figure}[htp!]
	\centering
	\begin{mdframed}[style=MyFrame]
		\begin{minipage}[t]{\textwidth}
			\begin{algorithm}[H]
				\DontPrintSemicolon
				\caption{$\partomita({G},K,\cand,\fini)$}
				\label{algo:partomita}
				\KwIn{${G}$ - the input graph \\ $K$ - a non-maximal clique to extend \\ 
					$\cand$ - set of vertices that may extend $K$ \\ 
					$\fini$ - vertices that have been used to extend $K$}
				\KwOut{A set of all maximal cliques of $G$ containing $K$ and vertices from $\cand$ but not containing any vertex from $\fini$}
				\If{$(\cand = \emptyset)$ \& $(\fini = \emptyset)$}{
					Output $K$ and \Return
				}
				$\pivot \gets \parpivot({G},\cand,\fini)$\;
				$\ext[1..\kappa] \gets \cand - \Gamma_{{G}}(\pivot)$ \tcp{in parallel} 
				\ForPar{$i \in$ $[1..\kappa]$}{
					$q\gets\ext[i]$\;
					$K_q \gets K\cup\{q\}$\;
					$\cand_q \gets \intersect(\cand\setminus\ext[1..i-1], \Gamma_{{G}}(q))$\;
					$\fini_q \gets \intersect(\fini\cup\ext[1..i-1], \Gamma_{{G}}(q))$\;
					$\partomita({G},K_q,\cand_q,\fini_q)$\;
				}
			\end{algorithm}
		\end{minipage}
	\end{mdframed}
\end{figure}

Here, we prove the work efficiency and low depth of $\partomita$ in the following lemma:
\begin{lemma}
\label{lem:partomita}
Total work of $\partomita$ (Algorithm~\ref{algo:partomita}) is $O(3^{n/3})$ and depth is $O(M\log{n})$ where $n$ is the number of vertices in the graph, and $M$ is the size of a maximum clique in $G$.
\end{lemma} 

\begin{proof}
First, we analyze the total work. Note that the computational tasks in $\partomita$ is different from $\tomita$ at Line~8 and Line~9 of $\partomita$ where at an iteration $i$, we remove all vertices $\{v_1, v_2, ..., v_{i-1}\}$ from $\cand$ and add all these vertices to $\fini$ as opposed to the removal of a single vertex $v_{i-1}$ from $\cand$ and addition of that vertex to $\fini$ as in $\tomita$ (Line~8 and Line~9 of Algorithm~\ref{algo:tomita}). Therefore, in $\partomita$, additional $O(n)$ work is required due to independent computations of $\cand_q$ and $\fini_q$. The total work, excluding the call to $\parpivot$ is $O(n^2)$. Adding up the work of $\parpivot$, which requires $O(n^2)$ work and $O(n^2)$ total work for each single call of $\partomita$ excluding further recursive calls (Algorithm~\ref{algo:partomita}, Line 10), which is the same as in original sequential algorithm $\tomita$ (Section 4, \cite{TTT06}). Hence, using Lemma~2 and Theorem~3 of \cite{TTT06}, we infer that the total work of $\partomita$ is the same as the sequential algorithm $\tomita$ and is bounded by $O(3^{n/3})$. 

Next, we analyze the depth of the algorithm. The depth of $\partomita$ consists of the (sum of the) following components: (1)~Depth of $\parpivot$, (2)~Depth of computation of $\ext$, (3)~Maximum depth of an iteration in the {\small\textbf{\texttt{for}}} loop from Line~5 to Line~10. According to Lemma~\ref{lem:parpivot}, the depth of $\parpivot$ is $O(\log{n})$. The depth of computing $\ext$ is $O(\log n)$ since it takes $O(1)$ time to check whether an element in $\cand$ is in the neighborhood of $\pivot$. Similarly, the depth of computing $\cand_q$ and $\fini_q$ at Line~8 and Line~9 are $O(\log n)$. The remaining is the depth of the call of $\partomita$ at Line~10. Notice that the recursive call of $\partomita$ continues until there is no further vertex to add for expanding $K$, and this depth can be at most the size of the maximum clique which is $M$ because, at each recursive call of $\partomita$, the size of $K$ increases by $1$. Thus, the overall depth of $\partomita$ is $O(M\log{n})$. 
\end{proof}

\begin{corollary}
\label{cor:partomita}
Using $P$ parallel processors, which are shared-memory, $\partomita$ (Algorithm~\ref{algo:partomita}) is a parallel algorithm for MCE and can achieve a worst case parallel time of $O\left(\frac{3^{n/3}}{M \log n} + P\right)$ using $P$ parallel processors. This is work-optimal and work-efficient as long as $P = O(\frac{3^{n/3}}{M \log n})$.
\end{corollary}

\begin{proof}
The parallel time follows Brent's theorem~\cite{BM10}, which states that the parallel time using $P$ processors is $O(w/d + P)$, where $w$ and $d$ are the work and the depth of the algorithm respectively. If the number of processors $P= O\left(\frac{3^{n/3}}{M \log n}\right)$, then using Lemma~\ref{lem:partomita}, the parallel time is $O\left(\max\{\frac{3^{n/3}}P , {M \log n}\}\right) = O\left(\frac{3^{n/3}}{P}\right)$. The total work across all processors is $O(3^{n/3})$, which is worst-case optimal, since the size of the output can be as large as $3^{n/3}$ maximal cliques (Moon and Moser~\cite{MM65}). 
\end{proof}
 






\subsection{Algorithm $\parmce$}

While $\partomita$ is a theoretically work-efficient parallel algorithm, we note that its runtime can be further improved. While the worst case work complexity of $\parpivot$ matches that of the pivoting routine in $\tomita$, in practice, the work in $\partomita$ can be higher, since computation of $\cand_q$ and $\fini_q$ has additional and growing overhead as the sizes of the $\cand_q$ and $\fini_q$ increase. This can result in a lower speedup than the theoretically expected one.



We set out to improve on this to derive a more efficient parallel implementation through a more selective use of $\parpivot$. In this way, the cost of pivoting can be reduced by carefully choosing many pivots in parallel instead of a single pivot element as in $\partomita$ at the beginning of the algorithm. We first note that the cost of $\parpivot$ is the highest during the iteration when set $K$ (the clique in the search space) is empty. During this iteration, the number of vertices in $\cand \cup \fini$ can be high, as large as the number of vertices in the graph. To improve upon this, we can perform the first few steps of pivoting, when $K$ is empty, using a sequential algorithm. Once set $K$ contains at least one element, the number of the vertices in $\cand \cup \fini$ is decremented to no more than the size of the intersection of neighborhoods of all vertices in $K$, which is typically a number much smaller than the number of vertices in the graph (this number is smaller than the smallest degree of a vertex in $K$). Problem instances with $K$, assigned to a single vertex, can be seen as sub-problems, and, on each of these sub-problems, the overhead of $\parpivot$ computation is much smaller since the number of vertices that have to be dealt with is also much smaller.

Based on this observation, we present a parallel algorithm $\parmce$ which works as follows. The algorithm can be viewed as considering, for each vertex $v \in V(G)$, a subgraph $G_v$ that is induced by the vertex $v$ and its neighborhood $\Gamma_G(v)$. It enumerates all maximal cliques from each subgraph $G_v$ in parallel using $\partomita$. While processing sub-problem $G_v$, it is important to not enumerate maximal cliques that are being enumerated elsewhere, in other sub-problems. To handle this, the algorithm considers a specific ordering of all vertices in $V$ such that $v$ is the least ranked vertex in each maximal clique enumerated from $G_v$. Subgraph $G_v$ for each vertex $v$ is handled in parallel -- these subgraphs need not be processed in any particular order. However, the ordering allows us to populate the $\cand$ and $\fini$ sets accordingly such that each maximal clique is enumerated in exactly one sub-problem. The order in which the vertices are considered is defined by a ``rank'' function \textbf{rank}, which indicates the position of a vertex in the total order. This global ordering on vertices has impact on the total work of the algorithm, as well as the load balance of the distribution of workloads across sub-problems.


\medskip

\textbf{Load Balancing:} Notice that the sizes of the subgraphs $G_v$ may vary widely because of two reasons: (1)~The subgraphs themselves may be of different sizes, depending on the vertex degrees. (2)~The number of maximal cliques and the sizes of the maximal cliques containing $v$ can vary widely from one vertex to another. Clearly, the sub-problems that deal with a large number of maximal cliques or maximal cliques of a large size are more computationally expensive than others.

In order to maintain the size of the sub-problems approximately balanced, we use an idea from PECO~\cite{SMT+15}, where we choose the rank function on the vertices in such a way that for any two vertices $v$ and $w$, \textbf{rank}($v$) $>$ \textbf{rank}($w$) if the complexity of enumerating maximal cliques from $G_v$ is higher than the complexity of enumerating maximal cliques from $G_w$. Indeed, by giving a higher rank to $v$ than $w$, we are decreasing the complexity of the sub-problem $G_v$ since the sub-problem at $G_v$ need not enumerate maximal cliques that involve any vertex whose rank is less than $v$. Therefore, the higher the rank of vertex $v$, the lower is its ``share'' (of maximal cliques it belongs to) of maximal cliques in $G_v$. We use this idea for approximately balancing the workload across sub-problems. The additional enhancements in $\parmce$ when compared with the idea from PECO are as follows: (1)~In PECO, the algorithm is designed for distributed memory such that the subgraphs and sub-problems have to be explicitly copied across the network. (2)~In $\parmce$, the vertex specific sub-problem, dealing with $G_v$ is itself handled through a parallel algorithm, $\partomita$, while, in PECO, the sub-problem for each vertex was handled through a sequential algorithm.


Note that it is computationally expensive to accurately count the number of maximal cliques within $G_v$, and, hence, it is not possible to compute the rank of each vertex exactly according to the complexity of handling $G_v$. Instead, we estimate the cost of handling $G_v$ using some easy-to-evaluate metrics on the subgraphs. In particular, we consider the following:


\begin{itemize}
\item 
\textbf{Degree Based Ranking: } For vertex $v$, define $\rank(v) = (d(v), id(v))$ where $d(v)$ and $id(v)$ are degree and identifier of $v$, respectively. For two vertices $v$ and $w$, $\rank(v) > \rank(w)$ if $d(v) > d(w)$ or $d(v) = d(w)$ and $id(v) > id(w)$, $rank(v) < rank(w)$ otherwise. 

\item
\textbf{Triangle Count Based Ranking: } For vertex $v$, define $\rank(v) = (t(v), id(v))$ where $t(v)$ is the number of triangles, which vertex $v$ is a part of. Note that ranking based on triangle count is more expensive to compute than degree based ranking but may yield a better estimate of the complexity of maximal cliques within $G_v$.\footnote{A triangle is a cycle of length three.}

\item
\textbf{Degeneracy Based Ranking~\cite{ELS10}: } For a vertex $v$, define $\rank(v) = (degen(v), id(v))$ where $degen(v)$ is the degeneracy of a vertex $v$. A vertex $v$ has degeneracy number $k$ when it belongs to a $k$-core but no $(k+1)$-core, where a $k$-core is a maximal induced subgraph such that the minimum degree of each vertex in the subgraph is $k$. A computational overhead of using this ranking is due to computing the degeneracy of vertices which takes $O(n+m)$ time, where $n$ is the number of vertices and $m$ is the number of edges.
\end{itemize}

The different implementations of $\parmce$ using degree, triangle, and degeneracy based rankings are called as $\parmcedegree$, $\parmcetriangle$, $\parmcedegen$ respectively.


\begin{figure}[htp!]
\centering
\begin{mdframed}[style=MyFrame]
\begin{minipage}[t]{\textwidth}
\begin{algorithm}[H]
\DontPrintSemicolon
\caption{$\parmce({G})$}
\label{algo:parmce}
\KwIn{${G}$ - the input graph.}
\KwOut{$\cliques({G})$ - a set of all maximal cliques of ${G}$.}
\ForPar {$v \in V({G})$}{
	Create ${G}_v$, the subgraph of ${G}$ induced by $\Gamma_{{G}}(v)\cup\{v\}$\;
	$K \gets \{v\}$, $\cand \gets \phi$, $\fini \gets \phi$\;
    \ForPar {$w \in \Gamma_{{G}}(v)$}{
    	\lIf{$rank(w) > rank(v)$}{
        	$\cand \gets \cand\cup\{w\}$
        }
        \lElse {
        	$\fini \gets \fini\cup\{w\}$
        }
    }
    $\partomita({G}_v, K, \cand, \fini)$
}
\end{algorithm}
\end{minipage}
\end{mdframed}
\end{figure}
\section{Parallel MCE Algorithm on a Dynamic Graph}
\label{sec:inc-algo}
When the graph changes over time due to addition of edges, the maximal cliques of the updated graph also change. The update in the set of maximal cliques consists of (1) The set of new maximal cliques -- the maximal cliques that are newly formed (2) The set of subsumed cliques -- maximal cliques of the original graph that are subsumed by the new maximal cliques. The combined set of new and subsumed maximal cliques is called the set of changes, and the size of this set refers to the size of change in the set of maximal cliques (see Figure~\ref{fig:dyn-change}).

\remove{
\begin{figure*}
\centering
\begin{tabular}{c}
	\includegraphics[width=.5\textwidth]{pic1.png}\\
\end{tabular}
\caption{\textbf{Change in maximal cliques when graph keeps changing due to addition of new edges. $\{1, 2, 5\}$ and $\{2, 3, 4\}$ are the maximal cliques of the initial graph \textbf{$G$}; $\{2, 3, 4, 5\}$ is the only new maximal clique and $\{2, 3, 4\}$ is the only subsumed maximal clique when \textbf{$G$} is updated to \textbf{$G'$} after adding edges $(3,5)$ and $(4,5)$ to \textbf{$G$} (in the middle); $\{1, 2, 3, 4, 5\}$ is the new maximal cliques and $\{1, 2, 5\}$, $\{2, 3, 4, 5\}$ are the subsumed cliques when \textbf{$G'$} is updated to \textbf{$G''$} after adding edges $(1,3)$ and $(1,4)$ to \textbf{$G'$} (in the right).}}
\label{fig:dyn-change-removed}
\end{figure*}
}

\begin{figure} [t]
	\centering
	\begin{tabular}{cccc}
		\includegraphics[width=.25\textwidth]{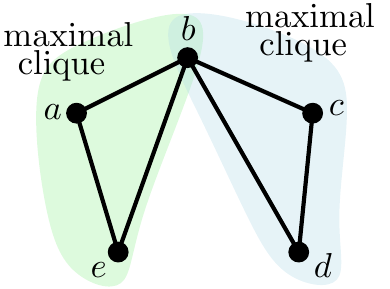} &
		\includegraphics[width=0.25\textwidth]{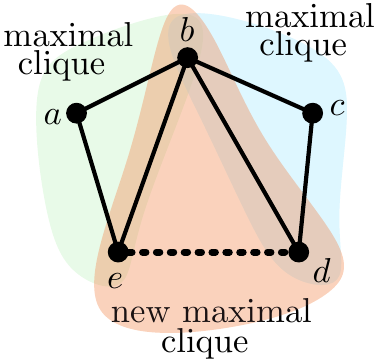} &
		\includegraphics[width=0.25\textwidth]{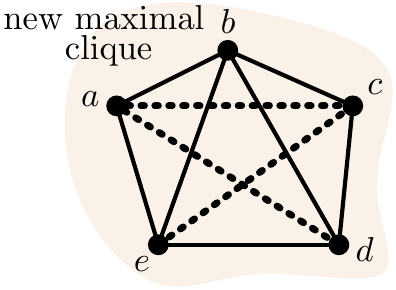} \\
		\textbf{(a)} Input Graph $G$  & \textbf{(b)} A new edge $(e, d)$ &\textbf{(c)} Three new edges: $(a,c),(a,d),(c,e)$ 
	\end{tabular}
\caption{
		Maximal cliques change upon addition of new edges. (a) Two subgraphs $\{a, b, e\}$ and $\{b, c, d\}$ are maximal cliques in the original graph \textbf{$G$}. (b) A new maximal clique (i.e. $\{b, d, e\}$) is created when edge $(e, d)$ is added to $G$. (c) When three more edges $(a, c)$, $(a, d)$, and $(c, e)$ are added, the entire graph turns into a maximal clique, which subsumes all maximal cliques in prior steps. 
	\label{fig:dyn-change}
}
\end{figure}

\remove{
It is important to update the set of maximal cliques in a dynamic graph when it serves as the building block in many important problems. For example, the work of Chateau et al.~\cite{CRR11} on maintaining common intervals among genomes, the work of Duan et al.~\cite{DL+12} on incremental $k$-clique clustering, the work of Hussain et al.~\cite{HM+15} on maintaining the maximum range-sum query over a point stream use maximal clique as the building block in an appropriately defined graph.
}
\begin{figure}[htp!]
	\centering
	\begin{mdframed}[style=MyFrame]
		\begin{minipage}[t]{\textwidth}
			\begin{algorithm}[H]
				\DontPrintSemicolon
				\caption{$\parcsnewttt(G,H)$}
				\label{algo:parcsnewttt}
				\KwIn{$G$ - the input graph \\  $H$ - a set of $\rho$ edges being added to $G$}
				\KwOut{Cliques in $\Lambda^{new} = \cliques(G+H)\setminus\cliques(G)$}
				$G' \gets G + H$\;
				Consider edges of $H$ in an arbitrary order $e_1, e_2,\ldots,e_{\rho}$\;
				
				\ForPar{$i \gets 1,2, \ldots,\rho$}{
					$e \gets e_i = (u,v)$\;
					
					$V_e \gets \{u,v\}\cup\{\Gamma_{G'}(u)\cap\Gamma_{G'}(v)\}$\;
					${G}'' \gets$ Graph induced by $V_e$ on $G'$\;
					$K \gets \{u,v\}$\;
					$\cand \gets V_e \setminus \{u,v\}$ ; $\fini \gets \emptyset$\;
					$S\gets\partomitaE({G}'', K, \cand, \fini, \{e_1, e_2, ...,e_{i-1}\})$\;
					$\Lambda^{new}\gets\Lambda^{new}\cup S$\;
				}
			\end{algorithm}
		\end{minipage}
	\end{mdframed}
\end{figure}

Note that the size of change can be as small as $O(1)$ or as large as exponential in the size of the graph upon addition of a new single edge. For example, consider a graph of size $n$ which is missing a single edge from being a clique. The size of change is only $3$ when that missing edge is added to the graph because there will be only one new maximal clique of size $n$ and two subsumed cliques, each of size $(n-1)$. On the other hand, consider a Moon-Moser graph~\cite{MM65} of size $n$. Addition of a single edge to this graph makes the size of the changes in the order of $O(3^{n/3})$.

In a previous work, we presented a sequential algorithm $\imce$~\cite{DST16}, which efficiently tackles the problem of updating the set of maximal cliques of a dynamic graph in an incremental model when new edges are added at a time. $\imce$ consists of $\csnewttt$ for computing new maximal cliques and $\cssub$ for computing subsumed cliques. However, $\imce$ is still unable to update the set of maximal cliques when the size of change is large. For instance, it takes $\imce$ around $9.4$ hours to update the set of maximal cliques, when the first $90K$ edges of graph $\cacit$ are added incrementally (with the original graph density 0.01). The high computational cost of $\imce$ calls for parallel methods.

In this section, we present parallel algorithms for enumerating the set of new and subsumed cliques when an edge set $H = \{e_1, e_2, ..., e_{\rho}\}$ is added to a graph $G$. Our parallel algorithms are based on $\imce$~\cite{DST16}. In this work, we focus on (1) processing new edges in parallel, (2) enumerating new maximal cliques using $\partomita$, and (3) Parallelizing $\cssub$~\cite{DST16}. First, we describe an efficient parallel algorithm for generating new maximal cliques, i.e. the maximal cliques in $G+H$, which are not present in $G$ and then an efficient parallel algorithm for generating subsumed maximal cliques, i.e. the cliques which are maximal in $G$ but not maximal in $G+H$. We present a shared-memory parallel algorithm $\parimce$ for the incremental maintenance of maximal cliques. $\parimce$ consists of (1) algorithm $\parcsnewttt$ for enumerating new maximal cliques and (2) algorithm $\parcssub$ for enumerating subsumed cliques. A brief description of the algorithms is discussed in Table~\ref{tab:algo-describe}. 
	\cref{fig:dyn-algo} also sketches the parallel shared-memory setup for enumerating maximal cliques in a dynamic graph upon addition of new batches. 

\begin{figure*}
	\centering
	\begin{tabular}{c}
		\includegraphics[width=1.0\textwidth]{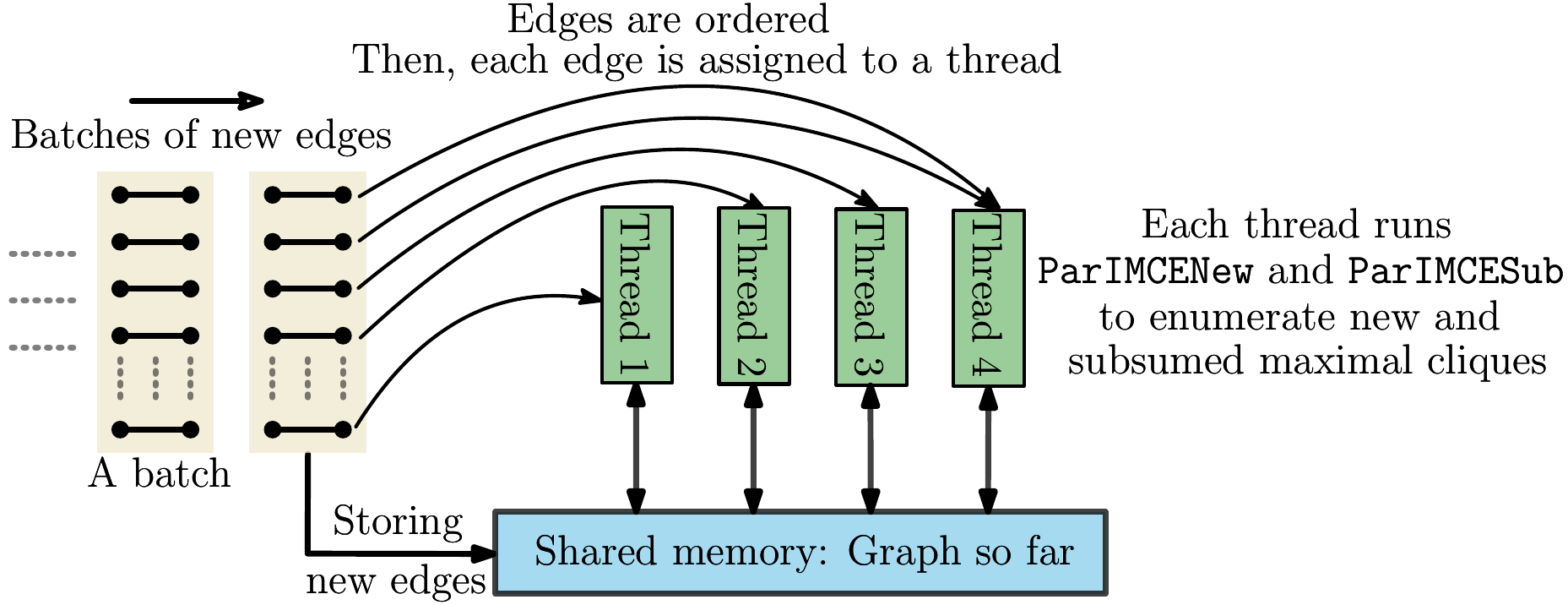}\\
	\end{tabular}
	\caption{
			Our shared-memory set up for processing a dynamic graph.
	}
	\label{fig:dyn-algo}
\end{figure*}
\begin{table}[b!]
	\caption{Brief description of the incremental algorithms in this work.}
	\label{tab:algo-describe}
	\resizebox{\textwidth}{!}{%
		\begin{tabular}{|c|c|c|l|}
			\hline
			\textbf{Objective}                                                     & \textbf{\begin{tabular}[c]{@{}c@{}}Sequential\\ Algorithm {\cite{DST16}}\end{tabular}} & \textbf{\begin{tabular}[c]{@{}c@{}}Parallel Algorithm\\  (this work)\end{tabular}} & \multicolumn{1}{c|}{\textbf{Overview of Parallel Algorithms}}                                                                                                                                                                                   \\ \hline\hline
			\begin{tabular}[c]{@{}c@{}}Enumerating new\\ maximal cliques\end{tabular}  & $\csnewttt$                                                                   & $\parcsnewttt$                                                                         & \begin{tabular}[c]{@{}l@{}}(1) Process new edges in parallel.\\ (2) Enumerate maximal cliques using $\partomitaE$.\end{tabular}                                                                                                            \\ \hline
			\begin{tabular}[c]{@{}c@{}}Enumerating\\ subsumed cliques\end{tabular} & $\cssub$                                                                    & $\parcssub$                                                                    & \begin{tabular}[c]{@{}l@{}}(1) Generate candidates in parallel\\ (executing inner for loop of $\cssub$ in parallel).\\ (2) Process each candidate in parallel\\ (executing candidate processing step of $\cssub$ in parallel).\end{tabular} \\ \hline
		\end{tabular}%
	}
\end{table}
\remove{
\begin{table}[htp!]
\caption{Comparison between $\pmb{\tomitaE}$ and $\pmb{\partomitaE}$}
\label{tab:algo-depends-on}
\centering
\scalebox{0.8}{
\begin{tabular}{|l | l |}
\toprule
$\tomitaE$~\cite{DST16} & $\partomitaE$ \\
\midrule
\makecell[l]{1. A subroutine used in $\csnewttt$.} & \makecell[l]{1. A parallel enumeration algorithm developed in this work and we use \\ it in $\parcsnewttt$ as a subroutine.}\\
\hline
\makecell[l]{2. Avoid enumerating duplicates using \\ the arrival ordering of the new edges.} & \makecell[l]{2. Avoid enumerating duplicates assuming \\ a global ordering of the new edges.} \\
\hline
3. This algorithm is based on $\tomita$. & \makecell[l]{3. This algorithm is based on parallelization technique of $\partomita$ and \\ duplicate avoidance technique of $\csnewttt$.}\\
\bottomrule
\end{tabular}
}
\end{table}
}
\remove{
$\parcssub$ is a theoretically efficient algorithm but suffers from the problem of generating cliques that are not actually subsumed (false positives). We propose an alternative algorithm $\parcssubnew$ that resolves the problem by carefully avoiding the enumeration of false positives. 
}

\subsection{Parallel Enumeration of New Maximal Cliques}
Here, we present a parallel algorithm $\parcsnewttt$ for enumerating the set of new maximal cliques when a set of edges is added to the graph. The idea is that we iterate over new edges in parallel and at each (parallel) iteration we construct a subgraph of the original graph and enumerate the set of all maximal cliques within the subgraph. We present an efficient parallel algorithm $\partomitaE$ (Algorithm~\ref{algo:partomitaext}) for this enumeration. The description of $\parcsnewttt$ is presented in Algorithm~\ref{algo:parcsnewttt}. 

\begin{figure}[htp!]
	\centering
	\begin{mdframed}[style=MyFrame]
		\begin{minipage}[t]{\textwidth}
			\begin{algorithm}[H]
				\DontPrintSemicolon
				\caption{$\partomitaE({G},K,\cand,\fini, \mathcal{E})$}
				\label{algo:partomitaext}
				\KwIn{${G}$ - the input graph \\  $K$ - Set of vertices forming a clique \\
					$\cand$ - a set of vertices that may extend $K$ \\ 
					$\fini$ - vertices that are not used to extend $K$, but are connected to all vertices of $K$ \\ 
					$\mathcal{E}$ - a set of edges to ignore}
				\If{$(\cand = \emptyset)$ \& $(\fini = \emptyset)$}{
					Output $K$ \tcp{K is a maximal clique}
					\Return\;
				}
				$\pivot \gets \parpivot({G},\cand,\fini)$\;
				$\ext[1..\kappa] \gets \cand - \Gamma_{{G}}(\pivot)$ \tcp{in parallel}
				\ForPar{$i \in$ $[1..\kappa]$}{
					$q\gets\ext[i]$\;
					$K_q \gets K\cup\{q\}$\;
					\If{$K_q \cap \mathcal{E} \neq \emptyset$}{
						\Return\;
					}
					$\cand_q \gets \intersect(\cand\setminus\ext[1..i-1], \Gamma_{{G}}(q))$\;
					$\fini_q \gets \intersect(\fini\cup\ext[1..i-1], \Gamma_{{G}}(q))$\;
					$\partomitaE({G},K_q,\cand_q,\fini_q,\mathcal{E})$\;
				}
			\end{algorithm}
		\end{minipage}
	\end{mdframed}
\end{figure}

Note that $\parcsnewttt$ is based upon an existing sequential algorithm $\csnewttt$~\cite{DST16} for enumerating new maximal cliques using $\tomitaE$ (Algorithm~\ref{algo:tomitaE})~\cite{DST16}, which lists all new maximal cliques without any duplication. $\partomitaE$ avoids the duplication in maximal clique enumeration similar to the technique used in $\tomitaE$ and uses a parallelization approach similar to of $\partomita$. More specifically, $\partomitaE$ follows a global ordering of the new edges to avoid redundancy in the enumeration process. Note that the correctness of $\parcsnewttt$ is followed by the correctness of the sequential algorithm $\csnewttt$. Following lemma shows the work efficiency and depth of $\parcsnewttt$.


\begin{lemma}\label{lem:incremental-new-work-efficiency}
Given a graph $G$ and a new edge set $H$, $\parcsnewttt$ is work-efficient, i.e, the total work is of the same order of the time complexity of $\csnewttt$. The depth of $\parcsnewttt$ is $O(\Delta^2 + M\log{\Delta})$, where $\Delta$ is the maximum degree and  $M$ is the size of a maximum clique in $G+H$.
\end{lemma}

We prove this lemma in Appendix~\ref{appendix-a} by showing the equivalence of the operations of $\csnewttt$ and $\parcsnewttt$.

\subsection{Parallel Enumeration of Subsumed Cliques}
\begin{figure}[htp!]
	\centering
	\begin{mdframed}[style=MyFrame]
		\begin{minipage}[t]{\textwidth}
			\begin{algorithm}[H]
				\DontPrintSemicolon
				\caption{$\parcssub(G,H,C,\Lambda^{new})$} 
				\label{algo:parcssub}
				\KwIn{$G$ - the input Graph \\ \hspace{1cm} $H$ - an edge set being added to $G$ \\ \hspace{1cm} $C$ - a set of maximal cliques in $G$ \\ \hspace{1cm} $\Lambda^{new}$ - a set of new maximal cliques in $G + H$}
				\KwOut{All cliques in $\Lambda^{del} = \cliques(G) \setminus \cliques(G+H)$}
				$\Lambda^{del} \gets \emptyset$\;
				\ForPar{$c \in \Lambda^{new}$}{
					$S\gets \{c\}$\;
					\For{$e = (u,v) \in E(c) \cap H$}{
						$S'\gets \phi$\;
						\ForPar{$c' \in S$}{
							\If{$e \in E(c')$}{
								$c_1 = c'\setminus\{u\}$ ; $c_2 = c'\setminus\{v\}$\;
								$S' \gets S'\cup c_1$ ; $S' \gets S'\cup c_2$\;
								
							}
							\Else{
								
								$S' \gets S'\cup c'$\;
							}
						}
						$S\gets S'$\;
					}
					\ForPar{$c' \in S$}{
						\If{$c' \in C$}{
							$\Lambda^{del} \gets \Lambda^{del} \cup c'$\;
							$C \gets C\setminus c'$\;
						}
					}
				}
			\end{algorithm}
		\end{minipage}
	\end{mdframed}
\end{figure}
In this section, we present a parallel algorithm $\parcssub$ based on the sequential algorithm $\cssub$~\cite{DST16} for enumerating subsumed cliques. In $\parcssub$, we perform parallelization in the following order: (1) Removing a single new edges from all the candidates in parallel and (2) Checking for the candidacy of the subsumed cliques in parallel. We present $\parcssub$ in Algorithm~\ref{algo:parcssub}. In the following lemma, we show the work efficiency and depth of $\parcssub$:

\begin{lemma}
Given a graph $G$ and a new edge set $H$, $\parcssub$ is work-efficient; the total work is of the same order of the time complexity of $\cssub$. The depth of $\parcsnewttt$ is $O(\mbox{min}\{M^2, \rho\})$ for processing each new maximal clique, where $M$ is the size of a maximum clique in $G+H$ and $\rho$ is the size of $H$.
\end{lemma}

\begin{proof}
First, note that the procedure of $\parcssub$ is exactly the same as the procedure of $\cssub$ except for the parallel loops at Line~6 and Line~13 of $\parcssub$ whereas these loops are sequential in $\cssub$. Since all the computations in $\parcssub$ is exactly the same as the computations in $\cssub$, except for the loop parallelization, $\parcssub$ is work-efficient.

For proving the parallel depth of $\parcssub$, first note that all the elements of $S$ at Line~6 of $\parcssub$ are processed in parallel, and the total cost of executing Lines~7 to 11 is $O(1)$. In addition, the depth of the operation at Line~12 of $\parcssub$ is $O(1)$ using the concurrent hashtable. Therefore, the overall depth of the procedures from Line~4 to Line~12 is the number of new edges in a new maximal clique $c$, processed at Line~2 of $\parcssub$ which is $O(\mbox{min}\{M^2, \rho\})$. Next, the depth of Lines~14 to 16 is $O(1)$ because it only takes $O(1)$ to execute Lines~15 and 16 using Theorem~\ref{thm:parset}. As a result, for each new maximal clique, the depth of the procedure for enumerating all subsumed cliques within the new maximal clique is $O(\mbox{min}\{M^2, \rho\})$.
\end{proof}

\subsection{Decremental Case}
We have so far discussed incremental maintenance of new and subsumed maximal cliques when the stream only contains new edges. We note that our algorithm can also support the deletion of edges through a reduction to the incremental case, this is similar to the methods used for sequential algorithms. Please see Sections 4.4 and 4.5 of \cite{DST16} for further details.

\remove{
\begin{algorithm}[htp!]
\DontPrintSemicolon
\caption{$\parcssubnew(G, H, C)$}
\label{algo:parcssuba}
\KwIn{$G$ - input graph \\ \hspace{1cm} $H$ - Set of $\rho$ edges being added to $G$ \\ \hspace{1cm} $C$ - Set of maximal cliques in $G$}
\KwOut{Cliques in $\Lambda^{del} = \cliques(G)\setminus\cliques(G+H)$}
Consider edges of $H$ in an arbitrary order $e_1, e_2,\ldots,e_{\rho}$\;

\ForPar{$i \gets 1,2, \ldots,\rho$}{
	$e \gets e_i = (u,v)$\;
	
	$V_e \gets \Gamma_{G+H}(u)\cap\Gamma_{G+H}(v)$\;
	$K \gets \emptyset$\;
	$\cand \gets V_e\cup\{u,v\}$ ; 
	
	$X \gets \bigcup_{u\in\cand}\Gamma_{G}(u)$\;
	
	$\fini \gets (\Gamma_{G}(u)\cup\Gamma_{G}(v)\cup X)\setminus\cand$\;
	$\Lambda^{del}\gets \Lambda^{del} \cup \partomita(G, K, \cand, \fini)$\;
	%
}
\end{algorithm}
}

\section{Evaluation}
\label{sec:exp}
In this section, we experimentally evaluate the performance of our shared-memory parallel static ($\partomita$ and $\parmce$) and dynamic ($\parimce$) algorithms for $\mce$ on static and dynamic (real world and synthetic) graphs to show the parallel speedup and scalability of our algorithms over efficient sequential algorithms $\tomita$ and $\imce$ respectively. We also compare our algorithms with state-of-the-art parallel algorithms for MCE to show that the performance of our algorithm has substantially improved over the prior works. We run the experiments on a computer configured with Intel Xeon (R) CPU E5-4620 running at 2.20GHz , with $32$ physical cores (4 NUMA nodes each with $8$ cores) and $1$TB RAM.



\subsection{Datasets}
We use eight different real-world static and dynamic networks from publicly available repositories KONECT~\cite{konnect13}, SNAP~\cite{JA14}, and Network Repository~\cite{RN15} for doing the experiments. Dataset statistics are summarized in Table~\ref{graph:summary}. For our experiments, we convert these networks to simple undirected graphs by removing self-loops, edge weights, parallel edges, and edge directions.

We consider networks $\dblp$, $\skitter$, $\wikipedia$, $\wtalk$, and $\orkut$ for the evaluation of the algorithms on static graphs and networks $\dblp$, $\flickr$, $\wikipedia$, $\journal$, and $\cacit$ for the evaluation of the algorithms on dynamic graphs.


$\dblp$ shows the collaboration of authors of papers from DBLP computer science bibliography. In this graph, vertices represent authors, and there is an edge between two authors if they have published a paper~\cite{RN15}. $\skitter$ is an Internet topology graph which represents the autonomous systems, connected to each other on the Internet~\cite{konnect13}. $\wikipedia$ is a network of English Wikipedia in 2013, where vertices represent pages in English Wikipedia, and there is an edge between two pages $p$ and $q$ if there is a hyperlink in page $p$ to page $q$~\cite{JA14}. $\wtalk$ contains users of Wikipedia as vertices where each edge between two users in this graph indicates that one of the users has edited the "page talk" of the other user on Wikipedia~\cite{JA14}. $\orkut$ is a social network where each vertex represents a user in the network, and there is an edge if there is a friendship relation between two users~\cite{konnect13}. Similar to $\orkut$, $\flickr$ and $\journal$ are also social networks where a vertex represents a user and an edge represents the friendship between two users. $\cacit$ is a citation network of high energy physics theory where a vertex represents a paper and there is an edge between paper $u$ and paper $v$ if $u$ cites $v$.

For the evaluation of $\partomita$ and $\parmce$, we give the entire graph as input to the algorithm in the form of an edge list and for the evaluation of the algorithms on dynamic graphs, we start with an empty graph that contains all vertices but no edges, and, at a time, we add a set of edges in an increasing order of timestamps for computing the changes in the set of maximal cliques. For $\parimce$, we use real dynamic graphs, where each edge has a timestamp. In addition, we evaluate $\parimce$ on $\journal$, which is a static graph. We convert this graph to a dynamic graph through randomly permuting the edges and processing edges in that ordering. 

\begin{figure} [htp!]
	\centering
	\begin{tabular}{cccc}
		\includegraphics[width=.3\textwidth]{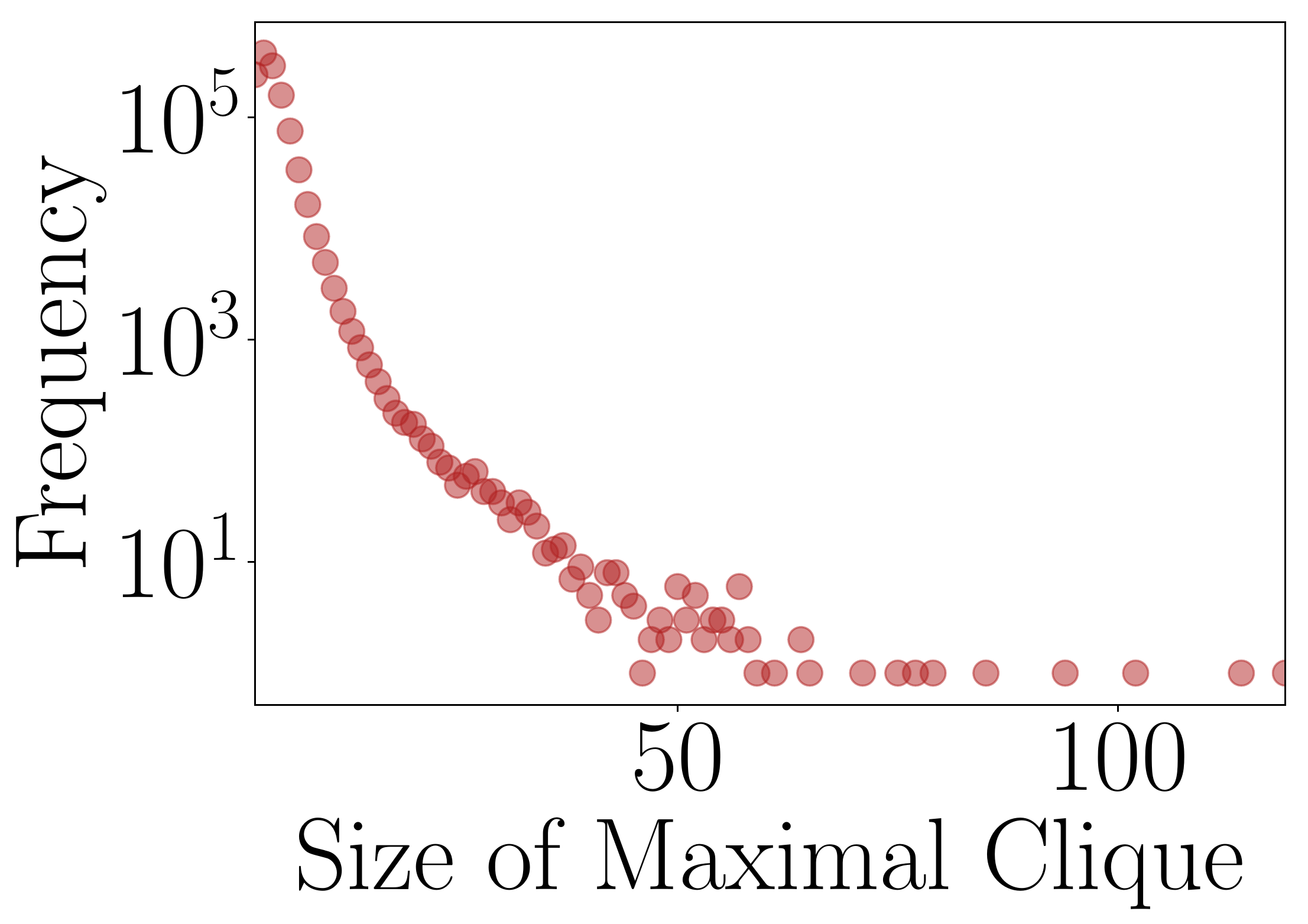} &
		\includegraphics[width=0.3\textwidth]{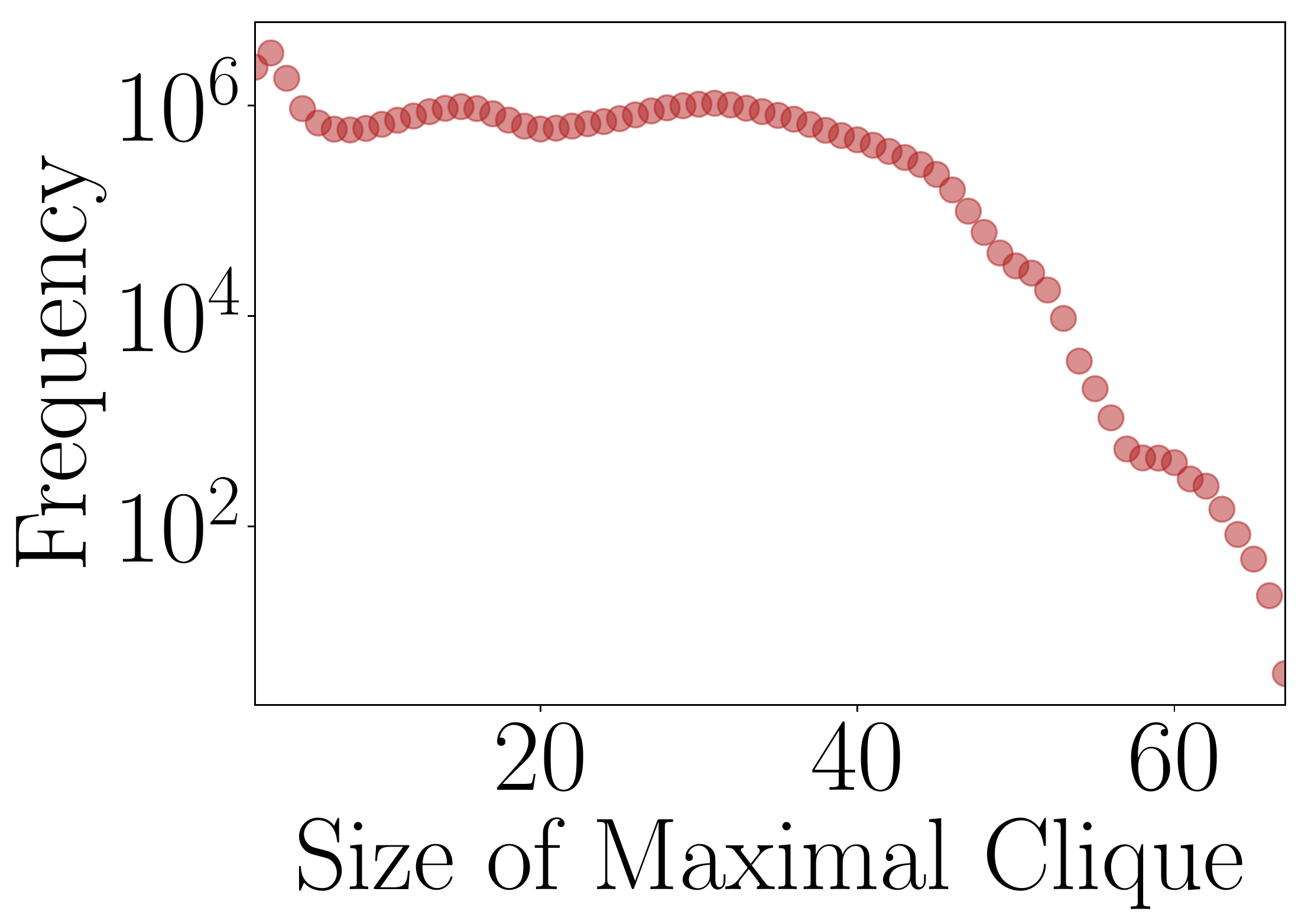} &
		\includegraphics[width=0.3\textwidth]{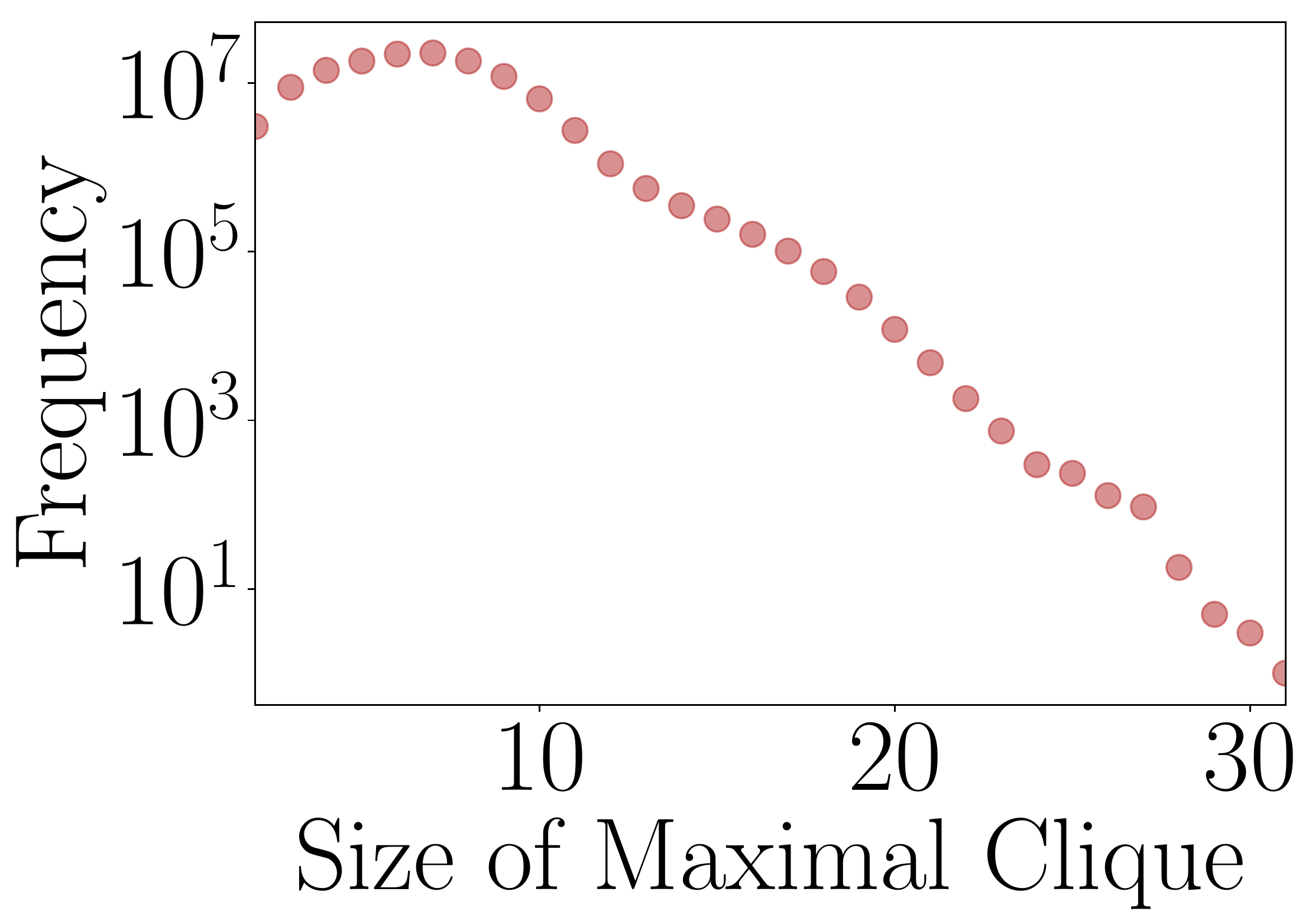} \\
		\textbf{(a) \dblp}  & \textbf{(b) \skitter} & \textbf{(c) \wikipedia}  \\[6pt]
	\end{tabular}
	\begin{tabular}{cccc}
		\includegraphics[width=0.3\textwidth]{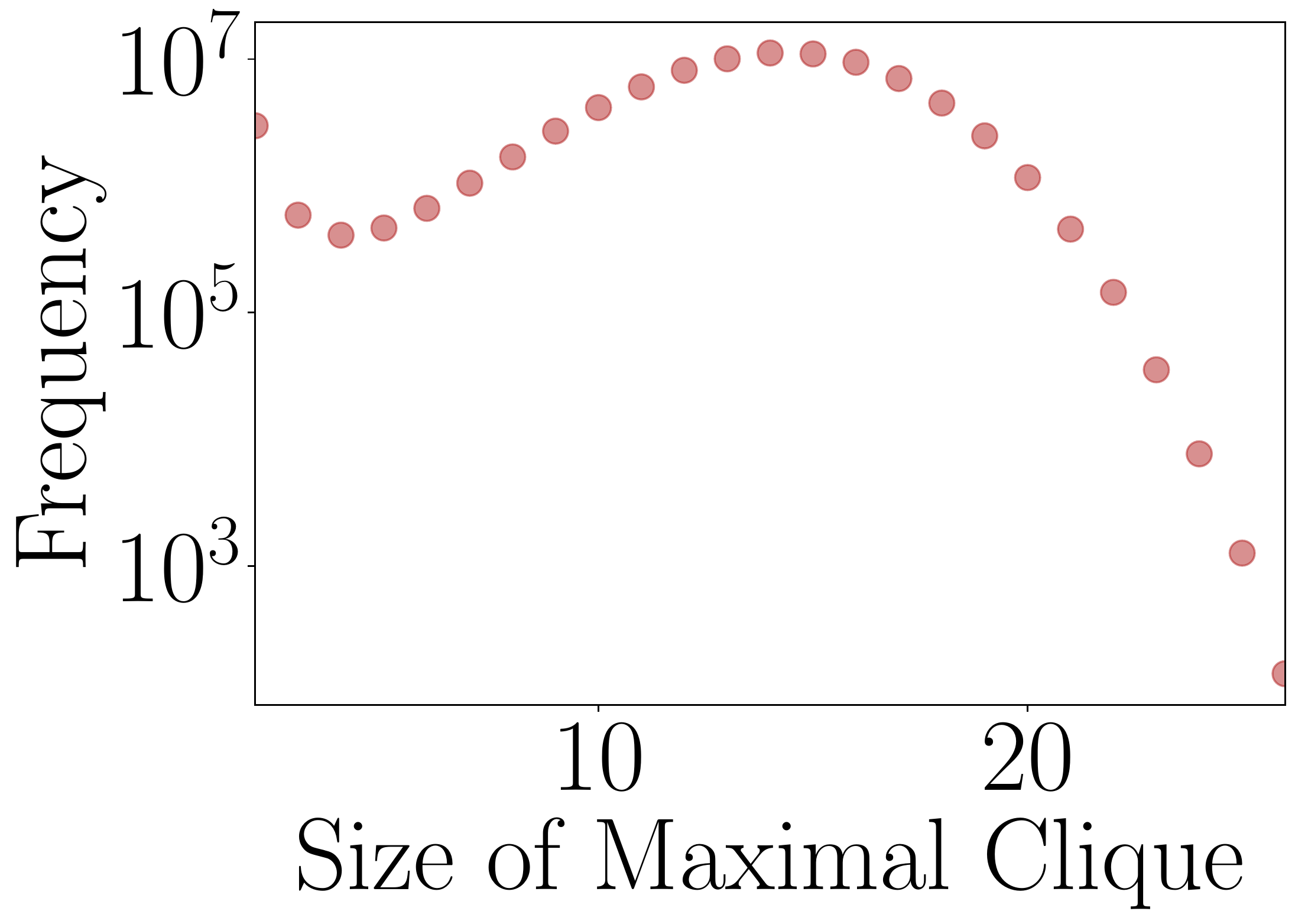} &
		\includegraphics[width=0.3\textwidth]{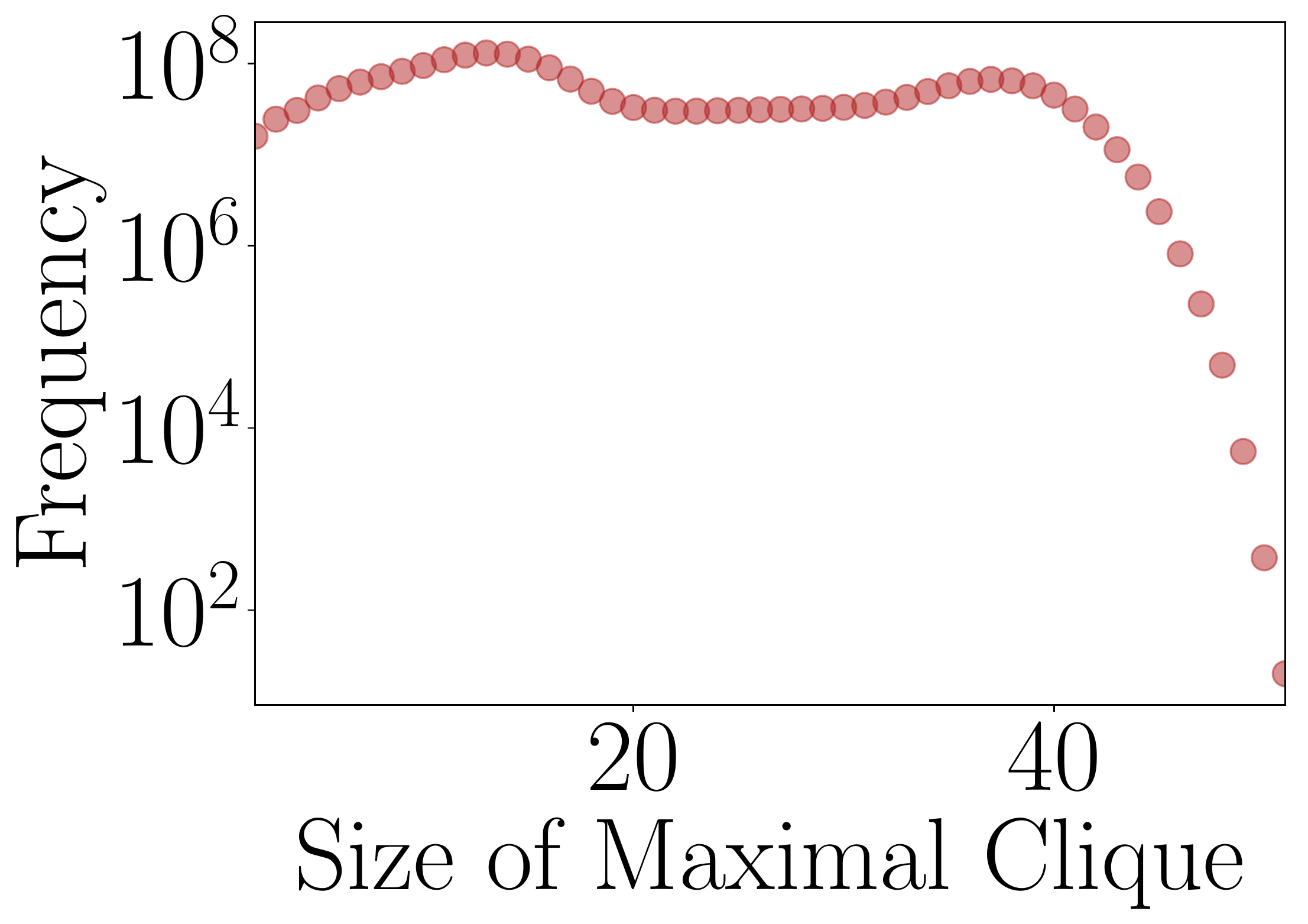} \\
		\textbf{(d) \wtalk}  & \textbf{(e) \orkut}  \\[6pt]
	\end{tabular}
	\caption{\textbf{Frequency distribution of sizes of maximal cliques across different input graphs.}}
	\label{fig:szMC-distribution}
\end{figure}
To understand the datasets better, we illustrated the frequency distribution of sizes of maximal cliques in~\cref{fig:szMC-distribution}. The size of maximal cliques in $\dblp$ can be as large as $100$ vertices. Although the number of such large maximal cliques are small, the depth of the search space might increase exponentially for discovering such large maximal cliques. As shown in \cref{fig:szMC-distribution}, most of the graphs contain more than tens of millions of maximal cliques. For example, $\orkut$ contains more than two billion maximal cliques. That being said, the depth and breadth of the search space makes MCE a challenging problem to solve.

\begin{table}[b!]
	\caption{Static and Dynamic Networks, used for evaluation, and their properties\label{graph:summary}. For some of the graphs (e.g. $\flickr$ and $\cacit$) used for evaluating the incremental algorithms, we could not report the information about maximal cliques as they did not finish within 10 days, even using parallel algorithms.}
	\resizebox{\textwidth}{!}{%
		\begin{tabular}{lrrrcc}
			\textbf{Dataset} & \multicolumn{1}{c}{\textbf{\#Vertices}} & \multicolumn{1}{c}{\textbf{\#Edges}} & \multicolumn{1}{c}{\textbf{\#Maximal Cliques}} & \textbf{\begin{tabular}[c]{@{}c@{}}Average Size of \\ Maximal Cliques\end{tabular}} & \textbf{\begin{tabular}[c]{@{}c@{}}Size of \\ Largest Clique\end{tabular}} \Tstrut\Tstrut\\ \hline
			$\dblp$ & 1,282,468 & 5,179,996 & 1,219,320 & 3 & 119 \\
			$\orkut$ & 3,072,441 & 117,184,899 & 2,270,456,447 & 20 & 51 \\
			$\skitter$ & 1,696,415 & 11,095,298 & 37,322,355 & 19 & 67 \\
			$\wtalk$ & 2,394,385 & 4,659,565 & 86,333,306 & 13 & 26 \\
			$\wikipedia$ & 1,870,709 & 36,532,531 & 131,652,971 & 6 & 31 \\
			$\journal$ & 4,033,137 & 27,933,062 & 38,413,665 & 29 & 214 \\
			$\flickr$ & 2,302,925 & 22,838,276 & \multicolumn{1}{r}{> 400 billion} & - & - \\
			$\cacit$ & 22,908 & 2,444,798 & \multicolumn{1}{r}{> 400 billion} & - & -
		\end{tabular}%
	}
\end{table}


\subsection{Implementation of the Algorithms}
In our parallel implementations of $\partomita$, $\parmce$, and $\parimce$, we utilize {\small\textbf{\texttt{parallel\_for}}} and {\small\textbf{\texttt{parallel\_for\_each}}}, developed by the Intel TBB parallel library~\cite{RL99}. We also utilize {\small\textbf{\texttt{concurrent\_hash\_map}}} for atomic operations on hashtable. We use C++11 standard for the implementation of the algorithms and compile all the sources using Intel ICC compiler version 18.0.3 with optimization level `-O3'. We use the command `numactl -i all' for balancing the memory in a NUMA machine. System level load balancing is performed using a dynamic work stealing scheduler~\cite{RL99} in TBB.



To compare with prior works on MCE, we implement some of them such as \cite{TTT06,ELS10,DW+09,SMT+15,ZA+05} in C++, and we use the executable of the C++ implementations for the rest of the algorithms ($\greedybb$~\cite{SAS18}, $\hashing$~\cite{LP+17}, and $\gp$~\cite{WC+17}), provided by the respective authors. See Subsection \ref{sec:exp-prior} for more details. 

We compute the degeneracy number and triangle count for each vertex using sequential procedures. While the computation of per-vertex triangle counts and the degeneracy ordering could be potentially parallelized, implementing a parallel method to rank vertices based on their degeneracy number or triangle count is itself a non-trivial task. We decided not to parallelize these routines since the degeneracy- and triangle-based ordering did not yield significant benefits when compared with degree-based ordering, whereas the degree-based ordering is trivially available, without any additional computation. 

We assume that the entire graph is available in a shared global memory. The Total Runtime (TR) of $\parmce$ consists of (1)~Ranking Time (RT): the time required to rank vertices of the graph based on the ranking metric used in the algorithm, i.e. degree, degeneracy number, or triangle count of vertices and (2)~Enumeration Time (ET): the time required to enumerate all maximal cliques. For $\parmcedegen$ and $\parmcetriangle$ algorithms, the runtime of ranking  is also reported. Figures~\ref{fig:static-par-speedup} and~\ref{fig:static-par-runtime} show the parallel speedup (with respect to the runtime of $\tomita$) and the enumeration time of $\parmce$ using different vertex ordering strategies. Table~\ref{result:runtime-splitup} shows the breakdown of the Total Runtime (TR) into Ranking Time (RT) and Enumeration Time (ET). 

The runtime of $\parimce$ consists of (1)~the computation time of $\parcsnewttt$ and (2)~the computation time of $\parcssub$. In the implementation of $\parcsnewttt$, we follow the design of $\parmce$ and instead of executing $\partomitaE$ on the entire (sub)graph for an edge as in Line~9 of $\parcsnewttt$, we run $\partomitaE$ on per-vertex sub-problems in parallel. For dealing with load balance, we use degree based ordering of the vertices of ${G}''$ in creating the sub-problems. We decided to implement $\parcsnewttt$ in this way as we observed a significant improvement in the performance of $\parmce$ over $\partomita$ when we used a degree-based vertex ordering. For experiment on dynamic graphs, we use the batch size of $1000$ edges for all the graphs except for an extremely dense graph $\cacit$ (with original graph density 0.01) where we use batch size of $10$.


\subsection{Discussion of the Results}
Here, we empirically evaluate  our parallel algorithms and prior methods. First, we show the parallel speedup (with respect to the sequential algorithms) and scalability (with respect to the number of cores) of our algorithms. Next, we compare our works with the state-of-the-art sequential and parallel algorithms for MCE. The results demonstrate that our solutions are faster than prior methods with significant margins.

\subsubsection{Parallel MCE on Static Graphs}

The total runtime of the parallel algorithms with $32$ threads are shown in Table~\ref{result:static-parallel-runtime}. We observe that $\partomita$ achieves a speedup of \textbf{5x}-\textbf{14x} over the sequential algorithm $\tomita$. The three versions of $\parmce$, i.e. $\parmcedegree$, $\parmcedegen$, $\parmcetriangle$, achieve a speedup of \textbf{15x}-\textbf{21x} with 32 threads, when we consider the runtime of enumeration task alone. The speedup ratio decreases in $\parmcedegen$ and $\parmcetriangle$ if we include the time taken by ranking strategies (See Figure~\ref{fig:static-par-speedup}).

The runtime of $\partomita$ is higher than the runtime of $\parmce$, due to a heavier cumulative overhead of pivot computation and of processing sets $\cand$ and $\fini$ in $\partomita$. For example, in $\dblp$ graph, when we run $\partomita$, the cumulative overhead of computing $\pivot$ is 248 seconds, and the cumulative overhead of updating $\cand$ and $\fini$ is 38 seconds while, in $\parmce$, these runtimes are 156 and 21 seconds, respectively. This helps $\parmce$ accelerate the entire process, which achieves 2x speedup over $\partomita$.
\remove{
\begin{table}[htp!]
\caption{\textbf{Runtime (in sec.) of $\parmce$ with vertex orderings (excluding the time taken by vertex ordering), $\partomita$ (with $32$ threads), and $\tomita$.}}
\label{result:static-parallel-runtime-removed}
\begin{minipage}{\columnwidth}
\centering
\scalebox{0.7}{
\begin{tabular}{| l | r | r | r | r | r |}
\toprule
\textbf{Dataset} & \tomita & \partomita & \parmcedegree & \parmcedegen & \parmcetriangle \\
\midrule
$\dblp$ & 42.4 & 4.6 & 2.6 & 3.1 & 2.9\\
\hline
$\orkut$ & \num{28923} & \num{3472.6} & \num{1676.4} & \num{2350} & \num{1959.3}\\
\hline
$\skitter$ & 660  & 68 & 39.2 & 42.8 & 48.2\\
\hline
$\wtalk$ & \num{961} & \num{109.4} & 51.6 & 77.8 & 57.6\\
\hline
$\wikipedia$ & 2646.6 & 160 & 123.3 & 155.3 & 178.8\\
\bottomrule
\end{tabular}}
\end{minipage}
\end{table}
}
\begin{table}[b!]
	\caption{Runtime (in sec.) of $\tomita$, $\partomita$, and $\parmce$ with different vertex orderings on 32 cores. The numbers exclude the time taken for vertex ordering. Note that the best algorithm, which uses degree based vertex ordering, has zero additional cost for computing the vertex ordering.}
	\label{result:static-parallel-runtime}
	\resizebox{0.85\textwidth}{!}{%
		\begin{tabular}{lrrrrr}
			\textbf{Dataset} & \multicolumn{1}{c}{$\tomita$} & \multicolumn{1}{c}{$\partomita$} & \multicolumn{1}{c}{$\parmcedegree$} & \multicolumn{1}{c}{$\parmcedegen$} & \multicolumn{1}{c}{$\parmcetriangle$} \\ \hline
			$\dblp$ & 42 & 4 & 2 & 3 & 3 \\
			$\orkut$ & 28923 & 3472 & 1676 & 2350 & 1959 \\
			$\skitter$ & 660 & 68 & 39 & 43 & 48 \\
			$\wtalk$ & 961 & 109 & 52 & 78 & 58 \\
			$\wikipedia$ & 2646 & 160 & 123 & 155 & 179
		\end{tabular}%
	}
\end{table}
\begin{figure} [t!]
	\centering
	\includegraphics[width=0.8\textwidth]{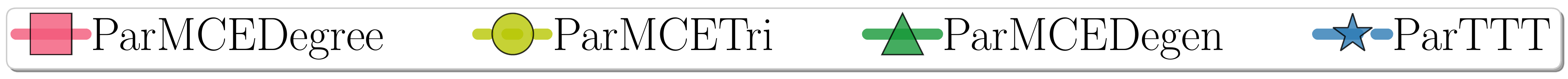}
	\begin{tabular}{cccc}
		\includegraphics[width=.3\textwidth]{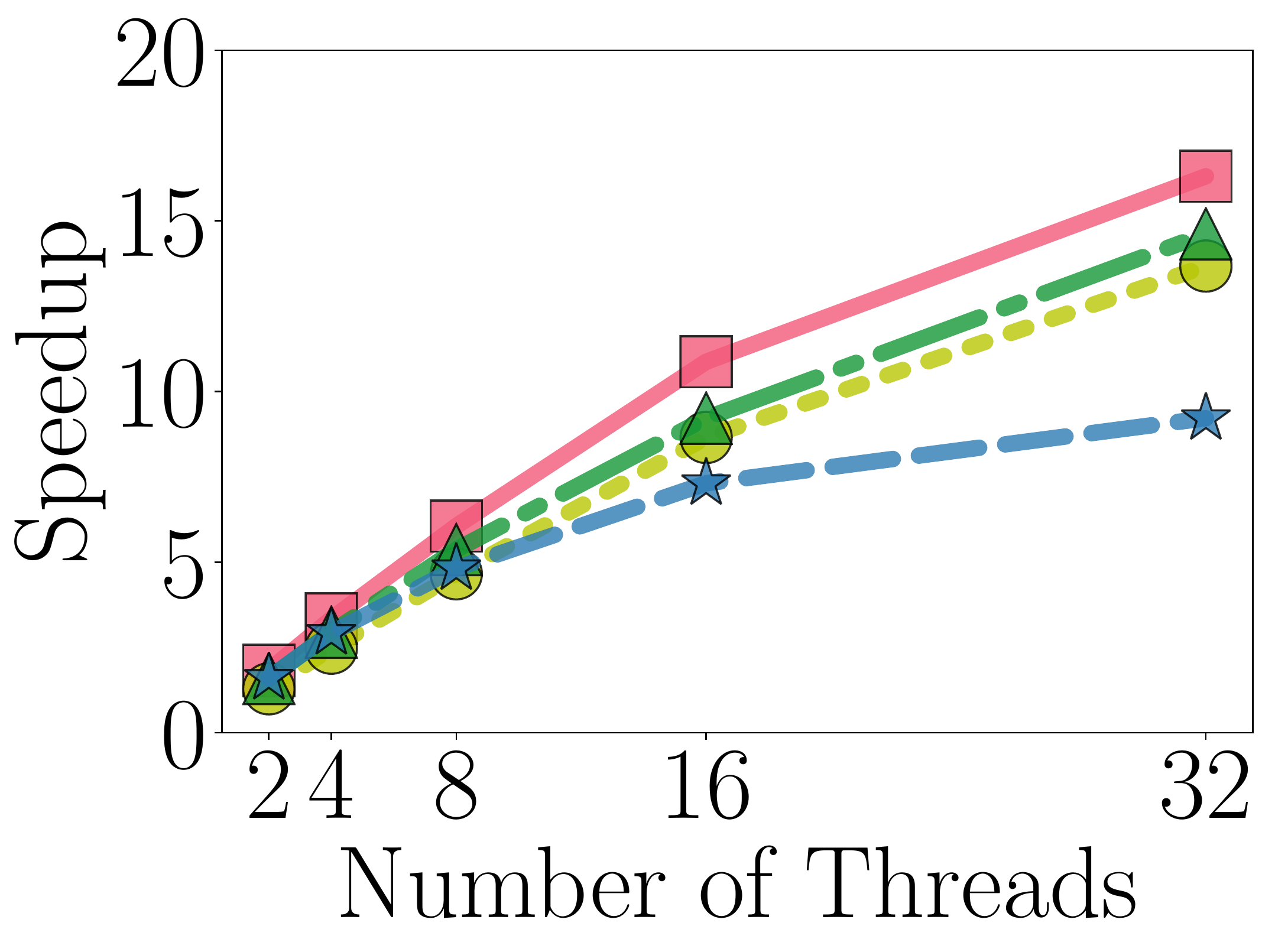} &
		\includegraphics[width=0.3\textwidth]{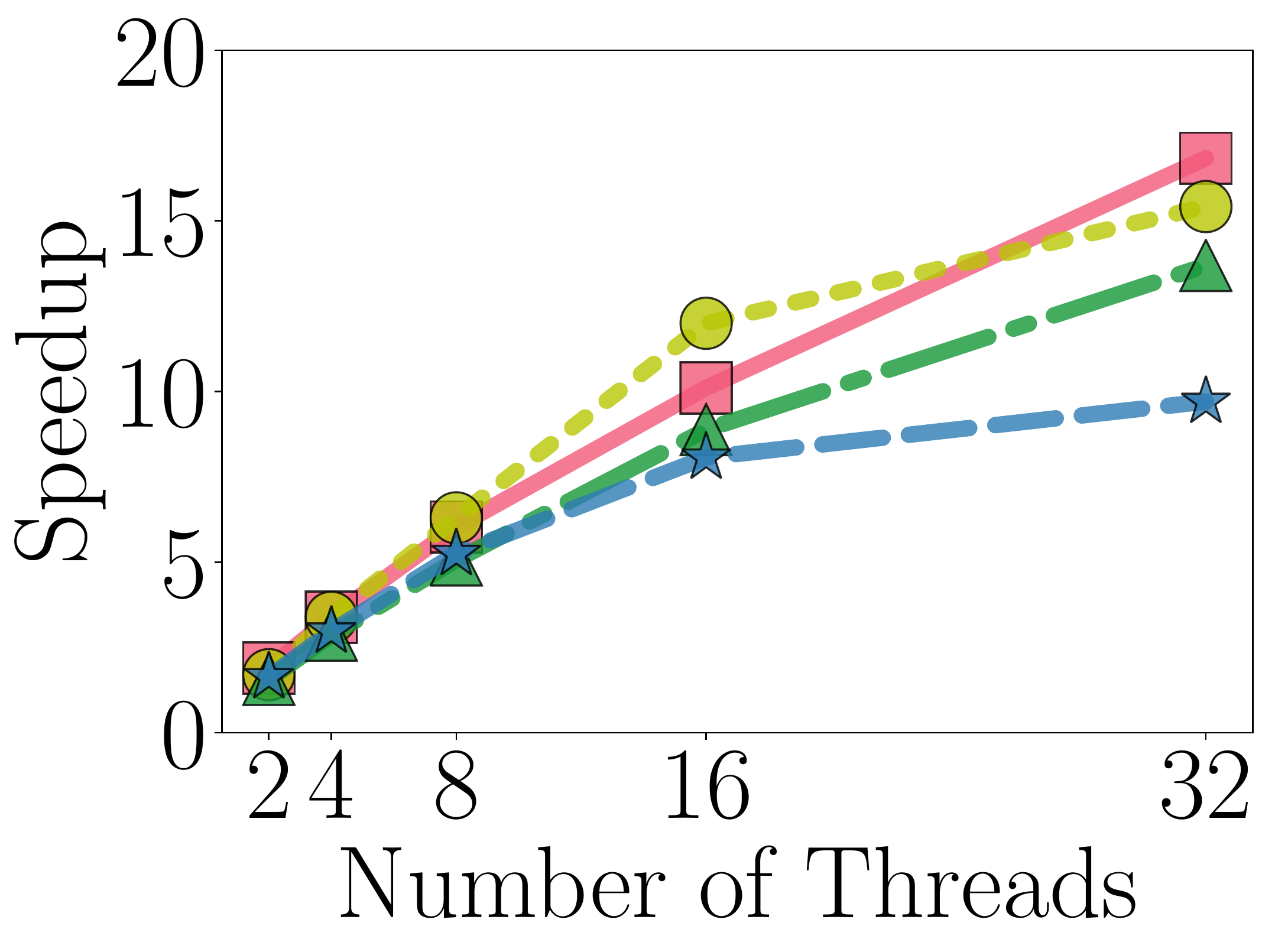} &
		\includegraphics[width=0.3\textwidth]{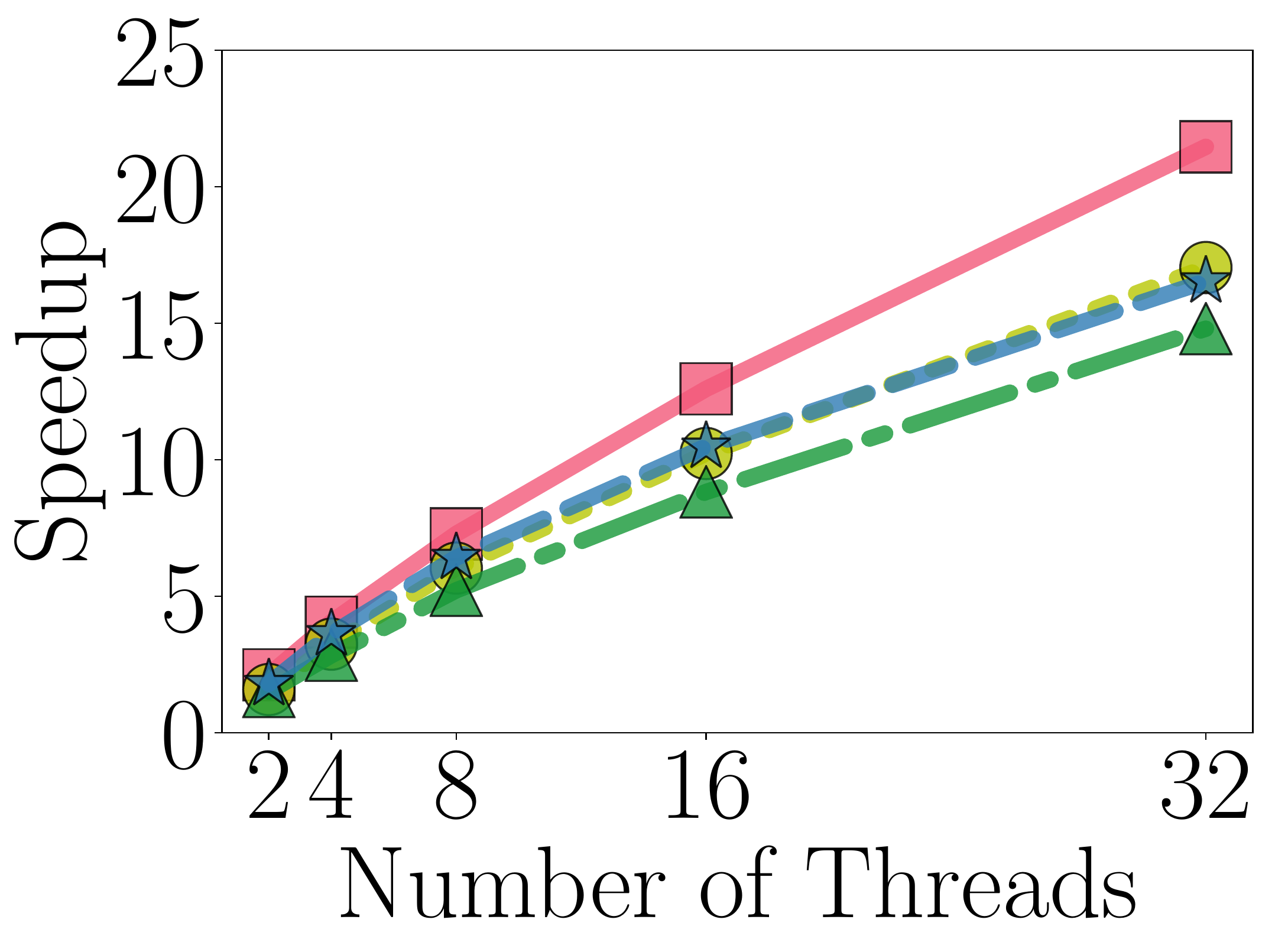} \\
		\textbf{(a) $\dblp$}  & \textbf{(b) $\skitter$} & \textbf{(c) $\wikipedia$}  \\[6pt]
	\end{tabular}
	\begin{tabular}{cccc}
		\includegraphics[width=0.3\textwidth]{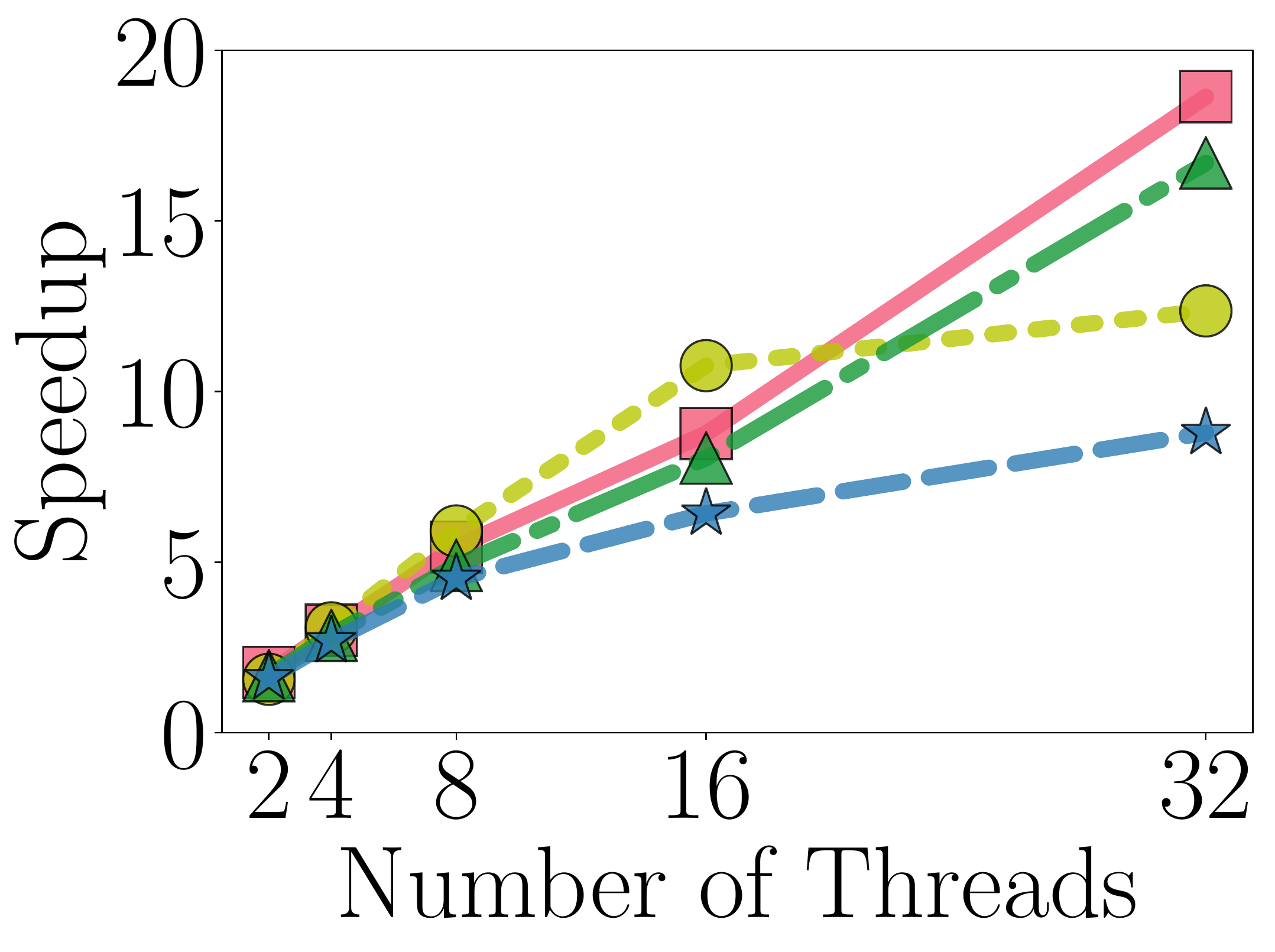} &
		\includegraphics[width=0.3\textwidth]{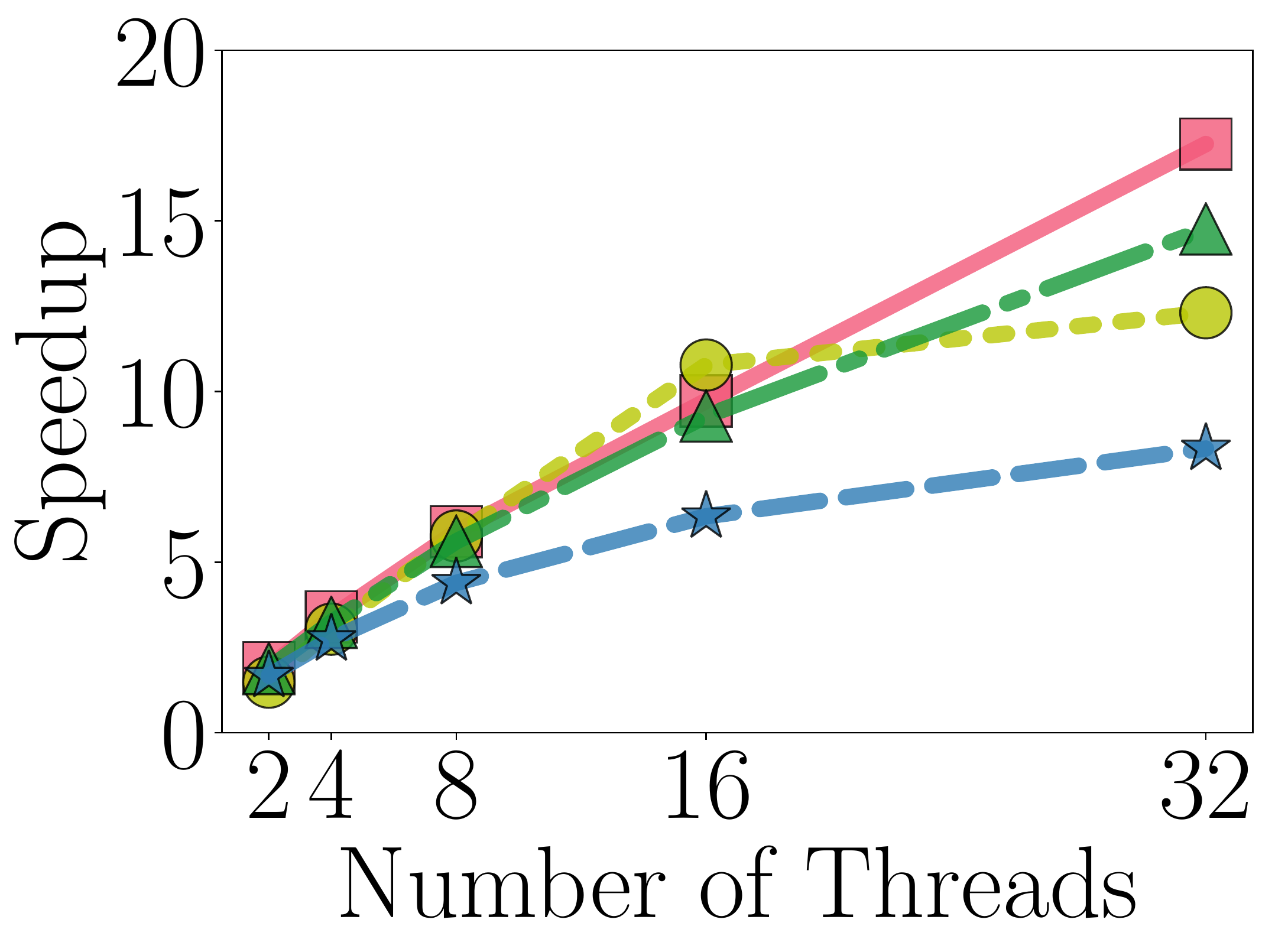} \\
		\textbf{(d) $\wtalk$}  & \textbf{(e) $\orkut$}  \\[6pt]
	\end{tabular}
\vspace{-2ex}
	\caption{Parallel speedup when compared with $\tomita$ (sequential algo. due to Tomita et al.~\cite{TTT06}) as a function of the number of threads.} 
	\label{fig:static-par-speedup}
\end{figure}
\begin{figure}[htp!]
	\centering
	\includegraphics[width=0.8\textwidth]{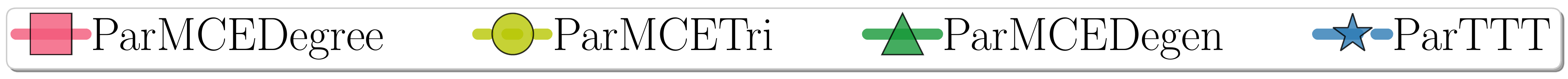}
	\begin{tabular}{cccc}
		\includegraphics[width=.3\textwidth]{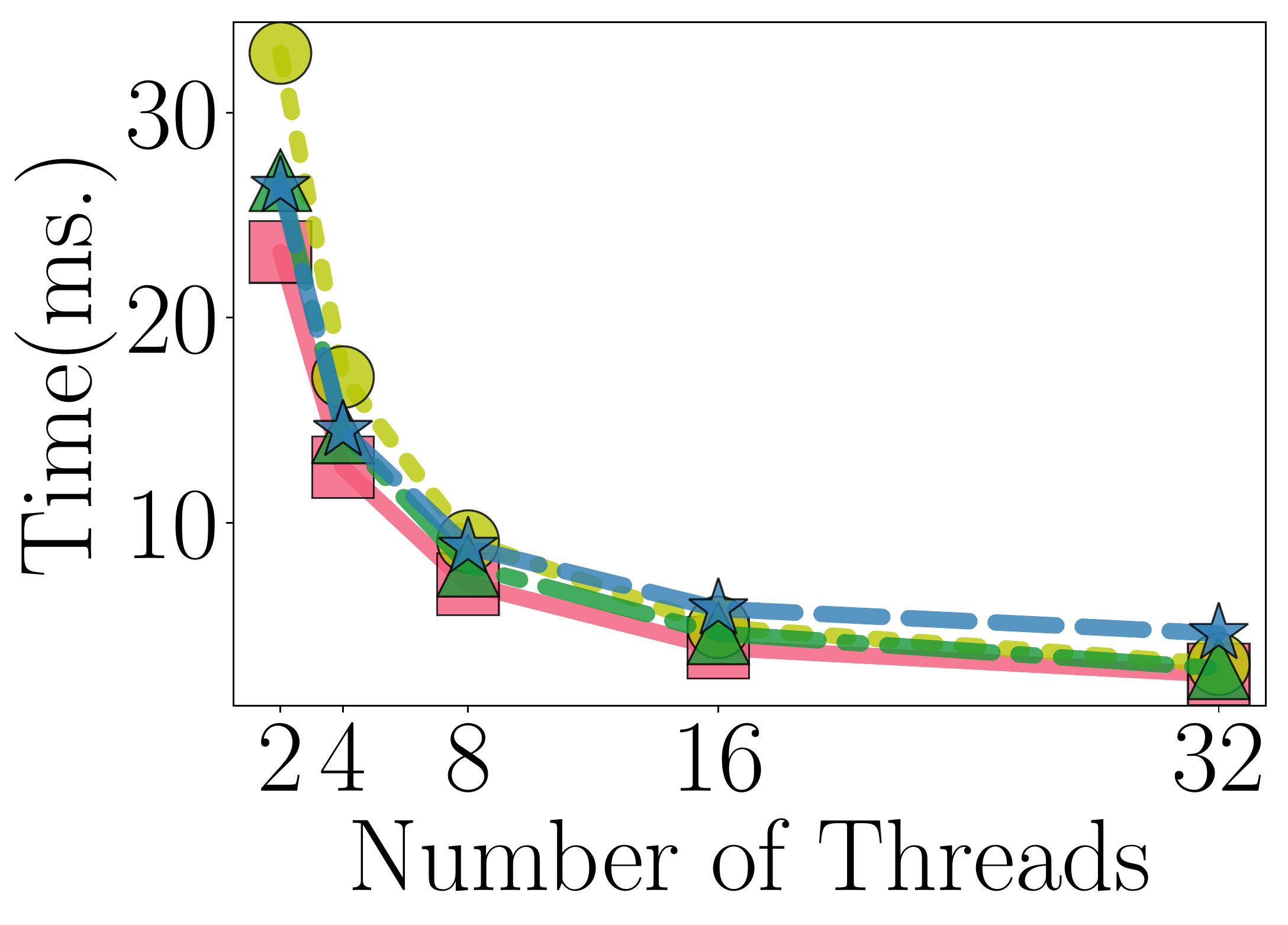} &
		\includegraphics[width=0.3\textwidth]{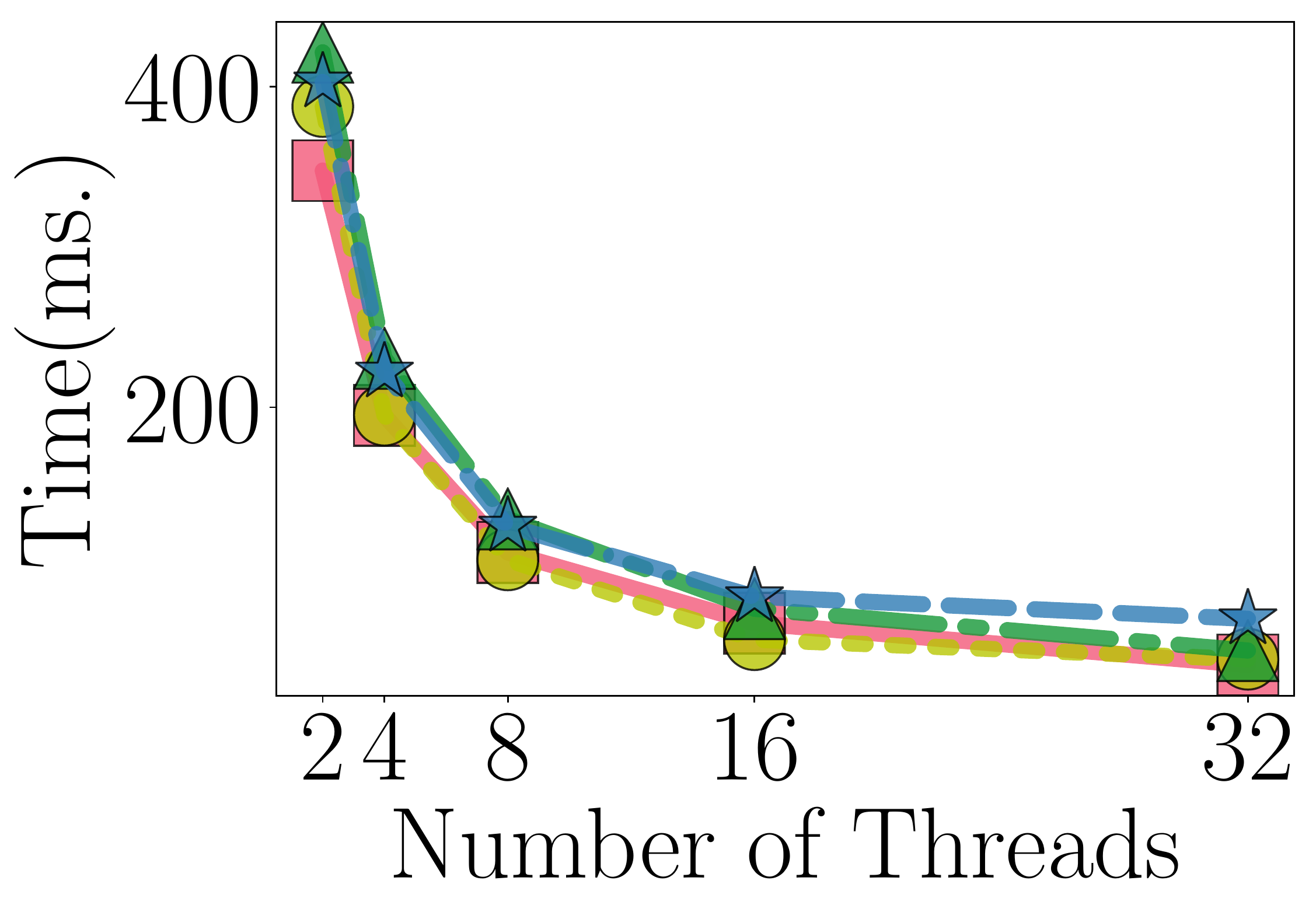} &
		\includegraphics[width=0.3\textwidth]{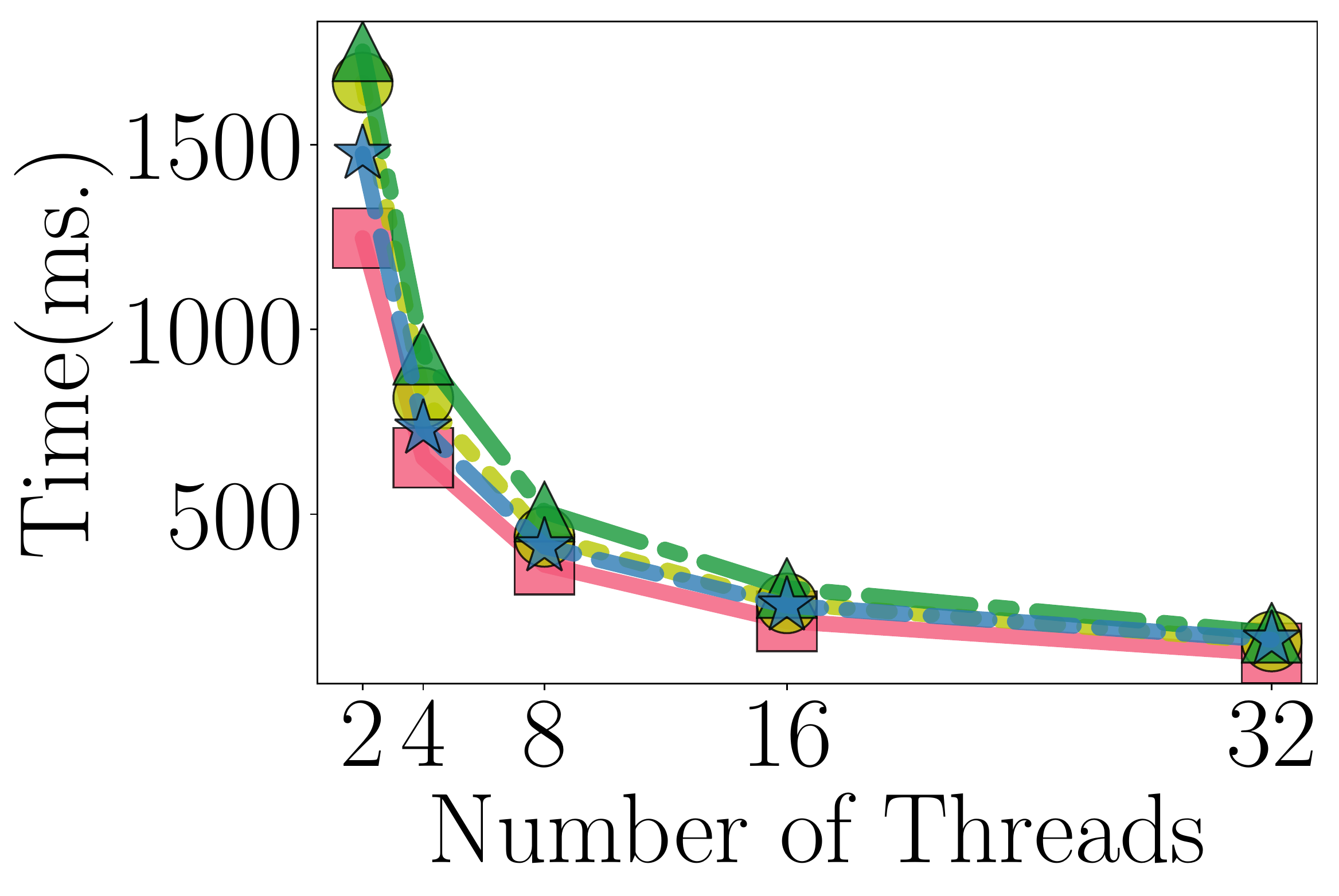} \\
		\textbf{(a) $\dblp$}  & \textbf{(b) $\skitter$} & \textbf{(c) $\wikipedia$}  \\[6pt]
	\end{tabular}
	\begin{tabular}{cccc}
		\includegraphics[width=0.3\textwidth]{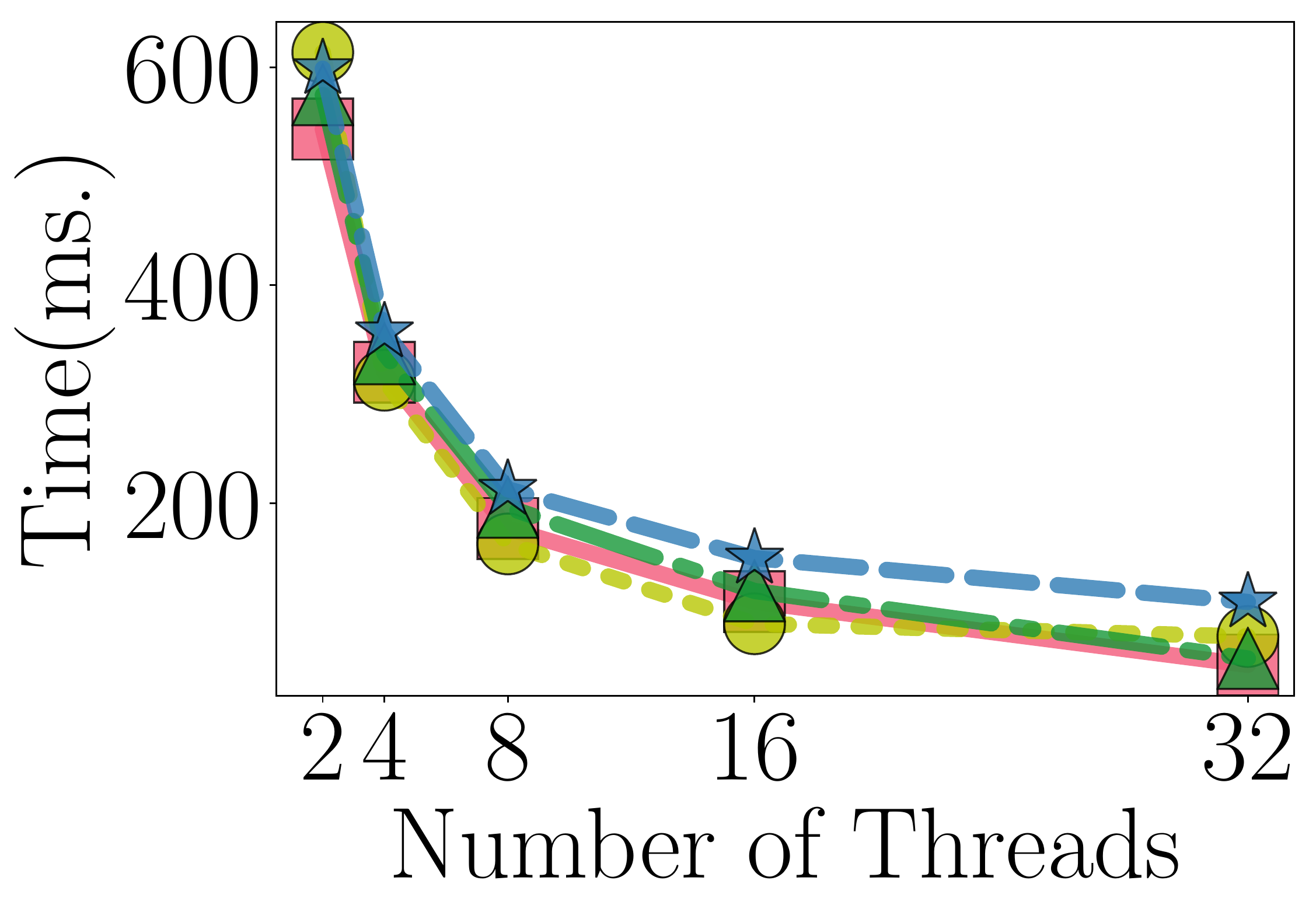} &
		\includegraphics[width=0.3\textwidth]{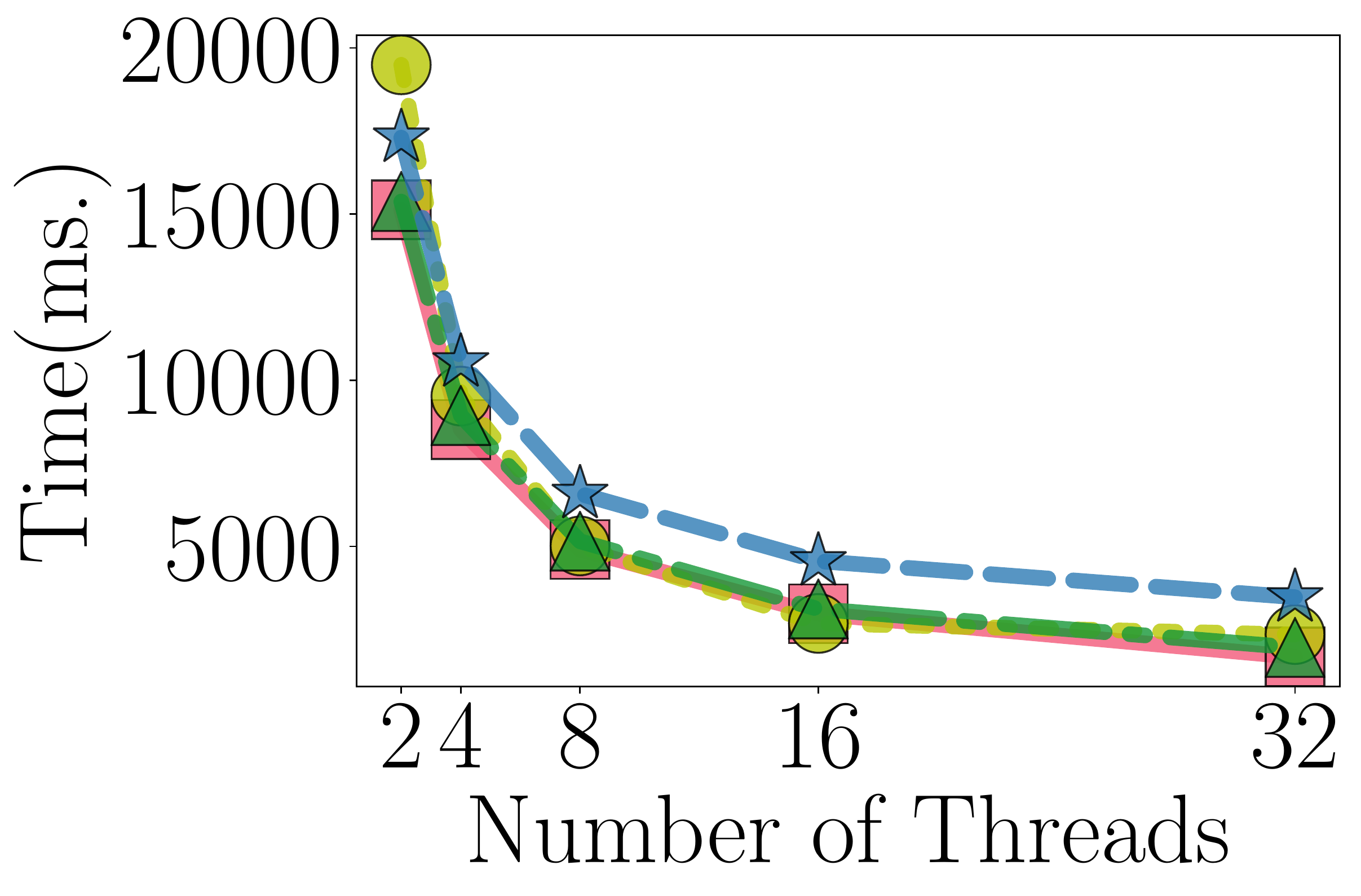} \\
		\textbf{(d) $\wtalk$}  & \textbf{(e) $\orkut$}  \\[6pt]
	\end{tabular}
\vspace{-2ex}
	\caption{
			The total runtime of parallel algorithms $\parmcedegree$, $\parmcetriangle$, $\parmcedegen$, and $\partomita$ in milliseconds as a function of the number of threads.
	}
	\label{fig:static-par-runtime}
\end{figure}
\vspace{-3ex}
\begin{figure}[htp!]
	\centering
	\includegraphics[width=0.8\textwidth]{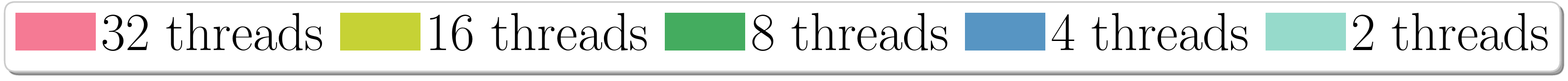}
	\begin{tabular}{cccc}
		\includegraphics[width=.3\textwidth]{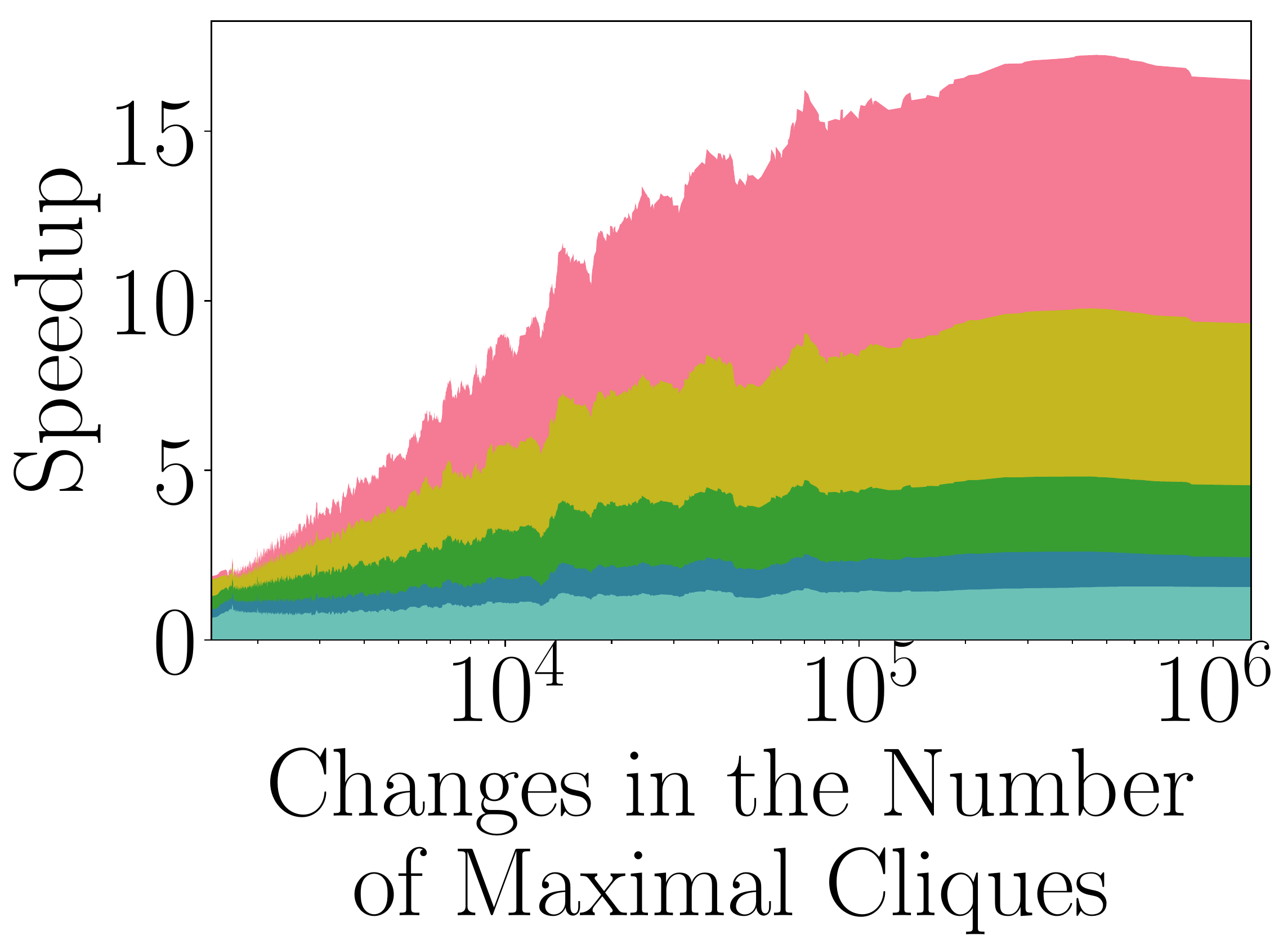} &
		\includegraphics[width=0.3\textwidth]{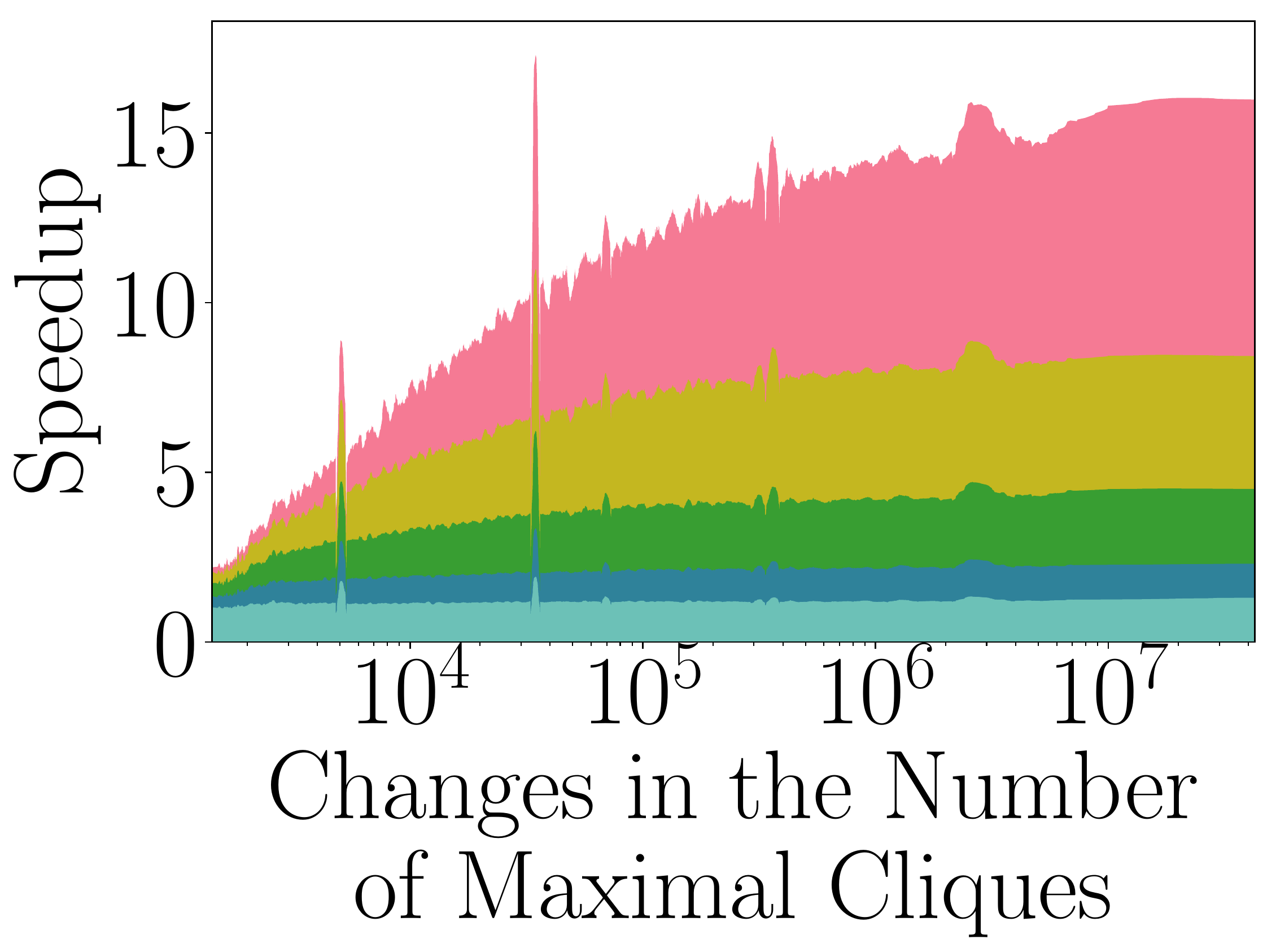} &
		\includegraphics[width=0.3\textwidth]{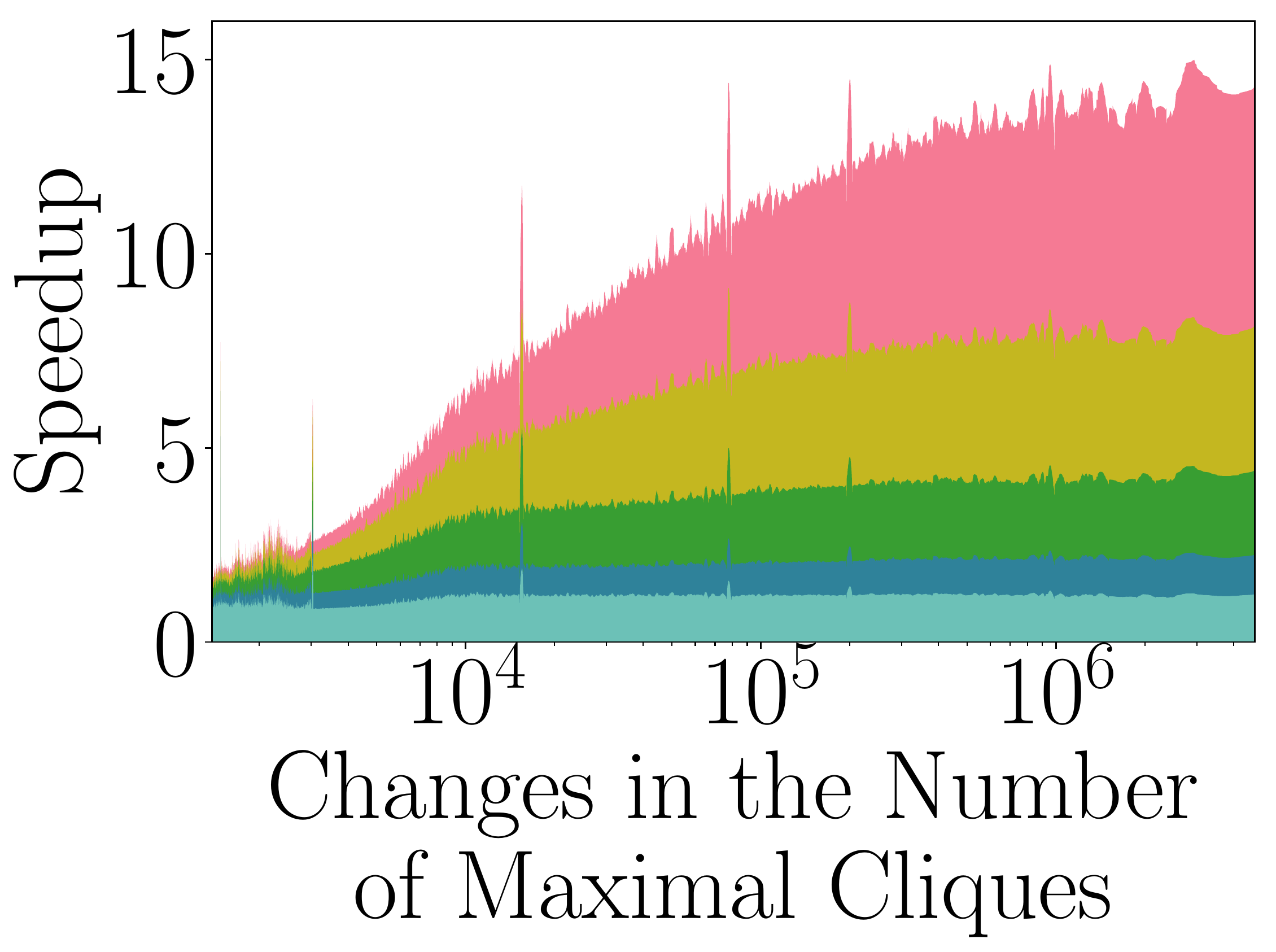} \\
		\textbf{(a) $\dblp$}  & \textbf{(b) $\flickr$} & \textbf{(c) $\journal$}  \\[6pt]
	\end{tabular}
	\begin{tabular}{cccc}
		\includegraphics[width=0.3\textwidth]{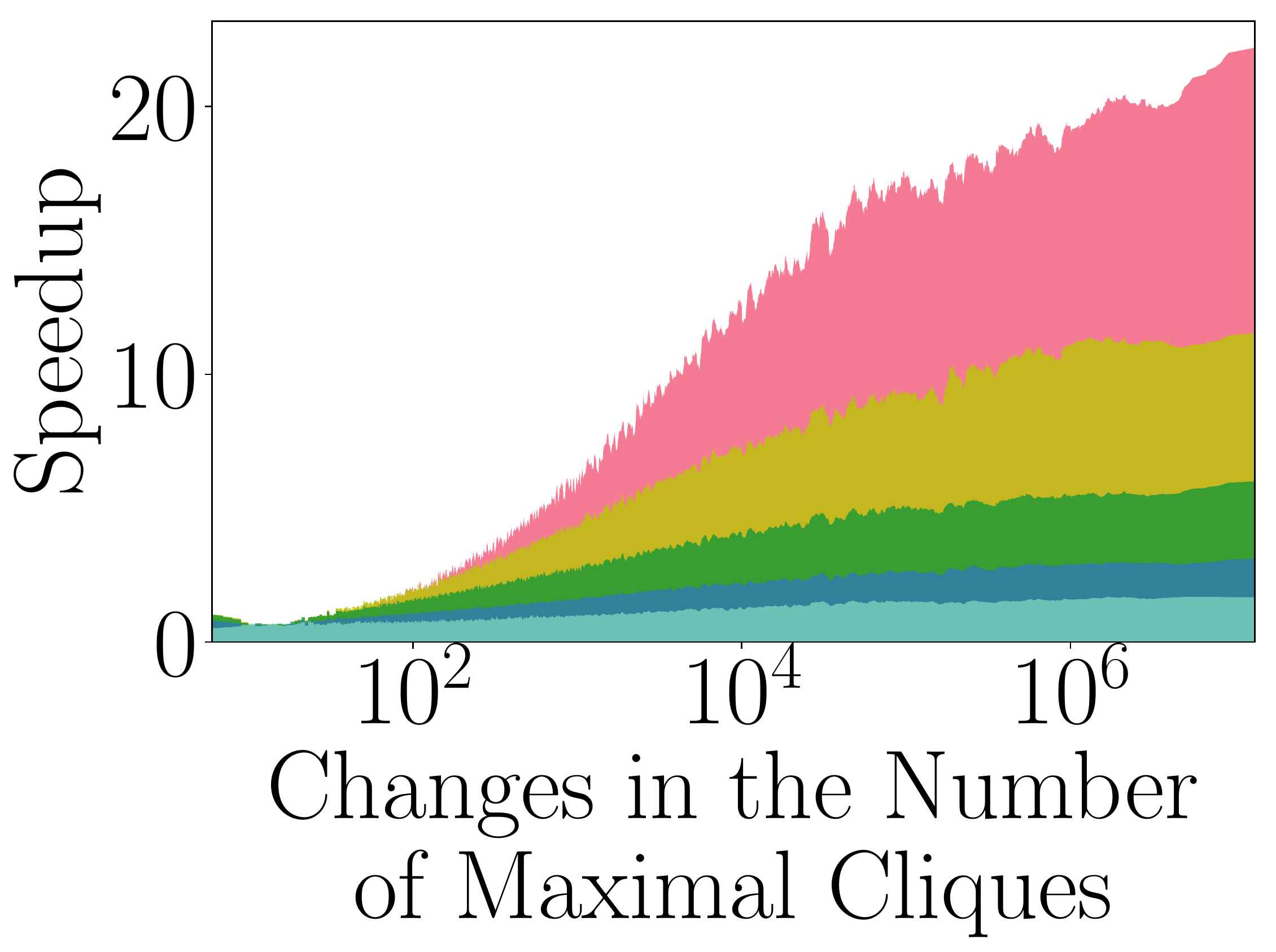} &
		\includegraphics[width=0.3\textwidth]{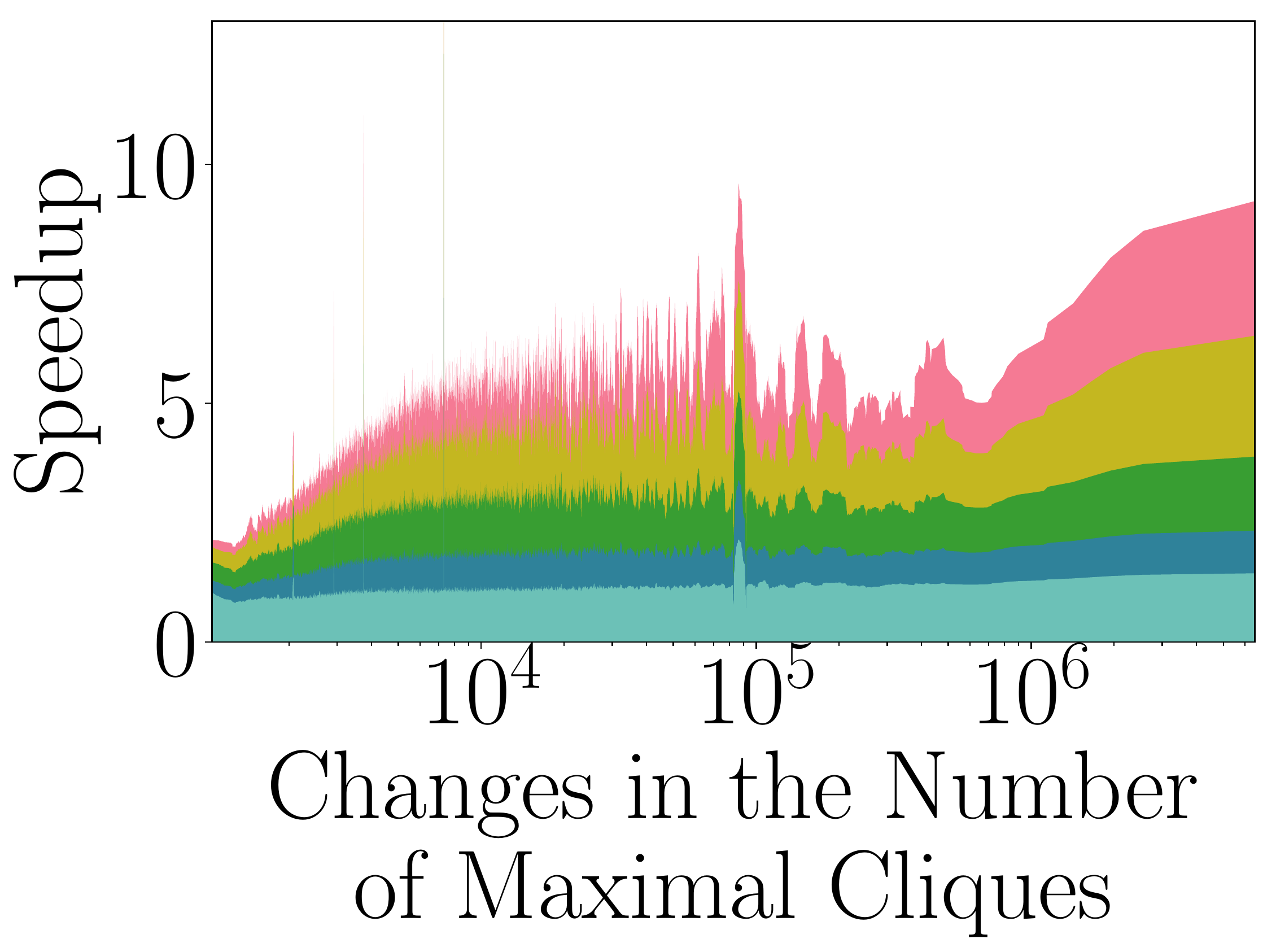} \\
		\textbf{(d) $\cacit$}  & \textbf{(e) $\wikipedia$}  \\[6pt]
	\end{tabular}
	\caption{Parallel speedup of $\parimce$ over $\imce$ as a function of the size of the change in the set of maximal cliques. 
		The size of the change is measured by the total number of new maximal cliques and subsumed maximal cliques when a batch of edges is added to the graph.}
	\label{fig:dyn-speedup}
	\vspace{2ex}
\end{figure}
\begin{figure}[htp!]
	\centering
	\begin{tabular}{cccc}
		\includegraphics[width=.3\textwidth]{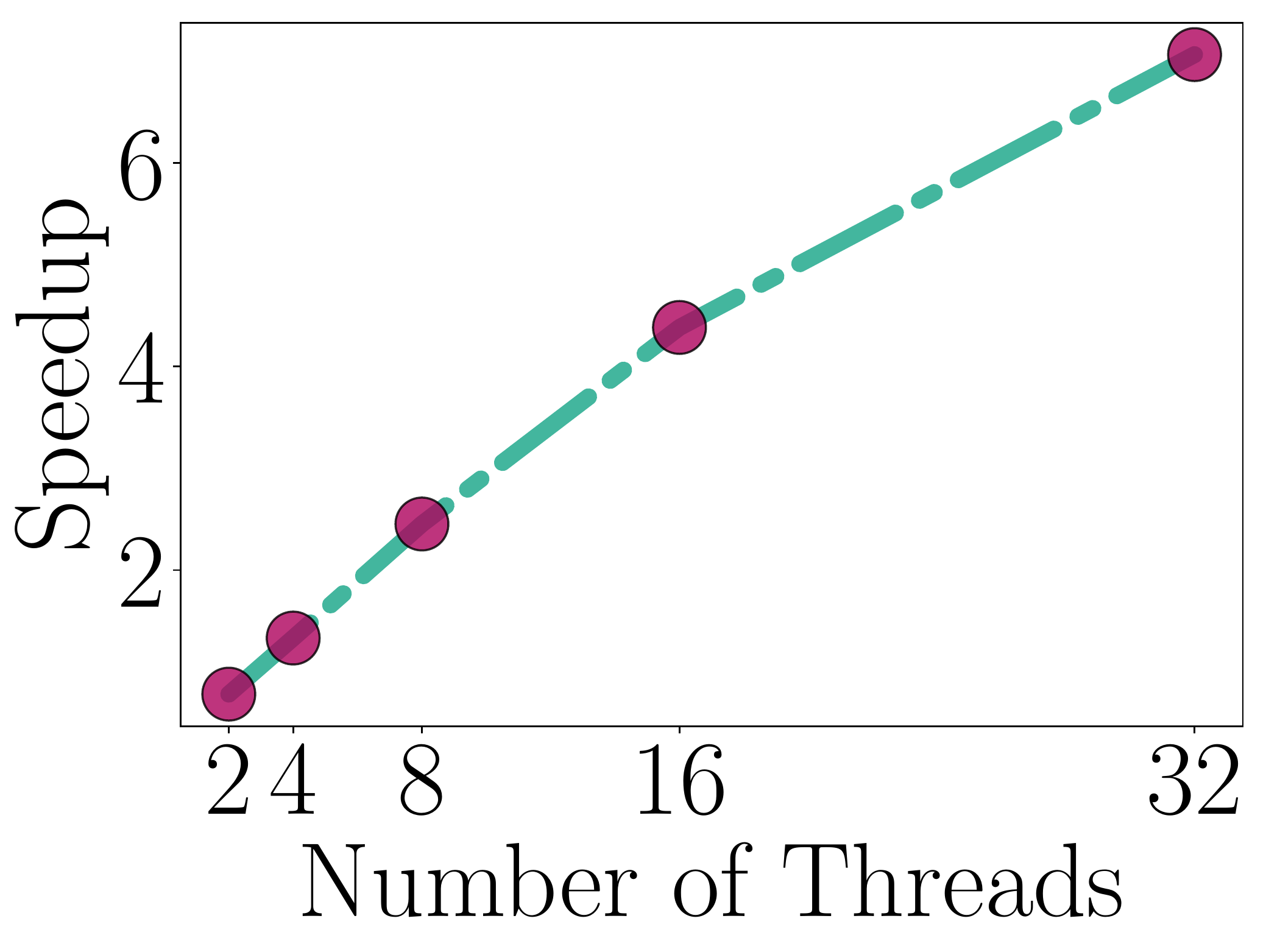} &
		\includegraphics[width=0.3\textwidth]{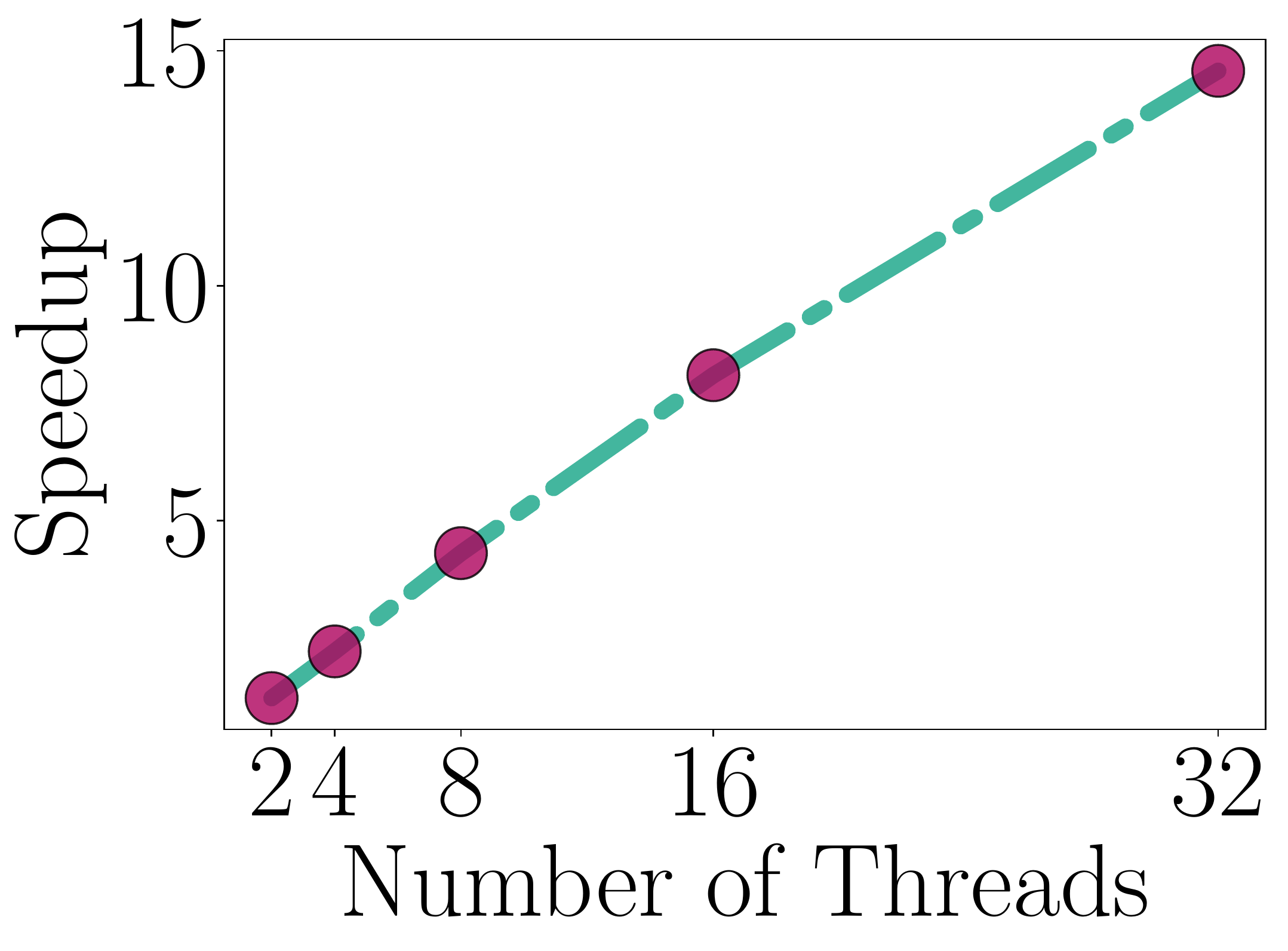} &
		\includegraphics[width=0.3\textwidth]{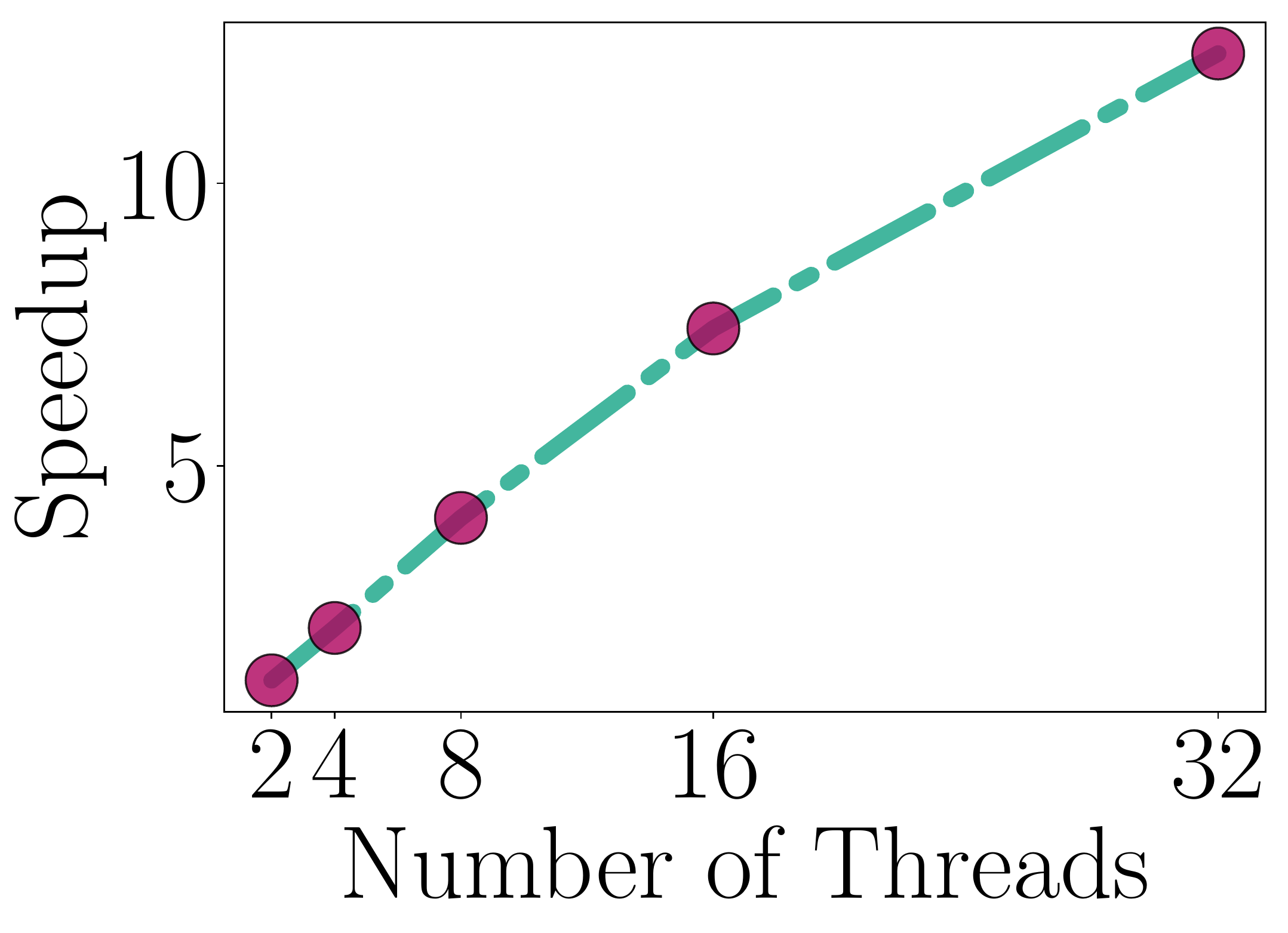} \\
		\textbf{(a) $\dblp$}  & \textbf{(b) $\flickr$} & \textbf{(c) $\journal$}  \\[6pt]
	\end{tabular}
	\begin{tabular}{cccc}
		\includegraphics[width=0.3\textwidth]{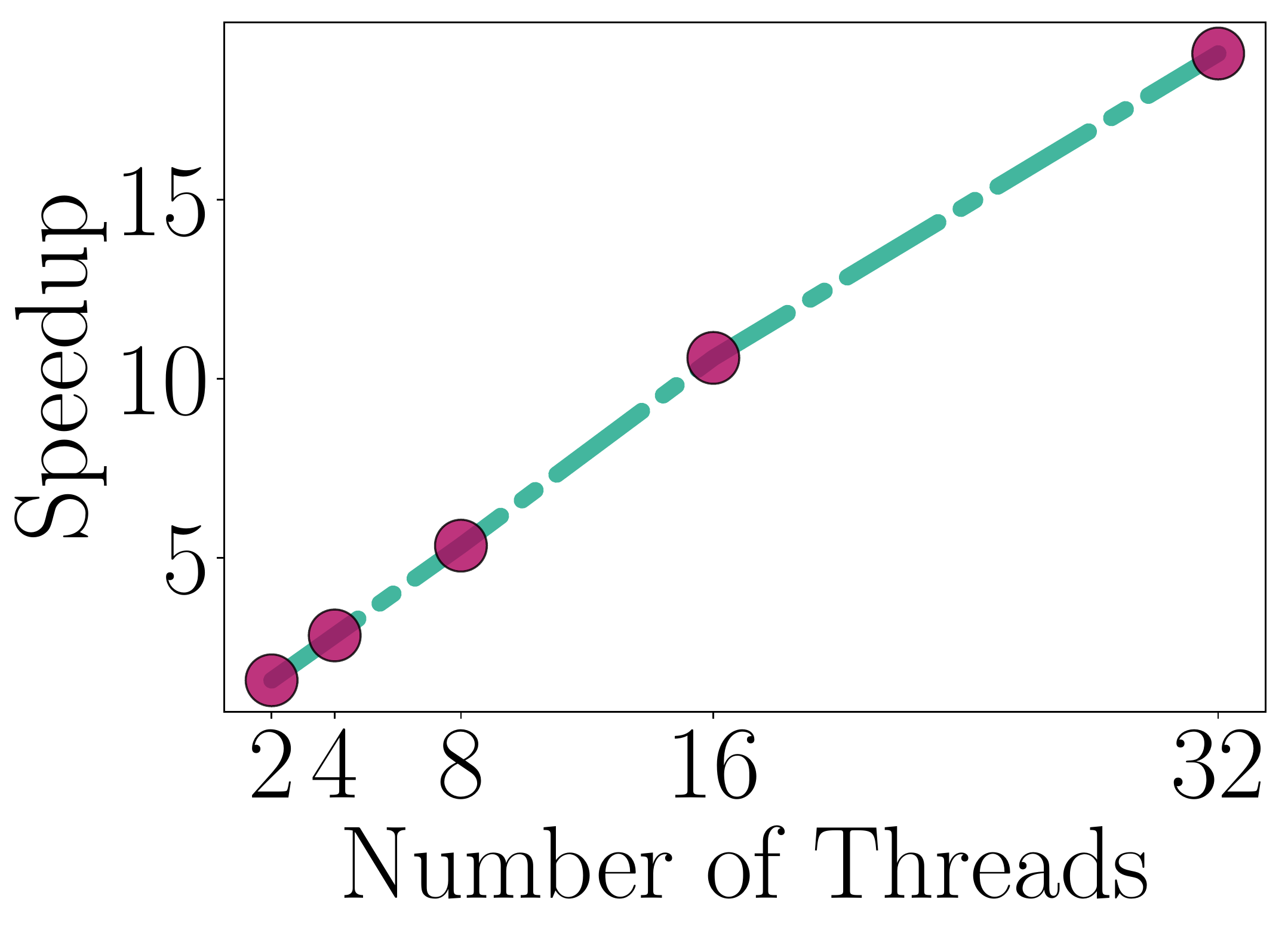} &
		\includegraphics[width=0.3\textwidth]{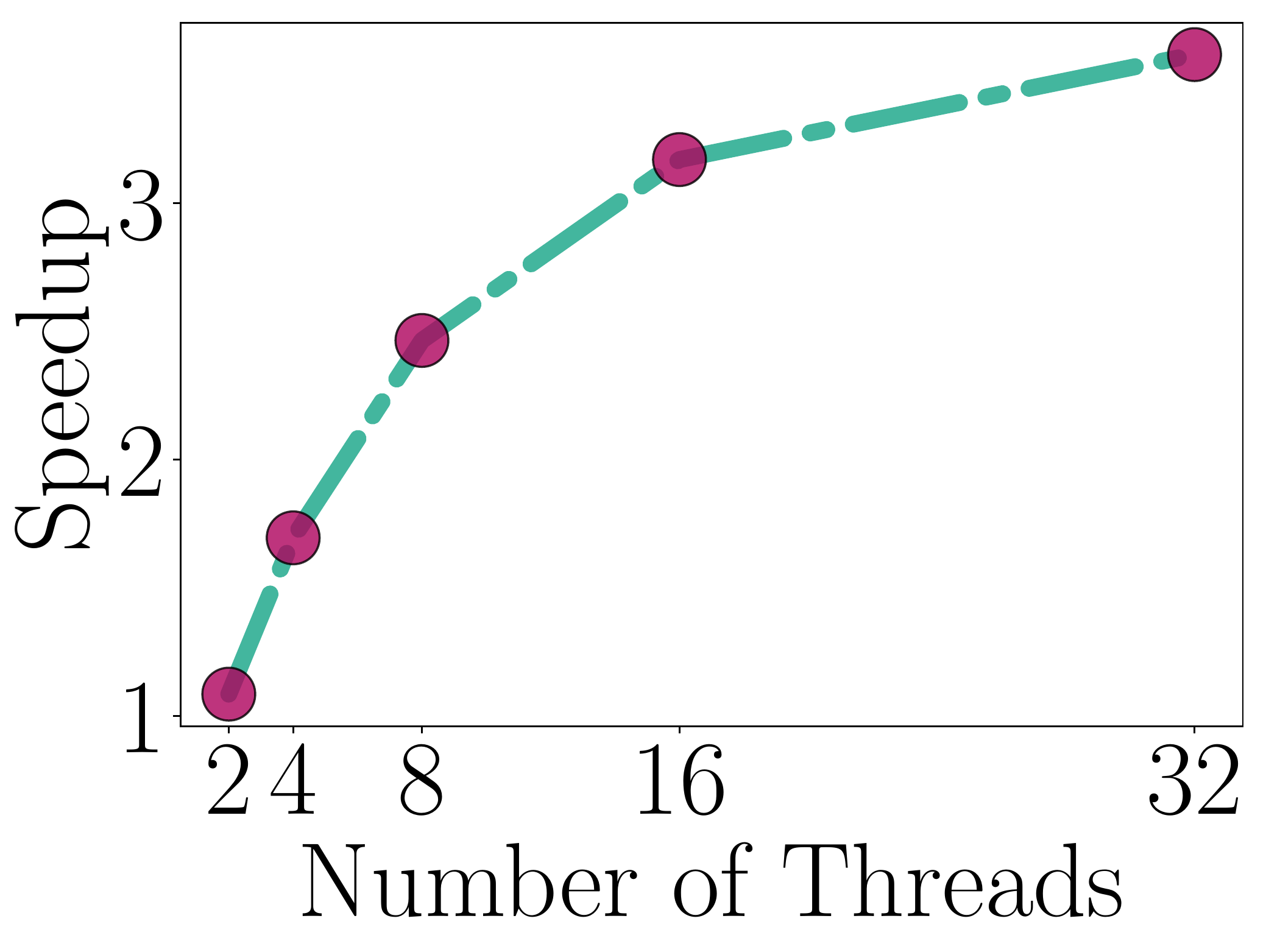} \\
		\textbf{(d) $\cacit$}  & \textbf{(e) $\wikipedia$}  \\[6pt]
	\end{tabular}
	\caption{Parallel speedup of $\parimce$ over $\imce$ as a function of number of threads, using the cumulative time of $\parimce$ and of $\imce$ for processing all batches of edges.}
	\label{fig:dyn-scalability}
\end{figure}
\\\\\\
\textbf{Impact of vertex ordering on overall performance of $\parmce$.} Here, we study the impact of different vertex ordering strategies, i.e. degree, degeneracy, and triangle count, on the overall performance of $\parmce$. \cref{result:runtime-splitup} presents the total computation time when we use different orderings. We observe that degree-based ordering ($\parmcedegree$) in most cases is the fastest strategy for clique enumeration, even when we do not consider the computation time for generating degeneracy- and triangle-based orderings. If we add up the runtime for the ordering step, {\em degree-based ordering is obviously better than degeneracy- or triangle-based orderings} since degree-based ordering is available for free when the input graph is read, while the degeneracy- and triangle-based orderings require additional computations.
\\\\
\textbf{Scaling up with the degree of parallelism.}
As the number of threads (and the degree of parallelism) increases, the runtime of $\parmce$ and of $\partomita$ decreases. \cref{fig:static-par-speedup} presents the speedup factor of $\parmce$ over the state-of-the-art sequential algorithm, i.e. $\tomita$, as a function of the number of threads, and  \cref{fig:static-par-runtime} presents the runtime of $\parmce$. $\parmcedegree$ achieves a speedup of more than \textbf{15x} on all graphs, when $32$ threads are used.


\begin{table}[b!]
	\caption{Total Runtime (in sec.) of $\parmce$ with different vertex orderings (using 32 threads). Total Runtime (TR) = Ranking Time (RT) + Enumeration Time (ET).}
	\label{result:runtime-splitup}
	\resizebox{0.8\textwidth}{!}{%
		\begin{tabular}{lccccccc}
			\multirow{2}{*}{\textbf{Dataset}} & \multicolumn{1}{l}{\multirow{2}{*}{\textbf{$\parmcedegree$}}} & \multicolumn{3}{c}{\textbf{$\parmcedegen$}} & \multicolumn{3}{c}{\textbf{$\parmcetriangle$}} \\   
			& \multicolumn{1}{l}{} & \multicolumn{1}{l}{\textbf{RT}} & \multicolumn{1}{l}{\textbf{ET}} & \multicolumn{1}{l}{\textbf{TR}} & \multicolumn{1}{l}{\textbf{RT}} & \multicolumn{1}{l}{\textbf{ET}} & \multicolumn{1}{l}{\textbf{TR}} \\ \hline
			$\dblp$ & 3 & 25 & 3 & 28 & 42 & 3 & 45 \\
			$\orkut$ & 1676 & 928 & 2350 & 3278 & 2166 & 1959 & 4125 \\
			$\skitter$ & 39 & 41 & 43 & 84 & 122 & 48 & 170 \\
			$\wtalk$ & 52 & 23 & 78 & 101 & 74 & 58 & 132 \\
			$\wikipedia$ & 123 & 244 & 155 & 399 & 950 & 179 & 1129
		\end{tabular}%
	}
\end{table}

\subsubsection{Parallel MCE on Dynamic Graphs}
The cumulative runtime of $\imce$ and $\parimce$ are presented in Table~\ref{tab:incr-par-time} which shows that the speedup achieved by $\parimce$ is \textbf{3.6x}-\textbf{19.1x} over $\imce$. This wide spectrum of speedups is mainly due to the variations in the size of the changes in the set of maximal cliques (number of new maximal cliques + number of subsumed maximal cliques) in the course of the incremental computation which can be observed in Figure~\ref{fig:dyn-speedup}. From this plot, we can see that the speedup increases with the increase in the size of the changes in the set of maximal cliques. This trend is as expected because the effect of parallelism will be prominent whenever the number of parallel tasks will become sufficiently large. This happens when the number of new and subsumed maximal cliques are large.
\begin{table}[]
	\caption{Cumulative runtime (in sec.) over the incremental computation across all edges, with $\imce$ and $\parimce$ using 32 threads. The total number of edges that are processed is also presented.}
	\label{tab:incr-par-time}
	\resizebox{0.8\textwidth}{!}{%
		\begin{tabular}{lcrrc}
			\textbf{Dataset} & \multicolumn{1}{l}{\textbf{\#Edges Processed}}  & $\imce$ & $\parimce$ & \multicolumn{1}{l}{\textbf{Parallel Speedup}} \\ \hline
			$\dblp$ & 5.1M & 6608 & 933 & \textbf{7x} \\
			$\flickr$ & 4.1M & 35238 & 2416 & \textbf{14.6x} \\
			$\wikipedia$ & 36.5M & 9402 & 2614 & \textbf{3.6x} \\
			$\journal$ & 19.2M & 30810 & 2497 & \textbf{12.3x} \\
			$\cacit$ & 93.8K & 33804 & 1767 & \textbf{19.1x}
		\end{tabular}%
	}
\end{table}
\\\\
\textbf{Scalability.}
The degree of parallelism increase with the increase in the number of threads. From Figure~\ref{fig:dyn-scalability} we can see that the speedup increases linearly with the number of threads. This behavior shows the scalability of our parallel algorithm $\parimce$. When the size of the changes will become large, scalability will become prominent because otherwise, most of the processors will remain idle when there will not be large amount of parallel tasks to fully utilize all the available processors. This is observed in $\wikipedia$ (Figure~\ref{fig:dyn-scalability}) where the cumulative size of change is relatively small. 
Additionally, the speedup observed on different input graphs vary with the similar size of change. For example, as we can see in Figure~\ref{fig:dyn-speedup}, $\journal$ gives more speedup than $\wikipedia$ with the size of change around $10^{5}$. To explain this, we first note that higher parallel speedup is observed when the cost of maintenance (which is proportional to the size of change in the set of maximal cliques) on a single batch is high. Following this, we see that in adding around $19$K batches of new edges (with batch size of $1000$ edges) starting from the empty graph, $\journal$ has more frequent (around $2000$) size of change in the order of $10^5$ compared to $\wikipedia$ (around $100$ size of changes in the order of $10^5$).

\subsection{Comparison with prior work \label{sec:exp-prior}}
We compare the performance of $\parmce$ with prior sequential and parallel algorithms for MCE. We consider the following sequential algorithms: $\greedybb$ due to Segundo et al.~\cite{SAS18}, $\tomita$ due to Tomita et al.~\cite{TTT06}, and $\bkd$ due to Eppstein et al.~\cite{ELS10}. For the comparison with parallel algorithm, we consider algorithm $\cliqueenum$ due to Zhang et al.~\cite{ZA+05}, $\peamc$ due to Du et al.~\cite{DW+09}, $\peco$ due to Svendsen et al.~\cite{SMT+15}, and the most recent parallel algorithm $\hashing$ due to Lessley et al.~\cite{LP+17}. The parallel algorithms $\cliqueenum$, $\peamc$, and $\hashing$ are designed for the shared-memory model, while $\peco$ is designed for distributed memory. We modified $\peco$ to work with shared-memory, by reusing the method for sub-problem construction and eliminating the need to communicate subgraphs by storing a single copy of the graph in shared-memory. We consider three different ordering strategies for $\peco$, which we call $\pecodegree, \pecodegen$, and $\pecotri$. The comparison of performance of $\parmce$ with $\peco$ is presented in Table~\ref{result:compare-with-peco}. We note that $\parmce$ is significantly better than that of $\peco$, no matter which ordering strategy was considered. \modified{We also compare the performance of $\parmce$ with a recent shared-nothing parallel algorithm $\gp$ to show that $\parmce$ outperforms $\gp$ with resources equivalent to our shared memory setting.}

\begin{table}[b]
	\caption{{Comparison of parallel runtime (in sec.) of $\parmce$, excluding the time for computing vertex ranking, with a version of $\peco$ that is modified to use shared-memory with 32 threads. Three different variants are considered for each algorithm based on the vertex ordering strategy.}}
	\label{result:compare-with-peco}
	\resizebox{\textwidth}{!}{%
		\begin{tabular}{lrrrrrr}
			\textbf{Dataset} & \multicolumn{1}{c}{$\pecodegree$} & \multicolumn{1}{c}{$\parmcedegree$} & \multicolumn{1}{c}{$\pecodegen$} & \multicolumn{1}{c}{$\parmcedegen$} & \multicolumn{1}{c}{$\pecotri$} & \multicolumn{1}{c}{$\parmcetriangle$} \\ \hline
			$\dblp$ & 6.4 & 2.6 & 6.9 & 3.1 & 6.8 & 2.9 \\
			$\orkut$ & 2050.7 & 1676.4 & 2183.4 & 2350 & 2361.9 & 1959.3 \\
			$\skitter$ & 261.5 & 39.2 & 331.8 & 42.8 & 260.9 & 48.2 \\
			$\wtalk$ & 1729.7 & 51.6 & 1728.2 & 77.8 & 1720 & 57.6 \\
			$\wikipedia$ & 8982.5 & 123.3 & 9110.4 & 155.3 & 8938 & 178.8
		\end{tabular}%
	}
\end{table}
\begin{table}[t!]
	\caption{Comparison of runtimes (in sec.) of $\parmce$ with prior works on shared-memory algorithms for MCE (with 32 threads).}
	\label{result:compare-with-hashing}
	\resizebox{\textwidth}{!}{%
		\begin{tabular}{lclll}
			\textbf{Dataset} & \multicolumn{1}{c}{$\parmcedegree$} & \multicolumn{1}{c}{$\hashing$} & \multicolumn{1}{c}{$\cliqueenum$} & \multicolumn{1}{c}{$\peamc$} \\ \hline
			$\dblp$ & 2.6 & Out of memory in 3 min. & Out of memory in 10 min. & Not complete in 5 hours. \\
			$\orkut$ & 1676.4 & Out of memory in 7 min. & Out of memory in 20 min. & Not complete in 5 hours. \\
			$\skitter$ & 39.2 & Out of memory in 5 min. & Out of memory in 10 min. & Not complete in 5 hours. \\
			$\wtalk$ & 51.6 & Out of memory in 10 min. & Out of memory in 20 min. & Not complete in 5 hours. \\
			$\wikipedia$ & 123.3 & Out of memory in 10 min. & Out of memory in 20 min. & Not complete in 5 hours.
		\end{tabular}%
	}
\end{table}

\begin{table}[]
	\caption{
			Speedup factor of $\parmcedegree$ over $\gp$ and $\pecodegree$. A speedup factor greater than $1$ indicates that our algorithm $\parmcedegree$ is faster than the prior methods. 8$^\text{\ding{72}}$ indicates that 8 threads for $\parmcedegree$ and 8 MPI nodes for $\gp$ are used ($^\text{\ding{72}}$ next to other numbers is interpreted similarly). \ding{56} indicates that $\gp$ runs out of memory. \label{result:compare-with-gp-peco}
	}
	\centering
	\resizebox{\textwidth}{!}{%
		\begin{tabular}{l|ccccc||ccccc|}
			\cline{2-11}
			& \multicolumn{5}{l||}{Speedup factor of $\parmcedegree$ over $\gp$} & \multicolumn{5}{c|}{Speedup factor of $\parmcedegree$ over $\pecodegree$} \\ \hline
			\multicolumn{1}{|l|}{\textbf{Dataset}}       & 2$^\text{\ding{72}}$         & 4$^\text{\ding{72}}$         & 8$^\text{\ding{72}}$         & 16$^\text{\ding{72}}$       & 32$^\text{\ding{72}}$       & 2 threads   & 4 threads   & 8 threads   & 16 threads   & 32 threads  \\ \hline\hline
			\multicolumn{1}{|l|}{$\dblp$} & 0.83       & 1.21       & 1.76       & 2.9       & 4.1       & 2.5         & 2.85        & 2.4         & 2.33         & 2.5         \\
			\multicolumn{1}{|l|}{$\orkut$}         & 1.92       & 1.68       & 1.73       & \ding{56}         & \ding{56}         & 1.2         & 1.11        & 1.04        & 0.85         & 1.48        \\
			\multicolumn{1}{|l|}{$\skitter$}    & 1.66       & 1.64       & 1.56       & 1.71      & 1.35      & 4.47        & 4.33        & 3.92        & 4.1          & 3.26        \\
			\multicolumn{1}{|l|}{$\wtalk$}     & 1.38       & 1.32       & 1.15       & 1.3       & 0.84      & 2.48        & 2.37        & 4.65        & 9.07         & 10.9        \\
			\multicolumn{1}{|l|}{$\wikipedia$}     & 1.5        & 1.45       & 1.44       & 1.86      & 1.69      & 9.6         & 9.74        & 14.07       & 25.34        & 29.6        \\ \hline
		\end{tabular}%
	}
\end{table}
The comparison of $\parmce$ with other shared-memory algorithms $\peamc$, $\cliqueenum$, and $\hashing$ is shown in Table~\ref{result:compare-with-hashing}. The performance of $\parmce$ is seen to be much better than that of any of these prior shared-memory parallel algorithms. For the graph $\dblp$, $\peamc$ did not finish within $5$ hours whereas $\parmce$ takes at most around $50$ secs for enumerating $1.2$ million maximal cliques. The poor running time of  $\peamc$ is due to two following reasons: (1)~The algorithm does not apply efficient pruning techniques such as pivoting, used in $\tomita$, and (2)~The method to determine the maximality of a clique in the search space is not efficient. The $\cliqueenum$ algorithm runs out of memory after a few minutes. The reason is that $\cliqueenum$ maintains a bit vector for each vertex that is as large as the size of the input graph and, additionally, needs to store intermediate non-maximal cliques. For each such non-maximal clique, it is required to maintain a bit vector of length equal to the size of the vertex set of the original graph. Therefore, in $\cliqueenum$, a memory issue is inevitable for a graph with millions of vertices.

A recent parallel algorithm in the literature, $\hashing$, also has a significant memory overhead and ran out of memory on the input graphs that we considered. The reason for its high memory requirement is that $\hashing$ enumerates intermediate non-maximal cliques before finally outputting maximal cliques. The number of such intermediate non-maximal cliques may be very large, even for graphs with few number of maximal cliques. For example, a maximal clique of size $c$ contains $2^c-1$ non-maximal cliques.
\begin{table}[b!]
\caption{Total runtime (sec.) of parallel algorithm $\parmce$ (with different vertex ranking, with 32 threads) and sequential algorithms $\bkd$ and $\greedybb$.}
\label{results:compare-with-sequential}
	\resizebox{\textwidth}{!}{%
		\begin{tabular}{lccccc}
			\textbf{Dataset}  & $\bkd$ & \multicolumn{1}{c}{$\greedybb$} & $\parmcedegree$ & $\parmcedegen$ & $\parmcetriangle$
			\\\hline
			$\dblp$ & 53.6 & \multicolumn{1}{c}{Not finish in 30 min.} & 2.6 & 28.1 & 44.3 \\
			$\orkut$ & 29812.3 & Out of memory in 5 min. & 1676.4 & 3278 & 4125.3 \\
			$\skitter$ & 641.7 & Out of memory in 10 min. & 39.2 & 83.8 & 170.2 \\
			$\wtalk$ & 1003.2 & Out of memory in 10 min. & 51.6 & 100.8 & 131.2 \\
			$\wikipedia$ & 2243.6 & Out of memory in 10 min. & 123.3 & 399 & 1128
		\end{tabular}%
	}
\end{table}

Next, we compare the performance of $\parmce$ with that of sequential algorithms $\bkd$ and a recent sequential algorithm $\greedybb$ -- results are in Table~\ref{results:compare-with-sequential}. For large graphs, the performance of $\bkd$ is almost similar to $\tomita$ whereas  $\greedybb$ performs much worse than $\tomita$. Since our $\parmce$ algorithm outperforms $\tomita$, we can conclude that $\parmce$ is significantly faster than other sequential algorithms.

\modified{Next, we compare with $\gp$~\cite{WC+17}, a recent distributed algorithm for MCE based on MPI}. This method first assigns each vertex $v$ to an MPI worker, and then the worker constructs subproblems of vertex $v$ by constructing the sets $\cand$ and $\fini$ for enumerating the maximal cliques which vertex $v$ is part of. Through an iterative computation, each sub-problem is broken down into smaller sub-problems, till a maximal clique is obtained. Subproblems within an MPI worker might be sent to another worker, where the receiver MPI process is chosen randomly. We compare the runtime performance of our work $\parmcedegree$ with of $\gp$. Note that $\parmcedegree$ is a shared-memory method while $\gp$ is implemented in a distributed memory model. 

We ran $\gp$ and $\parmcedegree$ on a machine configured with two $18$-core Intel Skylake $6140$ Xeon processors. In every experiment, we use $192$GB as the total memory. For a fair comparison, we use the same number of threads for $\parmcedegree$ as the number of MPI processes for $\gp$. In \cref{result:compare-with-gp-peco}, we report the speedup of $\parmcedegree$ over $\gp$. In most cases $\parmcedegree$ outperforms $\gp$. 
In $\dblp$, we observed that $\gp$ cannot achieve a better runtime while the number of MPI workers increases. For a better understanding of this case, we measured the runtime performance of exchanged sub-problems among MPI nodes and the runtime of clique enumeration. We observed that the enumeration takes at most one second while the overhead for exchanging sub-problems among workers is huge and skewed towards a few MPI nodes, which leads to an unbalanced workload. As shown in \cref{result:compare-with-gp-peco}, there are two cases that $\gp$ performs faster than $\parmcedegree$ -- in both cases, the difference in runtimes is small, about 8 seconds.


\subsection{Summary of Experimental Results}
We found that both $\partomita$ and $\parmce$ yield significant speedups over the sequential algorithm $\tomita$nearly linear in the number of cores available. $\parmce$ using the degree-based vertex ranking always performs better than $\partomita$. The runtime of $\parmce$ using degeneracy/triangle count based vertex ranking is sometimes worse than $\partomita$ due to the overhead of sequential computation of vertex ranking -- note that this overhead is not needed in $\partomita$. The parallel speedup of $\parmce$ is better when the input graph has many large sized maximal cliques. Overall, $\parmce$ consistently outperforms prior sequential and parallel algorithms for MCE. For a dynamic graph we found that $\parimce$ consistently yields a substantial speedup over the efficient sequential algorithm $\imce$. Further, the speedup of $\parimce$ improves as the size of the change (to be enumerated) becomes larger.

	\section{Conclusion}
	\label{sec:conclude}
We presented shared-memory parallel algorithms for enumerating maximal cliques from a graph. $\partomita$ is a work-efficient parallelization of a sequential algorithm due to Tomita et al.~\cite{TTT06},  and $\parmce$ is an adaptation of $\partomita$ that has more opportunities for parallelization and better load balancing. Our algorithms obtain significant improvements compared with the current state-of-the-art on MCE. Our experiments show that $\parmce$ has a speedup of up to 21x (on a 32 core machine) when compared with an efficient sequential baseline. In contrast, prior shared-memory parallel methods for MCE were either unable to process the same graphs in a reasonable time or ran out of memory. We also presented a shared-memory parallel algorithm $\parimce$  that can enumerate the change in the set of maximal cliques when new edges are added to the graph.

Many questions remain open: (1)~Can these methods scale to even larger graphs and to machines with larger numbers of cores (2)~How can one adapt these methods to other parallel systems such as a cluster of computers with a combination of shared and distributed memory or GPUs?\\
	
	\noindent\textbf{Acknowledgment:} The work of AD, SS, and ST is supported in part by the US National Science Foundation through grants 1527541 and 1725702.

	
	\bibliographystyle{ACM-Reference-Format}
	\bibliography{cliques}
	
	\pagebreak
	
	\appendix
	
\section{Proof of Work Efficiency of parallel MCE on Dynamic Graph}
\label{appendix-a}

Here, we show the work efficiency of $\parimce$ by proving Lemma~\ref{lem:incremental-new-work-efficiency}.

\begin{algorithm}[htp!]
\DontPrintSemicolon
\caption{$\tomitaE({G},K,\cand,\fini, \mathcal{E})$}
\label{algo:tomitaE}
\KwIn{${G}$ - The input graph \\ $K$ - a non-maximal clique to extend \\ 
$\cand$ - Set of vertices that may extend $K$ \\ 
$\fini$ - vertices that have been used to extend $K$ \\ 
$\mathcal{E}$ - set of edges to ignore}
\If{$(\cand = \emptyset)$ \& $(\fini = \emptyset)$}{
	Output $K$ and \Return\;
}
$\pivot \gets (u \in \cand \cup \fini)$ such that $u$ maximizes the size of the intersection of $\cand  \cap \Gamma_{{G}}(u)$\;
$\ext \gets \cand - \Gamma_{{G}}(\pivot)$\;
\For{$q \in$ \ext}{
	$K_q \gets K\cup\{q\}$\;
	\If{$K_q \cap \mathcal{E} \neq \emptyset$}{
		$\cand \gets \cand - \{q\}$\;
		$\fini \gets \fini\cup\{q\}$\;
		\textbf{continue}
	}
	$\cand_q \gets \cand\cap\Gamma_{{G}}(q)$\;
	$\fini_q \gets \fini\cap\Gamma_{{G}}(q)$\;
	$\tomitaE({G},K_q,\cand_q,\fini_q,\mathcal{E})$\;
	$\cand \gets \cand-\{q\}$\;
	$\fini \gets \fini \cup\{q\}$\;
}
\end{algorithm}

Lemma~\ref{lem:incremental-new-work-efficiency}:~Given a graph $G$ and a new edge set $H$, $\parcsnewttt$ is work-efficient, i.e, the total work is of the same order of the time complexity of $\csnewttt$. The depth of $\parcsnewttt$ is $O(\Delta^2 + M\log{\Delta})$ where $\Delta$ is the maximum degree and  $M$ is the size of a maximum clique in $G+H$.

\begin{proof}
	First, we prove the work efficiency of $\parcsnewttt$ followed by the depth of the algorithm. Note that, for proving the work-efficiency we will show that procedure at each line from Line~4 to Line~10 of $\parcsnewttt$ is work-efficient. \modified{Lines~4, 6, 7, and 8} are work-efficient because, all these procedures are sequential in $\parcsnewttt$. The parallel set operations at Line~5 and Line~10 are work-efficient using Theorem~\ref{thm:parset}. \modified{Now we will show the work efficiency of $\partomitaE$ as follows}. If we disregard Lines~7-10 of $\tomitaE$ and Lines~9-10 of $\partomitaE$ then the total work of $\partomitaE$ is the same as the time complexity of $\partomita$ following the work efficiency of $\partomita$. Next, we say that the time complexity of Lines~7-10 of $\tomitaE$ is the same as the time complexity of Lines~9-10 of $\partomitaE$ because in $\tomitaE$ we use two global hashtables - one for maintaining the adjacent vertices of the currently processing vertex in the set of new edges and another for maintaining the indexes of the new edges that we define before the beginning of the enumeration of new maximal cliques. With these two hashtables, we can check the \textit{if} condition at Line~9 of $\partomitaE$ in parallel with total work $O(n)$ using Theorem~\ref{thm:parset} which is of the same order of the time complexity of performing \textit{if} condition check at Line~7 of $\tomitaE$. This completes the proof of work efficiency of $\partomitaE$.
For proving the depth of $\parcsnewttt$, note that the depth is the sum of the depths of procedures at Line~5, 6, 9, 10 of $\parcsnewttt$ because the cost of all operations in other lines are $O(1)$ each. The depth of executing intersection in parallel at Line~5 is $O(1)$ using Theorem~\ref{thm:parset}, the depth of the procedure for constructing the graph at Line~6 is $O(\Delta^2)$ as we construct the graph sequentially, the depth $\partomitaE$ is $O(M\log{\Delta})$ following the depth of $\tomitaE$, and the depth of Line~10 is $O(1)$ because we can do this operation in parallel using Theorem~\ref{thm:parset}. Thus, the overall depth of $\parcsnewttt$ follows.
\end{proof}

\end{document}